\documentclass{lmcs}
\pdfoutput=1

\usepackage{lastpage}
\lmcsdoi{15}{3}{20}
\lmcsheading{}{\pageref{LastPage}}{}{}%
{Jan.~10,~2019}{Aug.~23,~2019}{}

\usepackage{xspace}
\usepackage{nicefrac}
\usepackage{algorithmic, algorithm}
\usepackage{mathtools,amsmath}
\usepackage{cite}
\usepackage[utf8]{inputenc}
\usepackage{hyperref}

\theoremstyle{plain}
\newtheorem{numrem}[thm]{Remark}

\usepackage{tikz}
\usetikzlibrary{decorations.pathreplacing}
\usetikzlibrary{decorations.pathmorphing}
\usetikzlibrary{decorations.markings}
\usetikzlibrary{shapes,shapes.symbols,automata,arrows}
\usetikzlibrary{calc}
\usetikzlibrary{patterns}
\usetikzlibrary{positioning}
\usetikzlibrary{fit}

\tikzset{p0/.style = {ellipse, draw, thick, minimum height = 0.7cm}}
\tikzset{p1/.style = {rectangle, minimum size=.7cm, draw, thick}}
\tikzset{>=stealth, shorten >=1pt}
\tikzset{every edge/.style = {thick, ->, draw}}
\tikzset{every loop/.style = {thick, ->, draw}}
\tikzset{
	weight/.style args={#1at#2anchor#3}{%
		postaction={decorate},%
		decoration={markings,mark=at position #2 with {\node[anchor=#3] {#1};}}
	}
}
\tikzset{
	elided/.style = {%
		postaction={decorate},%
		decoration={markings,mark=at position #1 with {\node[fill=white,transform shape] {$\cdots$};}}
	},
	elided/.default = .5
}
\tikzset{
	brace/.style args = {#1amplitude#2}{%
		draw,
		decorate,
		decoration={brace,amplitude=#2,mirror},
		postaction={decorate,decoration={markings,mark=at position .5 with {\node[anchor=north] {#1};}}},
	}
}

\pgfdeclarelayer{background}
\pgfsetlayers{background,main}

\newcommand{\parnode}[2]{$\nicefrac{#1}{#2}$}

\makeatletter
\tikzset{circle split part fill/.style  args={#1,#2}{%
 alias=tmp@name, 
  postaction={%
    insert path={
     \pgfextra{%
     \pgfpointdiff{\pgfpointanchor{\pgf@node@name}{center}}%
                  {\pgfpointanchor{\pgf@node@name}{east}}%
     \pgfmathsetmacro\insiderad{\pgf@x}
      \fill[#1] (\pgf@node@name.base) ([xshift=-\pgflinewidth]\pgf@node@name.east) arc
                          (0:180:\insiderad-\pgflinewidth)--cycle;
      \fill[#2] (\pgf@node@name.base) ([xshift=\pgflinewidth]\pgf@node@name.west)  arc
                           (180:360:\insiderad-\pgflinewidth)--cycle;
         }}}}}
 \makeatother

\usepackage{booktabs}

\usepackage{pifont}

\definecolor{myred}{rgb}{1,0.604,0.604}
\definecolor{mydarkred}{rgb}{1,0.345,0.345}
\definecolor{myblue}{rgb}{0.635,0.675,0.966}
\definecolor{mydarkblue}{rgb}{0.412,0.475,0.957}
\definecolor{myyellow}{rgb}{1,0.976,0.604}
\definecolor{mydarkyellow}{rgb}{1,0.961,0.345}
\definecolor{mygreen}{rgb}{0.604,1,0.604}
\definecolor{mydarkgreen}{rgb}{0.345,1,0.345}

\tikzset{
	assign/.style = { fill=myblue },
	choice/.style = { fill=myred },
	check/.style  = { fill=myyellow }
}

\makeatletter
\DeclareRobustCommand{\rvdots}{%
  \vbox{
    \baselineskip4\p@\lineskiplimit\z@
    \kern-\p@
    \hbox{.}\hbox{.}\hbox{.}
  }}
\makeatother

\newcommand{\myquot}[1]{``#1''}

\newcommand{\bigo}[0]{\mathcal{O}}

\newcommand{\sign}{\mathrm{Sgn}}
\newcommand{\abs}[1]{|#1|}
\newcommand{\size}[1]{|#1|}
\newcommand{\card}[1]{\size{#1}}
\newcommand{\set}[1]{\{ #1 \}}
\newcommand{\compl}[1]{\overline{#1}}
\newcommand{\nats}{\mathbb{N}}
\newcommand{\ints}{\mathbb{Z}}

\renewcommand{\epsilon}{\varepsilon}

\newcommand{\arena}{\mathcal{A}}
\newcommand{\game}{\mathcal{G}}
\newcommand{\thrarena}{\mathcal{A}'}
\newcommand{\thrgame}{\mathcal{G}_b}
\newcommand{\weight}{w}
\newcommand{\Cost}{\mathrm{Cost}}

\newcommand{\col}{\Omega}

\newcommand{\wincond}{\mathrm{Win}}
\newcommand{\winreg}{\mathcal{W}}

\newcommand{\parity}{\mathrm{Parity}}

\newcommand{\cp}{\mathrm{WeightParity}}
\newcommand{\bcp}{\mathrm{BndWeightParity}}

\newcommand{\energyparity}{\mathrm{EnergyParity}}

\newcommand{\mpp}{\mathrm{MeanPayoffParity}\xspace}
\newcommand{\countdown}{\text{Countdown}\xspace}

\newcommand{\mem}{\mathcal{M}}
\newcommand{\init}{\text{init}}
\newcommand{\update}{\mathrm{upd}}
\newcommand{\nxt}{\mathrm{Nxt}}

\newcommand{\ampl}{\mathrm{Ampl}}

\newcommand{\Cor}{\mathrm{Cor}}

\newcommand{\answer}[1]{\mathrm{Ans}({#1})} 

\newcommand{\att}[3]{\mathrm{Attr}_{#1}^{#2}(#3)} 

\newcommand{\np}{\textsc{NP}}
\newcommand{\conp}{\textsc{co-NP}}
\newcommand{\exptime}{\textsc{ExpTime}}

\newcommand{\ptime}{\textsc{PTime}}

\newcommand{\logspace}{\textsc{LogSpace}\xspace}
\newcommand{\pspace}{\textsc{PSpace}\xspace}

\newcommand{\epi}{\text{EPI}\xspace}
\newcommand{\epis}{\text{EPIs}\xspace}
\newcommand{\mrr}{\text{posMRR}\xspace}
\newcommand{\esi}{\text{ESI}\xspace}
\newcommand{\esis}{\text{ESIs}\xspace}

\newcommand{\req}{\text{req}}
\newcommand{\del}{\text{del}}
\newcommand{\ans}{\text{ans}}

\newcommand{\I}{\mathit{I}}
\newcommand{\II}{\mathit{II}}
\newcommand{\III}{\mathit{III}}
\newcommand{\IV}{\mathit{IV}}

\newcommand{\unpol}{\text{unpol}}

\newcommand{\reach}{\mathcal{R}}

\newcommand{\suffix}{\mathrm{sfx_\top}}


\title{Parity Games with Weights}

\author[S.~Schewe]{Sven Schewe\rsuper{a}}
\address{\lsuper{a}University of Liverpool, Liverpool, United Kingdom}
\email{\{sven.schewe,martin.zimmermann\}@liverpool.ac.uk}
\thanks{Supported by the EPSRC projects `Energy Efficient Control' (EP/M027287/1) and `Solving Parity Games in Theory and Practice' (EP/P020909/1).}

\author[A.~Weinert]{Alexander Weinert\rsuper{b}}
\address{\lsuper{b}German Aerospace Center (DLR), Intelligent and Distributed Systems, Köln, Germany}
\email{alexander.weinert@dlr.de}
\thanks{Supported by the project ``TriCS'' (ZI 1516/1--1) of the German Research Foundation (DFG) and the Saarbrücken Graduate School of Computer Science. Major parts of this work were performed while the author was employed with the Reactive Systems Group at Saarland University, Germany.}

\author[M.~Zimmermann]{Martin Zimmermann\rsuper{a}}

\begin{document}

\begin{abstract}
Quantitative extensions of parity games have recently attracted significant interest.
These extensions include parity games with energy and payoff conditions as well as finitary parity games and their generalization to parity games with costs.
Finitary parity games enjoy a special status among these extensions, as they offer a native combination of the qualitative and quantitative aspects in infinite games:
The quantitative aspect of finitary parity games is a quality measure for the qualitative aspect, as it measures the limit superior of the time it takes to answer an odd color by a larger even one.
Finitary parity games have been extended to parity games with costs, where each transition is labeled with a nonnegative weight that reflects the costs incurred by taking it.
We lift this restriction and consider parity games with costs with arbitrary integer weights.

We show that solving such games is in $\np \cap \conp$, the signature complexity for games of this type.
We also show that the protagonist has finite-state winning strategies, and provide tight pseudo-polynomial bounds for the memory he needs to win the game.
Naturally, the antagonist may need infinite memory to win.
Moreover, we present tight bounds on the quality of winning strategies for the protagonist.

Furthermore, we investigate the problem of determining, for a given threshold~$b$, whether the protagonist has a strategy of quality  at most~$b$ and show this problem to be \exptime-complete.
The protagonist inherits the necessity of exponential memory for implementing such strategies from the special case of finitary parity games.
\end{abstract}


\maketitle

\section{Introduction}%
\label{sec:introduction}
Finite games of infinite duration offer a wealth of challenges and applications that has garnered a lot of attention.
The traditional class of games under consideration were games with a simple parity~\cite{Kozen/83/mu,Emerson+Lei/86/Parity,Emerson+Jutla/91/Memoryless,McNaughton/93/Games,Browne-all/97/fixedpoint,Zielonka/98/Parity,Jurdzinski/98/UP,Jurdzinski/00/ParityGames,DBLP:conf/cav/VogeJ00,Jurdzinski/08/subex,Schewe/08/improvement,ScheweTV15/symmetric,Schewe/17/parity,CaludeJKLS/17/qp,JL17,FearnleyJS0W/17/QP,Lehtinen/17/qp} or payoff~\cite{Puri/95/simprove,Zwick+Paterson/96/payoff,Jurdzinski/98/UP,BjorklundVorobyov/07/subexp,ScheweTV15/symmetric} objective.
These games form a hierarchy with very simple tractable reductions from parity games through mean-payoff games~\cite{Puri/95/simprove,Zwick+Paterson/96/payoff,Jurdzinski/98/UP,BjorklundVorobyov/07/subexp,ScheweTV15/symmetric} and discounted-payoff games~\cite{Zwick+Paterson/96/payoff,Jurdzinski/98/UP,ScheweTV15/symmetric} to simple stochastic games~\cite{Condon/93/onalgorithms}.

More recently, games with a mixture of the qualitative parity condition and further quantitative objectives have been considered, including mean-payoff parity games~\cite{ChatterjeeHenzingerJurdzinski05} and energy parity games~\cite{ChatterjeeDoyen12}.
Finitary parity games~\cite{ChatterjeeHenzingerHorn09} take a special role within the class of games with mixed parity and payoff objectives.
To win a finitary parity game, Player~$0$ needs to enforce a play with a bound~$b$ such that almost all occurrences of an odd color are followed by a higher even color within at most~$b$ steps.

This is interesting, because it provides a natural link between the qualitative and quantitative objective.
One aspect that attracted attention is that, as long as one is not interested in optimizing the bound~$b$, these games are the only games of the lot that are known to be tractable~\cite{ChatterjeeHenzingerHorn09}.
However, minimizing the bound~$b$ itself is also interesting:
As the bounds limit the response time, small bounds translate to high quality guarantees~\cite{WeinertZimmermann17}.

This property calls for a generalization to different cost models, and a first generalization has been made with the introduction of parity games with costs~\cite{FijalkowZimmermann14}.
In parity games with costs, the basic cost function of finitary parity games---where each step incurs the same cost---is replaced with different nonnegative costs for different edges.
In this article, we generalize this further to \emph{general} integer costs: We decorate the edges with integer weights.
The quantitative aspect in these parity games with weights consists of having to answer almost all odd colors by a higher even color, such that the
amplitude of the accumulated weight on the path to this even color is bounded by a bound~$b$.

In addition to their conceptual charm, we show that parity games with weights are \ptime\ equivalent to energy parity games.
This indicates that these games are part of a natural complexity class, whereas the games with a plain objective appear to form a hierarchy.
We use the reduction from parity games with weights to energy parity games to solve them.
This reduction goes through intermediate reductions to and from \emph{bounded} parity games with weights.
These games have the additional restriction that the limit superior of the absolute weight of initial sequences of unanswered requests in a play is finite.
These bounded parity games with weights are then reduced to energy parity games.
The other direction of the reduction is through simple gadgets that preserve the main elements of winning strategies in games that are extended in two steps by very simple gadgets.
As a result, we obtain the same complexity results for parity games with weights as for energy parity games, i.e., $\np \cap \conp$, the signature complexity for finite games of infinite duration with parity conditions and their extensions.
Thereby, we obtain an argument that these games might be representatives of a natural complexity class, lending a further argument for the relevance of two player games with mixed qualitative and quantitative winning conditions.
Furthermore, Daviaud et al.\ recently showed that parity games with weights can even be solved in pseudo-quasi-polynomial time~\cite{DaviaudJurdzinskiLazic18}.

Naturally, parity games with weights subsume parity games (as a special case where all weights are zero), finitary parity games (as a special case where all weights are positive), and parity games with costs (as a special case where all weights are nonnegative).

We show that the protagonist has finite-state winning strategies, and provide tight pseudo-polynomial bounds for the memory he needs to win the game.
We also present tight bounds on the quality of winning strategies for the protagonist.
Naturally, the antagonist may need infinite memory to win.

Solving parity games with weights amounts to determining whether there exists a bound~$b$ such that the protagonist is able to enforce a play in which the amplitude of the accumulated weight between almost all requests and their corresponding answer is bounded by~$b$.
The value of~$b$, however, may be arbitrarily large, subject to the bounds on the quality of winning strategies.
Hence, it is natural to consider the threshold problem for parity games with weights:
``Given a parity game with weights~$\game$ and a bound~$b \in \nats$, is the protagonist able to enforce that the paths between almost all requests and their respective answers have an amplitude of at most~$b$?''
It is known that the threshold problem for finitary parity games and parity games with costs is \pspace-complete~\cite{WeinertZimmermann17}.
In this work, we show that the complexity increases even further in the case of parity games with weights, as the threshold problem for such games is \exptime-complete.
The complexity of strategies witnessing the answer to the threshold problem, however, does not increase: Both players require exponential memory in order to ensure or violate the bound~$b$, respectively, if they are able to do so at all.

This paper is an extended version of work published at CSL 2018~\cite{ScheweWeinertZimmermann18csl}.


\section{Preliminaries}%
\label{sec:preliminaries}
We denote the nonnegative integers by~$\nats$, the integers by~$\ints$, and define~$\nats_\infty = \nats \cup \set{\infty}$.
As usual, we have~$\infty > n$,~$-\infty < n$,~$n + \infty = \infty$, and~$-\infty - n = -\infty$ for all~$n \in \ints$.

An \textbf{arena}~$\arena=(V, V_0, V_1, E)$ consists of a finite, directed graph~$(V, E)$ and a partition~$\{V_0, V_1\}$ of $V$ into the positions of Player~$0$
(drawn as ellipses) and Player~$1$ (drawn as rectangles).
For pronomial convenience, we refer to Player~$0$ as he, and to Player~$1$ as she.
The size of $\arena$, denoted by $\size{\arena}$, is defined as $\size{V}$.
A \textbf{play} in $\arena$ is an infinite path~$\rho = v_0 v_1 v_2 \cdots$ through $(V, E)$.
To rule out finite plays, we require every vertex to be nonterminal.
We define~$\card{\rho} = \infty$.
Dually, for a finite play prefix~$\pi = v_0\cdots v_j$ we define~$\card{\pi} = j+1$.

A \textbf{game}~$\game = ( \arena, \wincond )$ consists of an arena $\arena$ with vertex set~$V$ and a set~$\wincond \subseteq V^\omega$ of winning plays for Player~$0$.
The set of winning plays for Player~$1$ is $V^\omega \setminus \wincond$. A winning condition~$\wincond$ is $0$-extendable if,
for all $\rho \in V^\omega$ and all $w \in V^*$, $\rho \in \wincond$ implies $w\rho \in \wincond$. Dually, $\wincond$ is $1$-extendable if,
for all $\rho \in V^\omega$ and all $w \in V^*$, $\rho \notin \wincond$ implies $w\rho \notin \wincond$. Finally, $\wincond$ is prefix-independent, if it is both $0$-extendable and $1$-extendable.

A \textbf{strategy} for Player~$i \in \set{0,1}$ is a mapping $\sigma \colon V^*V_i \rightarrow V$ such that $(v, \sigma(wv)) \in E$ holds true for all $wv \in V^* V_i$.
We say that $\sigma$ is \textbf{positional} if $\sigma(wv) = \sigma(v)$ holds true for every $wv \in V^*V_i$.
A play $v_0 v_1 v_2 \cdots$ is \textbf{consistent} with a strategy~$\sigma$ for Player~$i$, if $v_{j+1} = \sigma( v_0 \cdots v_j)$ holds true for every~$j$ with $v_j \in V_i$. A strategy~$\sigma$ for Player~$i$ is a \textbf{winning strategy} for $\game$ from~$v \in V$ if every play that starts in~$v$ and is consistent with $\sigma$ is won by Player~$i$.
If Player~$i$ has a winning strategy from~$v$, then we say Player~$i$ wins $\game$ from~$v$.
The \textbf{winning region} of Player~$i$ is the set of vertices, from which Player~$i$ wins~$\game$; it is denoted by~$\winreg_i(\game)$.
\textbf{Solving} a game amounts to determining its winning regions.
If~$\winreg_0(\game) \cup \winreg_1(\game) = V$, then we say that~$\game$ is \textbf{determined}.

Let~$\arena = (V, V_0, V_1, E)$ be an arena and let~$X \subseteq V$.
The~$i$-attractor of~$X$ is defined inductively as $\att{i}{}{X} = \att{i}{\card{V}}{X}$, where $\att{i}{0}{X} = X$ and
\begin{multline*}
	\att{i}{j}{X} = \att{i}{j-1}{X} \cup \set{ v \in V_i \mid \exists v' \in \att{i}{j-1}{X}.\, (v,v') \in E} \\ \cup \set{ v \in V_{1-i} \mid \forall (v, v') \in E.\, v' \in \att{i}{j-1}{X}}   \enspace.
\end{multline*}
Hence, $\att{i}{}{X}$ is the set of vertices from which Player~$i$ can force the play to enter~$X$:
Player~$i$ has a positional strategy~$\sigma_X$ such that each play that starts in some vertex in $\att{i}{}{X}$ and is consistent with~$\sigma_X$ eventually encounters some vertex from~$X$.
We call~$\sigma_X$ an attractor strategy towards~$X$.
Moreover, the~$i$-attractor can be computed in time linear in~$\card{E}$~\cite{NerodeRemmelYakhnis96}.
When we want to stress the arena~$\arena$ the attractor is computed in, we write~$\att{i}{\arena}{X}$.

A set~$X \subseteq V$ is a trap for Player~$i$, if every vertex in $X \cap V_i$ has only successors in $X$ and every vertex in $X \cap V_{1-i}$ has at least one successor in $X$. In this case, Player~$1-i$ has a positional strategy~$\tau_X$ such that every play starting in some vertex in $X$ and consistent with~$\tau_X$ never leaves $X$. We call such a strategy a trap strategy.

\begin{numrem}%
\label{remark:traps}
\hfill
\begin{enumerate}
	\item\label{remark:traps:attractorcomplement} The complement of an $i$-attractor is a trap for Player~$i$.

	\item\label{remark:traps:trapattractors} If $X$ is a trap for Player~$i$, then $\att{1-i}{}{X}$ is also a trap for Player~$i$.

	\item\label{remark:traps:extendability} If $\wincond$ is $i$-extendable and $(\arena,\wincond)$ determined, then $\winreg_{1-i}(\arena,\wincond)$ is a trap for Player~$i$.

\end{enumerate}
\end{numrem}

\noindent
A \textbf{memory structure}~$\mem = (M, \init, \update)$ for an arena $(V, V_0, V_1, E)$ consists of a finite set~$M$ of memory states, an initialization function $\init\colon V \rightarrow M$, and an update function~$\update\colon M \times E \rightarrow M$.
The update function can be extended to finite play prefixes in the usual way: $\update^+(v) = \init(v)$ and $\update^+(w v v') = \update(\update^+(w v ), (v,v'))$ for $w \in V^*$ and $(v,v') \in E$.
A next-move function $\nxt \colon V_i \times M \rightarrow V$ for Player~$i$ has to satisfy $(v, \nxt(v, m)) \in E$ for all $v \in V_i$ and $m \in M$.
It induces a strategy~$\sigma$ for Player~$i$ with memory~$\mem$ via $\sigma(v_0\cdots v_j) = \nxt(v_j, \update^+(v_0 \cdots v_j))$.
A strategy is called \textbf{finite-state} if it can be implemented by a memory structure.
We define $\card{\mem} = \card{M}$.
Slightly abusively, we say that the size of a finite-state strategy is the size of a memory structure implementing it.



\section{Parity Games with Weights}%
\label{sec:parity-weights-def}
Fix an arena~$\arena = (V, V_0, V_1, E)$.
A {\bf weighting} for~$\arena$ is a function~$\weight\colon E \rightarrow \ints$.
We define~$\weight(\epsilon) = \weight(v) = 0$ for all~$v \in V$ and extend~$\weight$ to sequences of vertices of length at least two by summing up the weights of the traversed edges.
Given a play (prefix)~$\pi = v_0v_1v_2\cdots$, we define the amplitude of~$\pi$ as~$\ampl(\pi) = \sup_{j < \card{\pi}} \abs{\weight(v_0 \cdots v_j)} \in \nats_\infty$.

A {\bf coloring} of $V$ is a function~$\col \colon V \rightarrow \nats$.
The classical parity condition requires almost all occurrences of odd colors to be answered by a later occurrence of a larger even color.
Hence, let $\answer{c} = \set{c' \in \nats \mid c' \ge c \text{ and $c'$ is even}}$ be the set of colors that ``answer'' a ``request'' for color~$c$.
We denote a vertex~$v$ of color~$c$ by~\parnode{v}{c}.

Fijalkow and Zimmermann introduced a generalization of the parity condition and the finitary parity condition~\cite{ChatterjeeHenzingerHorn09}, the parity condition with costs~\cite{FijalkowZimmermann14}.
There, the edges of the arena are labeled with \emph{nonnegative weights} and the winning condition demands that there exists a bound~$b$ such that almost all requests are answered with weight at most $b$, i.e., the weight of the infix between the request and the response has to be bounded by $b$.

Our aim is to extend the parity condition with costs by allowing for the full spectrum of weights to be used, i.e., by also incorporating negative weights.
In this setting, the weight of an infix between a request and a response might be negative. Thus, the extended condition requires the weight of the infix to be bounded from above and from below.\footnote{We discuss other possible interpretations of negative weights in Section~\ref{sec:conclusion}.}
To distinguish between the parity condition with costs and the extension introduced here, we call our extension the parity condition with weights.

Formally, let $\rho = v_0 v_1 v_2 \cdots$ be a play.
We define the cost-of-response at position~$j \in \nats$ of $\rho$ by
\[
	\Cor(\rho, j) = \min\set{\ampl(v_j \cdots v_{j'}) \mid j' \geq j, \col(v_{j'}) \in \answer{\col(v_j)}} ,
\]
where we use $\min \emptyset = \infty$.
As the amplitude of an infix only increases by extending the infix, $\Cor(\rho, j)$ is the amplitude of the shortest infix that starts at position~$j$ and ends at an answer to the request posed at position~$j$.
We illustrate this notion in Figure~\ref{fig:cor}.

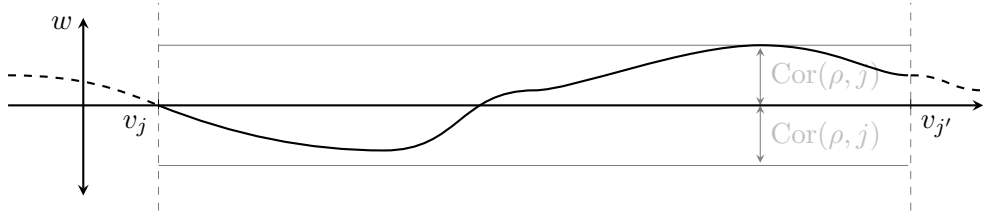
\begin{figure}
\centering
\begin{tikzpicture}[thick,yscale=.4]
	\draw[stealth-stealth] (-1,-3) -- (-1,3);
	\draw (-2,0) edge (11,0);
	
	\node [anchor=east] at (-1,2.75) {$\weight$};
	
	\draw[thick,dashed]
		(-2,1) .. controls (-1,1) and (-.5,.5) .. (0,0);
		
	\draw[thick]
		(0,0) .. controls (1,-1) and (2, -1.5) ..
		(3,-1.5) .. controls (4,-1.5) and (4, .5) ..
		(5,.5) .. controls (5.5,.5) and (7, 2) ..
		(8, 2) .. controls (9, 2) and (9.5,1) .. (10,1);
		
	\draw[thick,dashed]
		(10,1) .. controls (10.5,1) and (10.5,.5) .. (11,.5);
		
	\begin{pgfonlayer}{background}[thick]
		\draw[gray,dashed] (0,-3.5) -- (0,3.5);
		\draw[gray,dashed] (10,-3.5) -- (10,3.5);
		
		\draw[gray]
			(0,2) -- (10,2)
			(0,-2) -- (10,-2);
			
		\draw[stealth-stealth,gray]
			(8,0) -- node [anchor=west,lightgray] {$\Cor(\rho, j)$} (8,2);
		\draw[stealth-stealth,gray]
			(8,0) -- node [anchor=west,lightgray] {$\Cor(\rho, j)$} (8,-2);
	\end{pgfonlayer}{background}
		
	\draw (0,.2) -- (0,-.2);
	\node[anchor=north east] at (0,0) {$v_j$};
	\draw (10,.2) -- (10,-.2);
	\node[anchor=north west] at (10,0) {$v_{j'}$};
	
\end{tikzpicture}

\caption{The cost-of-response of some request posed by visiting vertex~$v_j$, which is answered by visiting vertex~$v_{j'}$.}%
\label{fig:cor}
\end{figure}

We say that a request at position~$j$ is answered with cost~$b$, if $\Cor(\rho, j) = b$.
Consequently, a request with an even color is answered with cost zero.
The cost-of-response of an unanswered request is infinite, even if the amplitude of the remaining play is bounded.
In particular, this means that an unanswered request at position~$j$ may be ``unanswered with finite cost~$b$'' (if the amplitude of the remaining play is~$b \in \nats$) or ``unanswered with infinite cost'' (if the amplitude of the remaining play is infinite).
In either case, however, we have~$\Cor(\rho, j) = \infty$.

We define the parity condition with weights as
\[\cp(\col, \weight) = \set{ \rho \in V^\omega \mid \limsup\nolimits_{j\rightarrow \infty} \Cor(\rho, j) \in \nats } \enspace.\]
I.e., $\rho$ satisfies the condition if and only if there exists a bound~$b \in \nats$ such that almost all requests are answered with cost less than $b$.
In particular, only finitely many requests may be unanswered, even with finite cost.
Note that the bound~$b$ may depend on the play $\rho$.

\begin{table}
\footnotesize
\centering
\begin{tabular}{lccccc}\toprule
	& Complexity & Mem.\ Pl.~$0$/Pl.~$1$ &  Bounds \\
	\midrule
	Parity Games~\cite{CaludeJKLS/17/qp} & quasi-poly. & pos./pos. & -- \\ 
	Energy Parity Games~\cite{ChatterjeeDoyen12,DaviaudJurdzinskiLazic18} &  pseudo-quasi-poly. & $\bigo(ndW)$/pos. & $\bigo(nW)$ \\ \midrule
	Finitary Parity Games~\cite{ChatterjeeHenzingerHorn09} & poly. & pos./inf. & $\bigo(nW)$ \\
	Parity Games with Costs~\cite{FijalkowZimmermann14,MogaveroMS15} &  quasi-poly. & pos./inf. & $\bigo(nW)$ \\ \bottomrule
\end{tabular}
\caption{Characteristic properties of variants of parity games.}%
\label{tab:characteristics-prev}
\end{table}

We call a game~$\game = (\arena, \cp(\col, \weight))$ a parity game with weights, and we define~$\card{\game} = \card{\arena} + \log(W)$, where~$W$ is the largest absolute weight assigned by~$\weight$; i.e., we assume weights to be encoded in binary.
If $\weight$ assigns zero to every edge, then $\cp(\col, \weight)$ is a classical (max-) parity condition, denoted by $\parity(\col)$.
Similarly, if $\weight$ assigns positive weights to every edge, then $\cp(\col, \weight)$ is equal to the finitary parity condition over~$\col$, as introduced by Chatterjee and Henzinger~\cite{ChatterjeeHenzinger06}.
Finally, if~$\weight$ assigns only nonnegative weights, then~$\cp(\col, \weight)$ is a parity condition with costs, as introduced by Fijalkow and Zimmermann~\cite{FijalkowZimmermann14}.
In these cases, we refer to~$\game$ as a parity game, a finitary parity game, or a parity game with costs, respectively. Dually, every parity game, finitary parity game, and parity game with costs is a parity game with weights.
We recall the characteristics of these special cases in Table~\ref{tab:characteristics-prev}.


\section{Solving Parity Games with Weights}%
\label{sec:parity-weights-solve}

We now show how to solve parity games with weights.
Our approach is inspired by the classic work on finitary parity games~\cite{ChatterjeeHenzingerHorn09} and parity games with costs~\cite{FijalkowZimmermann14}:
We first define a stricter variant of these games, which we call bounded parity games with weights, and then show two reductions:
\begin{itemize}
 \item parity games with weights can be solved in polynomial time with oracles that solve bounded parity games with weights (in this section); and
 \item bounded parity games with weights can be solved in polynomial time with oracles that solve energy parity games (Section~\ref{sec:bounded-parity-weights-solve}).
\end{itemize}

\noindent
Furthermore, we provide a polynomial time reduction from solving energy parity games to solving parity games with weights  in Section~\ref{sec:ep2(b)pw}.
We thereby show that parity games with weights, bounded parity games with weights, and energy parity games belong to the same complexity class.

The energy parity games that we reduce to are known to be efficiently solvable, as they are in $\np \cap \conp$ due to Chatterjee and Doyen~\cite{ChatterjeeDoyen12}.
Moreover, they are \logspace-equivalent to mean-payoff parity games as introduced by Chatterjee, Henzinger, and Jurdzi{\'n}ski~\cite{ChatterjeeHenzingerJurdzinski05}.
Hence they can be solved in pseudo-quasi-polynomial time, due to recent advances by Daviaud, Jurdzi{\'n}ski, and Lazi{\'c}~\cite{DaviaudJurdzinskiLazic18}.

We first introduce the \textbf{ bounded parity condition with weights}, which is a strengthening of the parity condition with weights. Hence, it is also induced by a coloring and a weighting:
\begin{multline*}
	\bcp(\col, \weight) = \cp(\col, \weight) \\
	\cap \set{ \rho \in V^\omega \mid \text{no request in } \rho \text{ is unanswered with infinite cost} } \enspace.
\end{multline*}
Note that this condition allows for a finite number of unanswered requests, as long as they are unanswered with finite cost.

We solve parity games with weights by repeatedly solving bounded parity games with weights.
To this end, we apply the following two properties of the winning conditions:
We have $\bcp(\col,\weight) \subseteq \cp(\col,\weight)$ as well as that $\cp(\col,\weight)$ is~$0$-extendable.
Hence, if Player~$0$ has a strategy from a vertex~$v$ such that every consistent play has a suffix in $\bcp(\col,\weight)$, then the strategy is winning for her from $v$ w.r.t.\ $\cp(\col,\weight)$.
Thus,
\[ \att{0}{}{\winreg_0(\arena, \bcp(\col, \weight))} \subseteq \winreg_0(\arena, \cp(\col,\weight)) \enspace . \]
The algorithm that solves parity games with weights repeatedly removes attractors of winning regions of the bounded parity game with weights until a fixed point is reached.
We will later formalize this sketch to show that the removed parts are a subset of Player~$0$'s winning region in the parity game with weights.

To show that the obtained fixed point covers the complete winning region of Player~$0$, we use the following lemma to show that the remaining vertices are a subset of Player~$1$'s winning region in the parity game with weights. The proof is very similar to the corresponding one for finitary parity games and parity games with costs.

\begin{lem}%
\label{lem:parity-to-bounded-parity:reduction}
Let~$\game = (\arena, \cp(\col, \weight))$ and let~$\game' = (\arena, \bcp(\col, \weight))$.
If $\winreg_0(\game') = \emptyset$, then $\winreg_0(\game) = \emptyset$.
\end{lem}

\begin{proof}
As bounded parity conditions with weights are Borel, bounded parity games with weights are determined~\cite{Martin75}.
Hence, $\winreg_0(\game') = \emptyset$ implies that, for every vertex~$v$ of~$\arena$, Player~$1$ has a strategy~$\tau_v$ that is winning in $\game'$ from $v$.

We combine these strategies into a single strategy~$\tau$ for Player~$1$ that is winning in~$\game$ from every vertex of $\arena$.
This strategy is controlled by a vertex~$v^*$ (initialized with the starting vertex of the play) and a counter~$\kappa$ ranging over $\nats$ (initialized with zero).
The strategy~$\tau$ mimics the strategy~$\tau_{v^*}$ from $v^*$ until a request is followed by an infix without an answer and with amplitude $\kappa$. This implies that the cost-of-response of this request is at least~$\kappa$. If such a situation is encountered, then $v^*$ is set to the current vertex and $\kappa$ is incremented.
Furthermore, the history of the play is discarded at this point in the play, and~$\tau$ behaves henceforth like $\tau_{v^*}$ when starting at $v^*$ when this happens.

We now show that~$\tau$ is winning for Player~$1$ from every vertex in~$\game$.
Consider a play~$\rho$ that is consistent with this strategy.
If, on the one hand, $\kappa$ is updated infinitely often along~$\rho$, then~$\rho$ contains, for every $b \in \nats$, a request that has a cost-of-response that is larger than $b$.
Hence, it violates the parity condition with weights.

If, on the other hand, $\kappa$ is only updated finitely often, then $\rho$ has a suffix~$\rho'$ that starts in some $v$, which is consistent with $\tau_v$.
As $\tau_v$ is winning for Player~$1$ from $v$ in~$\game'$, $\rho'$ violates the bounded parity condition with weights.
Also, because $\kappa$ is updated only finitely often during the suffix, there is a bound~$b$ such that the amplitude of every suffix of $\rho'$ that starts at a request is bounded by $b$. Hence, the only way for $\rho'$ to violate the bounded parity condition with weights is to violate the parity condition.
Thus, the full play~$\rho$ also violates the parity condition, and therefore also the parity condition with weights, which is a strengthening of the parity condition.
Therefore, $\tau$ is indeed winning for Player~$1$ from every vertex in $\game$.
\end{proof}

Lemma~\ref{lem:parity-to-bounded-parity:reduction} implies that the algorithm for solving parity games with weights by repeatedly solving bounded parity games with weights (see Algorithm~\ref{algorithm_fixpoint}) is correct.
Note that we use an oracle for solving bounded parity games with weights.
We provide a suitable algorithm in Section~\ref{sec:bounded-parity-weights-solve}.

\begin{algorithm}
\begin{algorithmic}[1]
 \REQUIRE{Parity game with weights~$\game$ with arena~$\arena$, coloring~$\col$ and weighting~$\weight$}
 \STATE{$k = 0$}; $W_0^k = \emptyset$; $\arena_k = \arena$
 \REPEAT%
 \STATE{$k = k + 1$}
 \STATE{$X_k = \winreg_0(\arena_{k-1},\bcp(\col, \weight))$}
 \STATE{$W_0^k = W_0^{k-1} \cup \att{0}{\arena_{k-1}}{X_k}$}
 \STATE{$\arena_k = \arena_{k-1} \setminus \att{0}{\arena_{k-1}}{X_k}$}
 \UNTIL{$X_k = \emptyset$}
 \RETURN{$W_0^k$}
\end{algorithmic}
\caption{A fixed-point algorithm computing $\winreg_0(\arena, \cp(\col, \weight))$.}%
\label{algorithm_fixpoint}
\end{algorithm}

The loop terminates after at most $\size{\arena}$  iterations (assuming the algorithm solving bounded parity games with weights terminates), as, during each iteration, at least one vertex is removed from the arena. The correctness proof relies on Lemma~\ref{lem:parity-to-bounded-parity:reduction} and is similar to the ones for finitary parity games~\cite{ChatterjeeHenzingerHorn09} and for parity games with costs~\cite{FijalkowZimmermann14}.
The underlying argument only relies on a few properties of the winning condition;
it is very general---and fairly standard.
We use the same argument again later to establish the correctness of Algorithm~\ref{algorithm_fixpoint2}.

\begin{lem}%
\label{lem:algorithm-parity-weights-correctness}
Algorithm~\ref{algorithm_fixpoint} returns $\winreg_0(\arena, \cp(\col, \weight))$.
\end{lem}

\begin{proof}
Let~$\game = (\arena, \cp(\col, \weight))$ and let $k^*$ be the final iteration when running the algorithm on $\game$, i.e., its output is $W_0^{k^*} = \bigcup_{0 < k < k^*} \att{0}{\arena_{k-1}}{X_{k}}$.
First, we consider Player~$0$ and show $W_0^{k^*}  \subseteq \winreg_0(\game)$. For every vertex~$v$ that is in some~$X_{k}$, Player~$0$ has a strategy~$\sigma_v$ for $\game_{k} = (\arena_{k-1},\bcp(\col, \weight))$ that is winning from~$v$. Furthermore, for every attractor~$\att{0}{\arena_{k-1}}{X_k}$, he has a positional attractor strategy~$\sigma_k$. Now, we compose these strategies to a strategy~$\sigma$ for Player~$0$ in $\arena$ via
\[
\sigma(v_0 \cdots v_j) = \begin{cases}
\sigma_{k}(v_j) & \text{if } v_j \in \att{0}{\arena_{k-1}}{X_k} \setminus X_{k},\\
\sigma_{v_{j'}}(v_{j'} \cdots v_{j}) & \text{if } v_j \in X_k.\\
\end{cases}
\]
In the second case, $v_{j'} \cdots v_{j}$ is the longest suffix of $v_{0} \cdots v_{j}$ that only contains vertices from $X_{k}$, the set of vertices from which Player~$0$ has a winning strategy for~$\game_k$.

Consider a play~$\rho = v_0 v_1 v_2 \cdots$ in $\arena$ that starts in $W_0^{k^*}$ and that is consistent with $\sigma$. For every $j$ there is a unique $k_j$ in the range~$0 < k_j < k^*$ such that $v _j \in \att{0}{\arena_{k_j-1}}{X_{k_j}}$. As $\bcp(\col,\weight)$ is $1$-extendable, Items~\ref{remark:traps:trapattractors} and~\ref{remark:traps:extendability} of Remark~\ref{remark:traps} imply that each $\att{0}{\arena_{k-1}}{X_{k}}$ is a trap for Player~$1$ in $\arena_{k-1}$. Hence, we obtain~$k_0 > k_1 > k_2 > \cdots$. As the~$k_j$ are always greater than zero, the sequence stabilizes eventually. This implies that $\rho$ has a suffix~$\rho' = v_j v_{j+1} v_{j+2} \cdots$ that is consistent with $\sigma_{v_j}$.

Hence, due to $\sigma_{v_j}$ being a winning strategy for Player~$0$ in $\game_{k}$ from $v_j$, we obtain $\rho' \in \bcp(\col,\weight)$. Hence, $\rho \in \cp(\col,\weight)$, since the bounded parity condition with weights is a strengthening of its unbounded variant and due to $0$-extendability of $\cp(\col,\weight)$. Hence, $\sigma$ is indeed winning from $W_0^{k^*}$.

Now, consider Player~$1$. We show $V \setminus W_0^{k^*} \subseteq \winreg_1(\game)$. Then, determinacy of parity games with weights (due to their winning conditions being Borel~\cite{Martin75}) yields $W_0^{k^*} = \winreg_0(\game)$ and $V \setminus W_0^{k^*} = \winreg_1(\game)$.

Due to $X_{k^*}$ being empty and bounded parity games with weights being determined (again due to their winning conditions being Borel), Player~$1$ wins the bounded parity game with weights~$\game_{k^*}$ from every vertex. Applying Lemma~\ref{lem:parity-to-bounded-parity:reduction} shows that she also wins the parity game with weights~$(\arena_{{k^*}-1},\cp(\col,\weight))$ from every vertex. Finally, as $V \setminus W_0^{k^*} = \winreg_1(\arena_{{k^*}-1},\cp(\col,\weight))$ is a trap for Player~$0$ in $\arena$ by construction, he also wins $\game = (\arena, \cp(\col,\weight))$ from every vertex in $V \setminus W_0^{k^*}$.
\end{proof}

The winning strategy for Player~$0$ defined in the proof of Lemma~\ref{lem:algorithm-parity-weights-correctness} can be implemented by a memory structure of size~$\max_{k < k^*} s_{k}$, where~$s_k$ is the size of a winning strategy~$\sigma_k$ for Player~$0$ in the bounded parity game with weights solved in the~$k$-th iteration, and where~$k^*$ is the value of~$k$ at termination.
To this end, one uses the fact that the winning regions~$X_k$ are disjoint and are never revisited once left.
Hence, the implementations of the $\sigma_k$ can use the same states yielding the upper bound~$\max_{k < k^*} s_{k}$.

\section{Solving Bounded Parity Games with Weights}%
\label{sec:bounded-parity-weights-solve}
After having reduced the problem of solving parity games with weights to that of solving (multiple) bounded parity games with weights, we reduce solving bounded parity games with weights to solving (multiple) energy parity games~\cite{ChatterjeeDoyen12}.

Similar to a parity game with weights, in an energy parity game, the vertices are colored and the edges are equipped with weights.
It is the goal of Player~$0$ to satisfy the parity condition, while, at the same time, ensuring that the accumulated weight of every prefix, its so-called energy level, is bounded from below.
In contrast to a parity game with weights, however, the weights in an energy parity game are not ``tied'' to the requests and responses denoted by the coloring.

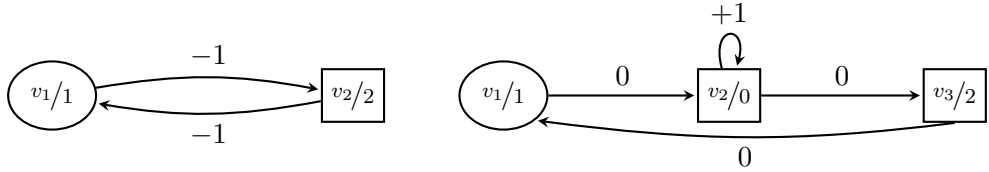
\begin{figure}
\centering
\begin{tikzpicture}
	\begin{scope}[shift={(-3,0)}]
		\node[p0] (v1) at (0,0) {$\nicefrac{v_1}{1}$};
		\node[p1] (v2) at (4,0) {$\nicefrac{v_2}{2}$};
		
		\path
			(v1) edge[bend left=10] node[anchor=south] {$-1$} (v2)
			(v2) edge[bend left=10] node[anchor=north] {$-1$} (v1);
	\end{scope}
	
	\begin{scope}[shift={(3,0)}]
		\node[p0] (v1) at (0,0) {$\nicefrac{v_1}{1}$};
		\node[p1] (v2) at (3,0) {$\nicefrac{v_2}{0}$};
		\node[p1] (v3) at (6,0) {$\nicefrac{v_3}{2}$};
		
		\path
			(v1) edge node[anchor=south] {$0$} (v2)
			(v2) edge[loop above] node[anchor=south] {$+1$} (v2)
			(v2) edge node[anchor=south] {$0$} (v3)
			(v3.south) edge[bend left=7.5] node[anchor=north] {$0$} (v1.south east);
	\end{scope}
\end{tikzpicture}

\caption{The difference between energy parity games and parity games with weights.}%
\label{fig:energyparity:difference}
\end{figure}
Consider, for example, the games shown in Figure~\ref{fig:energyparity:difference}.
In the game on the left-hand side, players only have a single, trivial strategy.
If we interpret this game as a parity game with weights, Player~$0$ wins from every vertex, as each request is answered with cost one.
If we, however, interpret that game as an energy parity game, Player~$1$ instead wins from every vertex, since the energy level decreases by one with every move.
In the game on the right-hand side, the situation is reversed:
When interpreting this game as a parity game with weights, Player~$1$ wins from every vertex, as she can easily unbound the costs of the requests for color one by staying in vertex~$v_2$ for an ever-increasing number of cycles.
Dually, when interpreting this game as an energy parity game, Player~$0$ wins from every vertex, since the parity condition is clearly satisfied in every play, and Player~$1$ is only able to increase the energy level, while it is never decreased.

In Section~\ref{subsec:bounded-parity-weights-solve:energy-parity-def}, we introduce energy parity games formally and present how to solve bounded parity games with weights via energy games in Section~\ref{subsec:bounded-parity-weights-solve:bpw2ep}.
Finally, Sections~\ref{subsec:bounded-parity-weights-solve:proof-player1} and~\ref{subsec:bounded-parity-weights-solve:proof-player0} are dedicated to the correctness proof of the construction.


\subsection{Energy Parity Games}%
\label{subsec:bounded-parity-weights-solve:energy-parity-def}
An energy parity game~$\game = (\arena, \col, \weight)$ consists of an arena~$\arena = (V, V_0, V_1, E)$, a coloring~$\col \colon V \rightarrow \nats$ of~$V$, and an edge weighting~$\weight\colon E \rightarrow \ints$ of~$E$.
Note that this definition is not compatible with the framework presented in Section~\ref{sec:preliminaries}, as we have not (yet) defined the winner of the plays.
This is because they depend on an initial credit, which is existentially quantified in the definition of winning the game~$\game$.
Formally, the set of winning plays with initial credit~$c_0 \in \nats$ is defined as
\[ \energyparity_{c_0}(\col, \weight) = \parity(\col) \cap \set{v_0v_1v_2\cdots \in V^\omega \mid \forall j \in \nats.\, c_0 + \weight(v_0\cdots v_j) \geq 0 } \enspace. \]
Now, we say that Player~$0$ wins~$\game$ from~$v$ if there exists some initial credit~$c_0 \in \nats$ such that he wins~$\game_{c_0} = (\arena, \energyparity_{c_0}(\col, \weight))$ from~$v$ (in the sense of the definitions in Section~\ref{sec:preliminaries}).
If this is not the case, i.e., if Player~$1$ wins $\game_{c_0}$ from $v$ for every $c_0$, then we say that Player~$1$ wins $\game$ from $v$. Note that the initial credit is uniform for all plays, unlike the bound on the cost-of-response in the definition of the parity condition with weights, which may, a priori, depend on the play.

Unraveling these definitions shows that Player~$0$ wins $\game$ from $v$ if there is an initial credit~$c_0$ and a strategy~$\sigma$, such that every play that starts in $v$ and is consistent with $\sigma$ satisfies the parity condition \emph{and} the accumulated weight over the play prefixes (the energy level) never drops below~$-c_0$.
We call such a strategy~$\sigma$ a winning strategy for Player~$0$ in~$\game$ from~$v$.
Dually, Player~$1$ wins $\game$ from $v$ if, for every initial credit~$c_0$, there is a strategy~$\tau_{c_0}$, such that every play that starts in $v$ and is consistent with $\tau_{c_0}$ violates the parity condition \emph{or} its energy level drops below~$-c_0$ at least once.
Thus, the strategy~$\tau_{c_0}$ may, as the notation suggests, depend on $c_0$.
However, Chatterjee and Doyen~\cite{ChatterjeeDoyen12} showed that using different strategies is not necessary:
There is a uniform strategy~$\tau$ that is winning from $v$ for every initial credit~$c_0$.

\begin{propC}[\cite{ChatterjeeDoyen12}]%
\label{prop:energyparity:player1}
Let $\game$ be an energy parity game. If Player~$1$ wins $\game$ from $v$, then she has a single positional strategy that is winning from $v$ in $\game_{c_0}$ for every~$c_0$.
\end{propC}

We call such a strategy as in Proposition~\ref{prop:energyparity:player1} a winning strategy for Player~$1$ from~$v$.
A play consistent with such a strategy either violates the parity condition, or the energy levels of its prefixes diverge towards~$-\infty$.

Furthermore, Chatterjee and Doyen obtained an upper bound on the initial credit necessary for Player~$0$ to win an energy parity game, as well as an upper bound on the size of a corresponding finite-state winning strategy.

\begin{propC}[\cite{ChatterjeeDoyen12}]%
\label{prop:energyparity:equivalence}
Let $\game$ be an energy parity game with $n$ vertices, $d$ colors, and largest absolute weight~$W$. The following are equivalent for a vertex~$v$ of $\game$:
\begin{enumerate}
	\item Player~$0$ wins $\game$ from~$v$.
	\item Player~$0$ wins $\game_{(n-1)W}$ from~$v$ with a finite-state strategy with at most $ndW$ states.
\end{enumerate}
\end{propC}

\noindent
The previous proposition yields that finite-state strategies of bounded size suffice for Player~$0$ to win.
Such strategies do not admit long expensive descents, which we show via a straightforward pumping argument.

\begin{lem}%
\label{lem:boundeddecreaseinenergyparity}
Let $\game$ be an energy parity game with $n$ vertices and largest absolute weight~$W$. Further, let $\sigma$ be a finite-state strategy of size~$s$, and let $\rho$ be a play that starts in some vertex, from which $\sigma$ is winning, and is consistent with $\sigma$. Every infix~$\pi$ of $\rho$ satisfies
\[
    \weight(\pi) > -(ns -1 )W -1 \enspace.
\]
\end{lem}

\begin{proof}
Let $\sigma$ be implemented by $\mem = (M, \init, \update)$ and let $\rho = v_0 v_1 v_2 \cdots$.
We assume towards a contradiction that there is an infix~$\pi = v_j \cdots v_{j'}$ with $\weight(\pi) \le -(ns -1 )W -1$.
We assume w.l.o.g.~$\pi$ to be minimal with this property, i.e., such that there exists no proper prefix~$\pi'$ of~$\pi$ with~$\weight(\pi') \leq -(ns -1 )W -1$.

We define a sequence~$j_0 < \cdots < j_k \leq j'$ of positions by starting with~$j_0 = j$.
Inductively, for each~$j_{k'}$, we define~$j_{k' + 1}$ to be the minimal position strictly greater than~$j_{k'}$ that satisfies~$\weight(v_{j} \cdots v_{j_{k' + 1}}) < \weight(v_{j} \cdots v_{j_{k'}})$.
Intuitively, the positions~$j_{k'}$ are those positions at which the weight of the infix~$v_j \cdots v_{j_{k'}}$ reaches a new lower bound.
Due to minimality of~$\pi'$, we obtain~$j_k = j'$.

We now show that there are two positions~$j_{k_\ell}$ and~$j_{k_{\ell'}}$ such that the weight along the infix from the former to the latter position decreases and such that we are able to pump that infix while retaining consistency with the strategy~$\sigma$.
To this end, we define~$w_\ell = \weight(v_j \cdots v_{j_\ell})$ for all~$\ell \in \set{0,\dots,k}$.
In particular, we have~$w_0 = 0$.
Since each edge has an absolute weight of at most~$W$, we obtain~$w_{\ell + 1} \geq w_\ell - W$ for all~$\ell \in \set{0,\dots,k-1}$, which in turn implies~$w_k \geq -kW$.
Moreover, since we have~$j_k = j'$, we additionally obtain~$w_k \leq -(ns - 1)W - 1$.
Rearranging the resulting inequality $-(ns - 1)W - 1 \geq -kW$ yields~$kW \geq (ns-1)W + 1$.

In order to obtain a lower bound for~$k$, we briefly have to argue that~$W > 0$ holds true.
First, we obtain~$W \geq 0$ by definition of~$W$.
Furthermore, if~$W = 0$ then we obtain~$\weight(\pi) = 0$, which contradicts~$\weight(\pi) \leq -(ns -1 )W -1 = -1$.
Hence, we indeed have~$W > 0$, which yields~$k > (ns-1)$ via the above inequality.
In particular, this implies that there exist at least~$ns + 1$ positions at which the accumulated weight of the infix so far attains a new minimum, since we start counting the number of positions at zero.

Due to the pigeon-hole-principle, we obtain that there exist indices~$\ell$,~$\ell'$ with~$0 \leq \ell < \ell' \leq k$ such that~$v_{j_\ell} = v_{j_{\ell'}}$ and such that~$\update^+(v_0\cdots v_{j_\ell}) = \update^+(v_0\cdots v_{j_{\ell'}})$.
Thus, the play~$v_0 \cdots v_j \cdots v_{j_\ell} {(v_{j_\ell + 1} \cdots v_{j_{\ell'}})}^\omega$ obtained by repeating the loop between $v_{j_\ell}$ and $v_{j_{\ell'}}$ ad infinitum is consistent with the strategy~$\sigma$ and violates the energy condition.
This, however, contradicts~$\sigma$ being a winning strategy from~$v_0$ for Player~$0$.
\end{proof}

Moreover, Chatterjee and Doyen gave an upper bound on the complexity of solving energy parity games, which was recently supplemented by Daviaud et al.~\cite{DaviaudJurdzinskiLazic18} with an algorithm solving them in pseudo-quasi-polynomial time.

\begin{propC}[\cite{ChatterjeeDoyen12,DaviaudJurdzinskiLazic18}]%
\label{prop:energyparity:complexity}
	The following problem is in~$\np \cap \conp$ and can be solved in pseudo-quasi-polynomial time:
	``Given an energy parity game~$\game$ and a vertex~$v$ in~$\game$, does Player~$0$ win~$\game$ from~$v$?''
\end{propC}



\subsection{From Bounded Parity Games with Weights to Energy Parity Games}%
\label{subsec:bounded-parity-weights-solve:bpw2ep}
Let $\game = (\arena, \bcp(\col, \weight))$ be a bounded parity game with weights with vertex set~$V$.
Without loss of generality, we assume~$\col(v) \geq 2$ for all~$v \in V$.
We construct, for each vertex~$v^*$ of~$\arena$, an energy parity game~$\game_{v^*}$ with the following property:

\begin{quotation}
Player~$1$ wins~$\game_{v^*}$ from some designated vertex induced by~$v^*$ if and only if she is able to unbound the amplitude for the request of the initial vertex of the play when starting from~$v^*$.
\end{quotation}

\noindent
This construction is the technical core of the fixed-point algorithm that solves bounded parity games with weights via solving energy parity games.

The main obstacle towards this is that, in the bounded parity game with weights~$\game$, Player~$1$ may win by unbounding the amplitude for a request from above or from below, while she can only win the energy parity game~$\game_{v^*}$ by unbounding the costs from below.
We model this in~$\game_{v^*}$ by constructing two copies of~$\arena$.
In one of these copies the edge weights are copied from~$\game$, while they are inverted in the other copy.
We allow Player~$1$ to switch between these copies arbitrarily.
To compensate for Player~$1$'s power to switch, Player~$0$ can increase the energy level in the resulting energy parity game during each switch.

First, we define the set of polarities $P = \set{+,-}$ as well as $\compl{+} = {-}$ and~$\compl{-} = {+}$.
Given a vertex~$v^*$ of $\arena$, define the ``polarized'' arena~$\arena_{v^*} = (V', V_0', V_1', E')$ of~$\arena = (V, V_0, V_1, E)$ with
\begin{itemize}
	\item $V' = (V \times P) \cup (E \times P \times \set{0,1}) $,
	\item 	$V_i' = (V_i \times P) \cup (E \times P \times \set{i}) $ for $ i \in \set{0,1} $, and
	\item  $E'$ contains the following edges for every edge~$e = (v,v') \in E$ with  $\col(v) \notin \answer{\col(v^*)}$ and every polarity~$p \in P$:
	 	\begin{itemize}
	 		\item $((v, p), (e, p, 1))$: The player whose turn it is at $v$ picks a successor~$v'$. The edge~$e = (v,v')$ is stored as well as the polarity~$p$.
	 		\item $((e,p,1),(v',p))$: Then, Player~$1$ can either keep the polarity~$p$ unchanged and execute the move to $v'$, or
	 		\item $((e,p,1),(e,p,0))$: she decides to change the polarity, and another auxiliary vertex is reached.
	 		\item $((e,p,0),(e,p,0))$: If the polarity is to be changed, then Player~$0$ is able to use a self-loop to increase the energy level (see below), before
	 		\item $((e,p,0), (v', \compl{p}))$: he can eventually complete the polarity switch by moving to $v'$.
	 		\end{itemize}
	 	\item Furthermore, for every vertex~$v$ with $\col(v) \in \answer{\col(v^*)}$ and every polarity~$p \in P$, $E'$ contains the self-loop~$((v,p),(v,p))$.\footnote{This definition introduces some terminal vertices, i.e., those of the form $((v,v'),p,i)$ with $\col(v) \in \answer{\col(v^*)}$. However, these vertices also have no incoming edges. Hence, to simplify the definition, we just ignore them.}
\end{itemize}
Thus, a play in $\arena_{v^*}$ simulates a play in $\arena$, unless Player~$0$ stops the simulation by using the self-loop at a vertex of the form~$(e,p,0)$ ad infinitum, and unless an answer to $\col(v^*)$ is reached.
We define the coloring and the weighting for $\arena_{v^*}$ so that Player~$0$ loses in the former case and wins in the latter case. Furthermore, the coloring is defined so that all simulating plays that are not stopped have the same color sequence as the simulated play (save for irrelevant colors on the auxiliary vertices in $E \times P \times \set{0,1}$).
Hence, we define
\[
	\col_{v^*}(v) = \begin{cases}
 		\col(v') &\text{if } v = (v', p) \text{ with } v' \notin \answer{\col(v^*)} \enspace ,\\
 		0 &\text{if } v = (v', p) \text{ with } v' \in \answer{\col(v^*)} \enspace ,\\
 		1 & \text{otherwise} \enspace .
	\end{cases}
\]
As desired, due to our assumption that~$\col(v) \geq 2$ for all~$v \in V$, the vertices from $E \times P \times \set{0,1}$ do not influence the maximal color visited infinitely often during a play, unless Player~$0$ opts to remain in some~$(e, p, 0)$ ad infinitum (and thereby violates the parity condition) or an answer to the color of $v^*$ is reached (and thereby satisfies the parity condition).

Moreover, recall that our aim is to allow Player~$1$ to choose the polarity of edges by switching between the two copies of~$\arena$ occurring in~$\arena_{v^*}$.
Intuitively, Player~$1$ should opt for positive polarity in order to unbound the costs incurred by the request posed by~$v^*$ from above, while she should opt for negative polarity in order to unbound these costs from below.
Since it is, broadly speaking, beneficial for Player~$1$ to move along edges of negative weight in an energy parity game, we negate the weights of edges in the copy of~$\arena$ with positive polarity.
Thus, we define
\[
	\weight_{v^*}(e) = \begin{cases}
		-\weight(v, v') &\text{if } e = ((v, +), ((v,v'), +, 1)) \enspace , \\
		\weight(v, v') &\text{if } e = ((v, -), ((v,v'), -, 1)) \enspace , \\
		1 &\text{if } e = ((e,p,0),(e,p,0)) \enspace , \\
		0 &\text{otherwise}  \enspace .
	\end{cases}
\]
This definition implies that the self-loops at vertices of the form~$(v,p)$ with $\col(v) \in \answer{\col(v^*)}$ have weight zero. Combined with the fact that these vertices have color zero, this allows Player~$0$ to win $\game_{v^*}$ by reaching such a vertex. Intuitively, answering the request posed at $v^*$ is beneficial for Player~$0$. In particular, if $\col(v^*)$ is even, then Player~$0$ wins $\game_{v^*}$ trivially from $(v^*, p)$, as we then have $\col(v^*) \in \answer{\col(v^*)}$.

Finally, define the energy parity game~$\game_{v^*} = ( \arena_{v^*}, \col_{v^*}, \weight_{v^*} )$. In the following, we are only interested in plays starting in vertex~$(v^*,+)$ in $\game_{v^*}$.

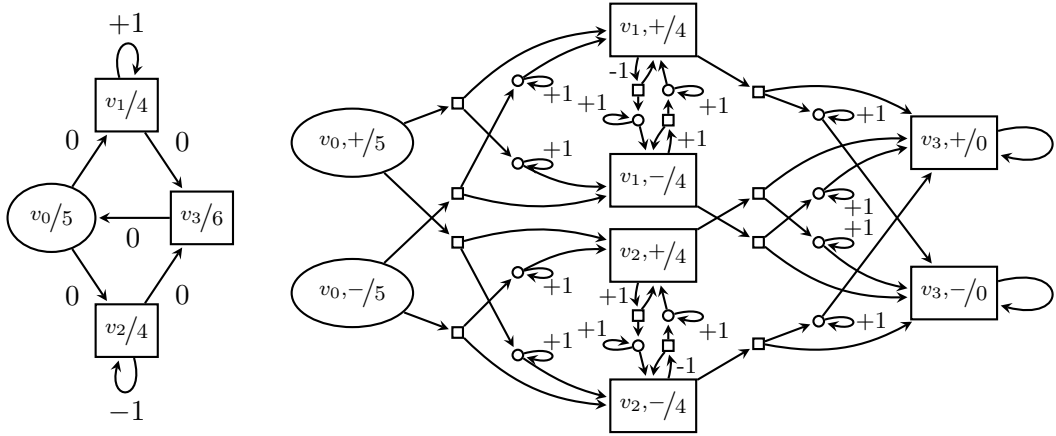
\begin{figure}[ht]
\begin{tikzpicture}
	\begin{scope}
		\node[p0] (v0) at (0,0) {\parnode{v_0}{5}};
		\node[p1] (v1) at (1,1.5) {\parnode{v_1}{4}};
		\node[p1] (v2) at (1,-1.5) {\parnode{v_2}{4}};
		\node[p1] (v3) at (2,0) {\parnode{v_3}{6}};
		
		\path
			(v0) edge node[anchor=south east] {$0$} (v1) edge node[anchor=north east] {$0$} (v2)
			(v1) edge [loop above] node[anchor=south] {$+1$} (v1) edge node[anchor=south west] {$0$} (v3)
			(v2) edge [loop below] node[anchor=north] {$-1$} (v2) edge node[anchor=north west] {$0$} (v3)
			(v3) edge node[anchor=north] {$0$} (v0);
	\end{scope}
	
	\begin{scope}[shift={(4,1)}]
		\begin{scope}
			\node[p0] (v0-p) at (0,0)     {\parnode{v_0, +}{5}};
			\node[p1] (v1-p) at (4,1.5)   {\parnode{v_1, +}{4}};
			\node[p1] (v2-p) at (4,-1.5)  {\parnode{v_2, +}{4}};
			\node[p1] (v3-p) at (8,0)     {\parnode{v_3, +}{0}};
		\end{scope}
		\begin{scope}[shift={(0,-2)}]
			\node[p0] (v0-n) at (0,0)     {\parnode{v_0, -}{5}};
			\node[p1] (v1-n) at (4,1.5)   {\parnode{v_1, -}{4}};
			\node[p1] (v2-n) at (4,-1.5)  {\parnode{v_2, -}{4}};
			\node[p1] (v3-n) at (8,0)     {\parnode{v_3, -}{0}};
		\end{scope}	
		
		\tikzstyle{1-trans}=[draw,thick,minimum size=4pt,inner sep=0pt]
		\tikzstyle{0-trans}=[1-trans,ellipse]
		
		\newcommand{\increment}{\footnotesize{+1}}
		\newcommand{\decrement}{\footnotesize{-1}}
		
		
		\node[1-trans] (v0-v1-p-1) at ($(v0-p) ! .35 ! (v1-p) - (0,0)$) {};
		\node[1-trans] (v0-v1-n-1) at ($(v0-n) ! .35 ! (v1-n) + (0,.8)$) {};
		\node[0-trans] (v0-v1-p-0) at ($(v0-p) ! .55 ! (v1-n)$) {};
		\node[0-trans] (v0-v1-n-0) at ($(v0-p) ! .55 ! (v1-p)$) {};
		
		\path
			(v0-p) edge (v0-v1-p-1)
			(v0-v1-p-1) edge[bend left=20] (v1-p) edge (v0-v1-p-0)
			(v0-v1-p-0) edge[loop right] node[pos=.2,inner sep=0,anchor=south west] {\increment} (v0-v1-p-0) edge[bend right=15] (v1-n);
		\path
			(v0-n) edge (v0-v1-n-1)
			(v0-v1-n-1) edge[bend right=15] (v1-n) edge (v0-v1-n-0)
			(v0-v1-n-0) edge[loop right] node[pos=.8,inner sep=0,anchor=north west] {\increment} (v0-v1-n-0) edge[bend left=10] (v1-p);
			
		
		\node[1-trans] (v0-v2-p-1) at ($(v0-p) ! .35 ! (v2-p) - (0,.8)$) {};
		\node[1-trans] (v0-v2-n-1) at ($(v0-n) ! .35 ! (v2-n) + (0,.0)$) {};
		\node[0-trans] (v0-v2-p-0) at ($(v0-n) ! .55 ! (v2-n)$) {};
		\node[0-trans] (v0-v2-n-0) at ($(v0-n) ! .55 ! (v2-p)$) {};
		
		\path
			(v0-p) edge (v0-v2-p-1)
			(v0-v2-p-1) edge[bend left=15] (v2-p) edge (v0-v2-p-0)
			(v0-v2-p-0) edge[loop right] node[pos=.2,inner sep=0,anchor=south west] {\increment} (v0-v2-p-0) edge[bend right=10] (v2-n);
		\path
			(v0-n) edge (v0-v2-n-1)
			(v0-v2-n-1) edge[bend right=20] (v2-n) edge (v0-v2-n-0)
			(v0-v2-n-0) edge[loop right] node[pos=.8,inner sep=0,anchor=north west] {\increment} (v0-v2-p-0) edge[bend left=15] (v2-p);
			
		
		\node[1-trans] (v1-v3-p-1) at ($(v1-p) ! .35 ! (v3-p) - (0,.3)$) {};
		\node[1-trans] (v1-v3-n-1) at ($(v1-n) ! .35 ! (v3-n) - (0,.3)$) {};
		\node[0-trans] (v1-v3-p-0) at ($(v1-p) ! .55 ! (v3-p) - (0,.3)$) {};
		\node[0-trans] (v1-v3-n-0) at ($(v2-p) ! .55 ! (v3-p)$) {};
		
		\path
			(v1-p) edge[] (v1-v3-p-1)
			(v1-v3-p-1) edge[bend left=20] (v3-p) edge (v1-v3-p-0)
			(v1-v3-p-0) edge[loop right] node[inner sep=1pt,anchor=west] {\increment} (v1-v3-p-0) edge (v3-n);
		\path
			(v1-n) edge (v1-v3-n-1)
			(v1-v3-n-1) edge[bend right=20] (v3-n) edge (v1-v3-n-0)
			(v1-v3-n-0) edge[loop right] node[pos=.8,inner sep=0,anchor=north west] {\increment} (v1-v3-n-0) edge[bend left=15] (v3-p);
			
		
		\node[1-trans] (v2-v3-p-1) at ($(v2-p) ! .35 ! (v3-p) + (0,.3)$) {};
		\node[1-trans] (v2-v3-n-1) at ($(v2-n) ! .35 ! (v3-n) + (0,.3)$) {};
		\node[0-trans] (v2-v3-p-0) at ($(v1-n) ! .55 ! (v3-n)$) {};
		\node[0-trans] (v2-v3-n-0) at ($(v2-n) ! .55 ! (v3-n) + (0,.3)$) {};
		
		\path
			(v2-p) edge (v2-v3-p-1)
			(v2-v3-p-1) edge[bend left=20] (v3-p) edge (v2-v3-p-0)
			(v2-v3-p-0) edge[loop right] node[pos=.2,inner sep=0,anchor=south west] {\increment} (v2-v3-p-0) edge[bend right=15] (v3-n);
		\path
			(v2-n) edge[] (v2-v3-n-1)
			(v2-v3-n-1) edge[bend right=20] (v3-n) edge (v2-v3-n-0)
			(v2-v3-n-0) edge[loop right] node[inner sep=1pt,anchor=west] {\increment} (v2-v3-n-0) edge (v3-p);
						
		\node[1-trans] (v1-v1-p-1) at ($(v1-p) ! .5 ! (v1-n) - (.2,-.2)$) {};
		\node[1-trans] (v1-v1-n-1) at ($(v1-p) ! .5 ! (v1-n) + (.2,-.2)$) {};
		\node[0-trans] (v1-v1-p-0) at ($(v1-p) ! .5 ! (v1-n) - (.2,.2)$) {};
		\node[0-trans] (v1-v1-n-0) at ($(v1-p) ! .5 ! (v1-n) + (.2,.2)$) {};
		
		\path
			(v1-p) edge[bend right=15] node[inner sep=2pt,anchor=east] {\decrement} (v1-v1-p-1)
			(v1-v1-p-1) edge[bend right=15] (v1-p) edge (v1-v1-p-0)
			(v1-v1-p-0) edge[loop left] node[pos=.75,inner sep=1pt,anchor=south east] {\increment} (v1-v1-p-0) edge (v1-n);
		\path
			(v1-n) edge[bend right=15] node[inner sep=2pt,anchor=west] {\increment} (v1-v1-n-1)
			(v1-v1-n-1) edge[bend right=15] (v1-n) edge (v1-v1-n-0)
			(v1-v1-n-0) edge[loop right] node[pos=.75,inner sep=1pt,anchor=north west] {\increment} (v1-v1-n-0) edge (v1-p);
			
		\node[1-trans] (v2-v2-p-1) at ($(v2-p) ! .5 ! (v2-n) - (.2,-.2)$) {};
		\node[1-trans] (v2-v2-n-1) at ($(v2-p) ! .5 ! (v2-n) + (.2,-.2)$) {};
		\node[0-trans] (v2-v2-p-0) at ($(v2-p) ! .5 ! (v2-n) - (.2,.2)$) {};
		\node[0-trans] (v2-v2-n-0) at ($(v2-p) ! .5 ! (v2-n) + (.2,.2)$) {};
		
		\path
			(v2-p) edge[bend right=15] node[inner sep=2pt,anchor=east] {\increment} (v2-v2-p-1)
			(v2-v2-p-1) edge[bend right=15] (v2-p) edge (v2-v2-p-0)
			(v2-v2-p-0) edge[loop left] node[pos=.75,inner sep=1pt,anchor=south east] {\increment} (v2-v2-p-0) edge (v2-n);
		\path
			(v2-n) edge[bend right=15] node[inner sep=2pt,anchor=west] {\decrement} (v2-v2-n-1)
			(v2-v2-n-1) edge[bend right=15] (v2-n) edge (v2-v2-n-0)
			(v2-v2-n-0) edge[loop right] node[pos=.75,inner sep=1pt,anchor=north west] {\increment} (v2-v2-n-0) edge (v2-p);
			
		\path
			(v3-p) edge [loop right] (v3-p)
			(v3-n) edge [loop right] (v3-n);
		
	\end{scope}

\end{tikzpicture}
\caption{A bounded parity game with weights~$\game$ (left) and the associated energy parity game~$\game_{v_0}$ (right). The unnamed vertices of Player~$1$ (Player~$0$) are of the form $((v,v'),p,1)$ (of the form~$((v,v'),p,0)$) when between the vertices~$(v,p)$ and $(v',p')$. All missing edge weights in $\game_{v_0}$ are $0$.}%
\label{fig:epexample}
\end{figure}

\begin{exa}
Consider the bounded parity game with weights depicted on the left hand side of Figure~\ref{fig:epexample} and the associated energy parity game~$\game_{v_0}$ on the right side. First, let us note that all  other $\game_v$ for $v \neq v_0$ are trivial in that they all consist of a single vertex (reachable from $(v,+)$), which has even color with a self-loop of weight zero. Hence, Player~$0$ wins each of these games from $(v,+)$.

Player~$1$ wins $\game$ from $v_0$, where a request for color~$5$ is opened, which is then kept unanswered with infinite cost by using the self-loop at $v_1$ or $v_2$ ad infinitum, depending on which successor Player~$0$ picks.

We show that Player~$1$ wins $\game_{v_0}$ from $(v_0, +)$: the outgoing edges of $(v_0, +)$ correspond to picking the successor $v_1$ or $v_2$ as in $\game$. Before this is executed, however, Player~$1$ gets to pick the polarity of the successor: she should pick $+$ for $v_1$ and $-$ for $v_2$. Now, Player~$0$ may  use the self-loop at her \myquot{tiny} vertices ad infinitum. These vertices have color one, i.e., Player~$1$ wins the resulting play. Otherwise, we reach the vertex~$(v_1, +)$ or $(v_2,-)$. From both vertices, Player~$1$ can enforce a loop of negative weight, which allows him to win by violating the energy condition.
\end{exa}

Note that the winning strategy for Player~$1$ for $\game$ from $v_0$ is very similar to that for her for $\game_{v_0}$ from $(v_0,+)$. We show that one direction holds in general: A winning strategy for Player~$0$ for $\game_{v}$ from $(v,+)$ is \myquot{essentially} one for him in $\game$ from $v$.

The other direction does, in general, not hold.
This can be seen by extending $\game$ in Figure~\ref{fig:epexample} by a vertex~$v_{-1}$ of color $3$ with a single outgoing edge to $v_0$ (with arbitrary weight) and no incoming edges.
We call this extended bounded parity game with weights~$\game'$.
In the resulting energy parity game~$\game'_{v_{-1}}$, vertices of the form $(v_i,p)$ with $i \in \set{1,2}$ in $\game_{v_{-1}}$ are winning sinks for Player~$0$.
Also, Player~$0$ has a strategy to ensure that such a sink is reached when starting in $(v_{-1}, +)$.
Hence, he wins $\game'_{v_{-1}}$ from $(v_{-1},+)$.
However, he does not win the extended bounded parity game with weights~$\game'$ from $v_{-1}$, as Player~$1$ can still win by remaining in either~$v_1$ or~$v_2$ ad infinitum.
By doing this, she keeps the request of color~$5$, which is opened when visiting~$v_0$ at the second position of every play starting in~$v_{-1}$, unanswered with infinite cost.

The reason for the other direction failing is the special role the initial request of the vertex~$v$ inducing $\game_v$ plays in the construction: It is the request Player~$1$ aims to keep unanswered with infinite cost.
To overcome this and to complete our construction, we show a statement reminiscent of Lemma~\ref{lem:parity-to-bounded-parity:reduction}: If Player~$0$ wins $\game_v$ from $(v,+)$ for every $v$, then she also wins $\game$ from every vertex.
With this relation at hand, one can again construct a fixed-point algorithm solving bounded parity games with weights using an oracle for solving energy parity games that is very similar to Algorithm~\ref{algorithm_fixpoint}.

Formally, we have the following lemma, which forms the technical core of our algorithm that solves bounded parity games with weights by solving energy parity games.

\begin{lem}%
\label{lem:bounded-cost-parity:reduction}
Let~$\game$ be a bounded parity game with weights with vertex set~$V$.
\begin{enumerate}
	\item\label{lem:bounded-cost-parity:reduction:player1} Let $v^* \in V$. If Player~$1$ wins~$\game_{v^*}$ from~$(v^*, +)$, then~$v^* \in \winreg_1(\game)$.
	\item\label{lem:bounded-cost-parity:reduction:player0} If Player~$0$ wins $\game_{v^*}$ from $(v^*,+)$ for all~$v^* \in V$, then $\winreg_1(\game) = \emptyset$.
\end{enumerate}
\end{lem}

\noindent
Before we prove this lemma, we first note that it is the main building block for the algorithm that solves bounded parity games with weights by repeatedly solving energy parity games, which is very similar to Algorithm~\ref{algorithm_fixpoint}.
Indeed, we just swap the roles of the players:
We compute $1$-attractors instead of $0$-attractors and we change the definition of $X_k$. Hence, we obtain Algorithm~\ref{algorithm_fixpoint2}.

\begin{algorithm}
\begin{algorithmic}[1]
 \REQUIRE{Bounded parity game with weights~$\game$ with arena~$\arena$, coloring~$\col$ and weighting~$\weight$}
 \STATE{$k = 0$}; $W_1^k = \emptyset$; $\arena_k = \arena$
 \REPEAT%
 \STATE{$k = k + 1$}
\STATE{$X_k = \set{ v^* \mid \text{Player~$1$ wins the energy parity game~$({(\arena_{k-1})}_{v^*}, \col_{v^*},\weight_{v^*})$ from $(v^*,+)$} }$}
 \STATE{$W_1^k = W_1^{k-1} \cup \att{1}{\arena_{k-1}}{X_k}$}
 \STATE{$\arena_k = \arena_{k-1} \setminus \att{1}{\arena_{k-1}}{X_k}$}
 \UNTIL{$X_k = \emptyset$}
 \RETURN{$W_1^k$}
\end{algorithmic}
\caption{A fixed-point algorithm computing $\winreg_1(\arena, \bcp(\col, \weight))$.}%
\label{algorithm_fixpoint2}
\end{algorithm}

Algorithm~\ref{algorithm_fixpoint2} terminates after solving at most a quadratic number of energy parity games of polynomial size. 
Furthermore, the proof of correctness is analogous to the one for Algorithm~\ref{algorithm_fixpoint}, relying on Lemma~\ref{lem:bounded-cost-parity:reduction}. We only need two further properties: the $1$-extendability of $\bcp(\col,\weight)$, and an assertion that $\att{1}{\arena_{k-1}}{X_k}$ is a trap for Player~$0$ in~$\arena_{k-1}$.
Both are easy to verify.

After plugging Algorithm~\ref{algorithm_fixpoint2} into Algorithm~\ref{algorithm_fixpoint}, Proposition~\ref{prop:energyparity:complexity} yields our main theorem, settling the complexity of solving parity games with weights.

\begin{thm}%
\label{thm:bounded-parity-with-weights:complexity}
The following problem is in $\np \cap \conp$ and can be solved in pseudo-quasi-polynomial time:
\begin{quotation}
	``Given a parity game with weights~$\game$ and a vertex~$v$ in~$\game$, does Player~$0$ win~$\game$ from~$v$?''
\end{quotation}
\end{thm}

\noindent
It remains to prove Lemma~\ref{lem:bounded-cost-parity:reduction}.
We do so in the following section.

\subsection{Proof of Lemma~\ref{lem:bounded-cost-parity:reduction}}%
\label{subsec:bounded-parity-weights-solve:proof}
We prove the two assertions of Lemma~\ref{lem:bounded-cost-parity:reduction} separately from each other:
We first show Item~\ref{lem:bounded-cost-parity:reduction:player1} of Lemma~\ref{lem:bounded-cost-parity:reduction}, before continuing to show Item~\ref{lem:bounded-cost-parity:reduction:player0}.
In order to prepare for this, however, we first introduce some notation.
Let~$v^* \in V$ and consider~$\game_{v^*}$.
We distinguish three types of plays in~$\game_{v^*}$:
\begin{description}
	\item[Type -1] Plays that have a suffix~${(e, p, 0)}^\omega$ for some~$e \in E$ and some~$p \in P$.
	\item[Type \hphantom{-}0] Plays that visit infinitely many vertices from both~$V \times P$ and~$E \times P \times \set{0,1}$.
	\item[Type \hphantom{-}1] Plays that have a suffix~${(v, p)}^\omega$. Note that this implies~$\col(v) \in \answer{\col(v^*)}$.
\end{description}

Clearly, plays of Type~$-1$ are losing for Player~$0$ due to the coloring of~$\game'_{v^*}$ labeling vertices of the form~$(e, p, 0)$ with the odd color one.
Dually, plays of Type~$1$ are losing for Player~$1$, since~$\col(v) \in \answer{\col(v^*)}$ implies that~$(v, p)$ carries color zero and its only outgoing edge is a self-loop of weight~$0$.
We formalize this observation in the following remark.

\begin{numrem}%
\label{rem:types}
	Let~$\rho'$ be a play in~$\game_{v^*}$ that starts in~$(v^*, p)$.
	\begin{enumerate}
		\item\label{rem:types:playerzero} If~$\rho'$ is consistent with a winning strategy for Player~$0$ from~$(v^*, p)$, then~$\rho'$ is not a play of Type~$-1$.
		\item\label{rem:types:playerone} If~$\rho'$ is consistent with a winning strategy for Player~$1$ from~$(v^*, p)$, then~$\rho'$ is not a play of Type~$1$.
	\end{enumerate}
\end{numrem}

\noindent
In order to remove the added vertices of the form~$E \times P \times \set{0,1}$ from plays in~$\game_{v^*}$, we define the homomorphism~$\unpol \colon {(V')}^* \cup {(V')}^\omega \rightarrow V^* \cup V^\omega$ induced by $\unpol(v, p) = v$ and~$\unpol(e, p, i) = \epsilon$ for~$v \in V$,~$e \in E$,~$p \in P$, and~$i \in \set{0,1}$.
Let~$\rho' \in {(V')}^* \cup {(V')}^\omega$.
We call $\unpol(\rho')$ the \textbf{unpolarization} of~$\rho'$.

\begin{numrem}%
\label{rem:unpolparity}
	Let~$\rho'$ be a play of Type~$0$ in some~$\game_{v^*}$.
	We have~$\rho' \in \parity(\col_{v^*})$ if and only if $\unpol(\rho') \in \parity(\col)$.
\end{numrem}

\subsubsection{Proof of Item~\ref{lem:bounded-cost-parity:reduction:player1} of Lemma~\ref{lem:bounded-cost-parity:reduction}}%
\label{subsec:bounded-parity-weights-solve:proof-player1}
Recall that we need to show that Player~$1$ wins the bounded parity game with weights~$\game$ from $v^*$ if she wins the energy-parity game~$\game_{v^*}$ from $(v^*,+)$.
Thus, let~$\tau_{v^*}$ be a winning strategy for Player~$1$ from~$(v^*, +)$ in~$\game_{v^*}$.
We define a winning strategy~$\tau$ for her from~$v^*$ in~$\game$ such that~$\tau$ mimics the moves made by~$\tau_{v^*}$.
To this end,~$\tau$ keeps track of a play prefix $\game_{v^*}$.
Formally, we define~$\tau$ together with a simulation function~$h$ that satisfies the following invariant:

\begin{quotation}
If~$\pi$ is a nonempty play prefix in~$\arena$ that starts in~$v^*$, is consistent with~$\tau$, and ends in some~$v$,
then~$h(\pi)$ is a play prefix in~$\arena_{v^*}$ that starts in~$(v^*, +)$, is consistent with~$\tau_{v^*}$, and ends in some~$(v, p)$. Furthermore, $\unpol(h(\pi)) = \pi$.
\end{quotation}
Recall that, if $h$ has the properties described above, then, due to the structure of~$\arena_{v^*}$, for each~$\pi$, given~$h(\pi)$, the strategy~$\tau_{v^*}$ prescribes a move to some vertex~$((v, v'), p, 1)$, where~$(v, v') \in E$.
We can mimic this choice by moving to~$v'$ in~$\game$.

We now define~$h$ and~$\tau$ formally and begin with~$h(v^*) = (v^*, +)$, which clearly satisfies the invariant.
Now let~$\pi = v_0 \cdots v_j$ be some nonempty play prefix in~$\arena$ beginning in~$v^*$ and consistent with~$\tau$ such that~$h(\pi)$ is defined.
Due to the invariant,~$h(\pi)$ ends in~$(v_j, p_j)$ for some $p_j \in P$.

If~$v_j \in V_1$, there is a unique vertex~$v_{j+1}$  such that $h(\pi) \cdot ((v_j, v_{j+1}), p_j, 1)$ is consistent with~$\tau_{v^*}$.
We define~$\tau(\pi) = v_{j+1}$.
Such a~$v_{j+1}$ exists, because $(v_j, p_j)$, the last vertex of $h(\pi)$, satisfies $\col(v_j) \notin \answer{\col(v^*)}$ due to the invariant, Item~\ref{rem:types:playerone} of Remark~\ref{rem:types}, and because the answering vertices are sinks.
If, however,~$v_j \in V_0$, then let~$v_{j+1}$ be an arbitrary successor  of~$v_j$ in~$\arena$.
In either case, it remains to define~$h(\pi \cdot v_{j+1})$.

Since we want to simulate the move from~$v_j$ to~$v_{j+1}$ in~$h(\pi \cdot v_{j+1})$, we first move from~$(v_j, p_j)$ to~$((v_j, v_{j+1}), p_j, 1)$.
Moreover, in order to satisfy the invariant, we aim to simulate the play prefix~$\pi \cdot v_{j+1}$ such that~$h(\pi \cdot v_{j+1})$ is consistent with~$\tau_{v^*}$.
This strategy may prescribe for Player~$1$ to either preserve the polarity~$p_j$, or to switch it during the simulated move from~$v_j$ to~$v_{j+1}$.

In the former case, i.e., if~$\tau_{v^*}(h(\pi) \cdot ((v_j, v_{j+1}), p_j, 1)) = (v_{j+1}, p_j)$, we define
\[
	h(\pi \cdot v_{j+1}) = h(\pi) \cdot ((v_j, v_{j+1}), p_j, 1) \cdot (v_{j+1}, p_j) \enspace .
\]
In the latter case, Player~$0$ gets an opportunity to recharge the energy by taking the self-loop of the vertex~$((v_j, v_{j+1}), p_j, 0)$ finitely often.
We opt to let her recover the energy lost so far in the play prefix, i.e., we pick~$c_j = \max\set{0, -\weight(h(\pi) \cdot ((v_j, v_{j+1}), p_j, 1))}$ and define
\[
	h(\pi \cdot v_{j+1}) = h(\pi) \cdot ((v_j, v_{j+1}), p_j, 1) \cdot {((v_j, v_{j+1}), p_j, 0)}^{c_j + 1} \cdot (v_{j+1}, \compl{p_j})
\]
in this case.
Since~$h(\pi \cdot v_{j+1})$ is consistent with~$\tau_{v^*}$ in either case, we satisfy the invariant in either case.
This completes the definition of~$\tau$ and~$h$.

It remains to show that~$\tau$ is indeed winning from~$v^*$ in~$\game$.
To this end, let~$\rho = v_0v_1v_2\cdots$ be a play in~$\arena$ that starts in~$v^*$ and that is consistent with~$\tau$.
We need to show~$\rho \notin \bcp(\col, \weight)$.

Note that $h(v_0 \cdots v_j)$ is a strict prefix of $h(v_0 \cdots v_{}j+1)$ for every $j$.
As each such $h(v_0 \cdots v_j)$ is a play prefix in $\arena_{v^*}$, there is a unique infinite play~$\rho'$ in $\arena_{v^*}$ such that each $h(v_0 \cdots v_j)$ is a prefix of $\rho'$, i.e, $\rho'$ is the limit of the~$h(v_0 \cdots v_j)$ for increasing prefixes~$v_0 \cdots v_j$ of~$\rho$.
Due to the invariant, $\rho'$ starts in $(v^*, +)$ and is consistent with $\tau_{v^*}$.
Moreover, due to the construction of~$h$, we obtain~$\unpol(\rho') = \rho$.
Finally, we have that~$\rho'$ is a play of Type~$0$ in~$\game_{v^*}$.
Hence, due to Remark~\ref{rem:unpolparity},~$\rho$ satisfies the parity condition if and only if~$\rho'$ satisfies the parity condition.

As the play $\rho'$ is consistent with the winning strategy~$\tau_{v^*}$ for Player~$1$, we have $ \rho' \notin \energyparity(\col_{v^*}, \weight_{v^*}) $, i.e.,~$\rho'$ either violates the parity condition or the energy condition.
Hence, as argued above, if~$\rho'$ violates the parity condition, then so does~$\rho$, i.e., $\rho$ is indeed winning for Player~$1$.

Now assume that~$\rho'$ violates the energy condition.
Due to the structure of~$\arena_{v^*}$ and the construction of~$h$ we have
\[ \rho' = \Pi_{j = 0,1,2,\dots} (v_j, p_j)\cdot ((v_j, v_{j+1}), p_j, 1)\cdot {((v_j, v_{j+1}), p_j, 0)}^{m_j} \]
for some~$m_j \in \nats$.
Since~$\rho'$ violates the energy condition, we have $\inf_{j \in \nats} \weight((v_0, p_0) \cdots (v_j, p_j) \cdot ((v_j, v_{j+1}), p_j, 1)) = -\infty$.
The restriction to play prefixes of this form suffices due to the structure of~$\arena_{v^*}$ and, in turn, the structure of~$\rho'$.
Moreover, since Player~$1$ wins~$\game_{v^*}$ from~$(v^*, +)$, the initial vertices~$v^*$ and~$(v^*, +)$ of~$\rho$ and~$\rho'$, respectively, have the same odd color.
Also, as~$\rho'$ is a play of Type~$0$, the request for the color~$\col(v^*)$ is never answered in~$\rho$ or~$\rho'$.
We show that the request for~$\col(v^*)$ in $\rho$ is unanswered with infinite cost, which concludes the proof.

To this end, we split~$\rho'$ into infixes of constant polarity.
Given a vertex~$v = (v', p)$ or~$v = ((v', v''), p, i)$, we call~$p$ the \textbf{polarity} of~$v$.
Let~$\rho' = \mu'_0\mu'_1\mu'_2\cdots$, where each~$\mu'_j$ is a maximal finite (or infinite) infix (or suffix) of~$\rho'$, such that all vertices in~$\mu'_j$ have the same polarity.
We call an infix~$\mu'_j$ of~$\rho'$ an \textbf{equi-polarity infix} (\epi) of~$\rho'$.

Since the polarity remains constant throughout each~$\mu_j'$, Player~$0$ only resets the energy via repeatedly traversing a self-loop of a vertex in~$\arena_{v^*}$ at the last vertex visited in~$\mu'_j$, if at all.
Hence, the energy levels attained during $\mu'_j$ and~$\unpol(\mu_j')$ are closely related.

\begin{numrem}%
\label{rem:induced-epra:epi:equivalence}
	Let~$\mu'$ be an \epi beginning in~$(v_j, p_j)$ and let~$\mu = \unpol(\mu') = v_{j}v_{j+1}v_{j+2}\cdots$.
	For each $j'$ with~$j \leq j' < j + \card{\mu}$, we have~$\card{\weight(v_j \cdots v_{j'})} = \card{\weight((v_j, p_j) \cdots (v_{j'}, p_{j'}))}$.
\end{numrem}

In particular, Remark~\ref{rem:induced-epra:epi:equivalence}, the structure of~$\arena_{v^*}$, and the definition of~$h$ imply that we have~$\ampl(\unpol(\mu')) = \ampl(\mu')$ for all \epis~$\mu'$ of~$\rho'$.
Thus, if there exist only finitely many~\epis of~$\rho'$, let~$\mu'$ be the infinite final~\epi of~$\rho'$, let~$\mu = \unpol(\mu')$, and note that, due to~$\ampl(\rho') = \infty$, we have~$\ampl(\mu') = \infty$.
Due to Remark~\ref{rem:induced-epra:epi:equivalence}, we obtain~$\ampl(\mu) = \infty$, which implies that the request posed at the initial position of~$\rho$ is unanswered with infinite cost due to the reasoning above, as~$\mu$ is a suffix of~$\rho$.

If, however, there exist infinitely many~\epis of~$\rho'$, assume towards a contradiction that the cost of answering the request posed at the initial position of~$\rho$ is finite.
By construction of~$\rho'$, the energy level is nonnegative at the end of each \epi.
Since~$\rho'$ violates the energy condition, for each bound~$b \in \nats$ there exists an \epi~$\mu'$ of~$\rho'$ with a prefix of weight strictly smaller than~$-b$.
We obtain~$\ampl(\unpol(\mu')) > b$ via Remark~\ref{rem:induced-epra:epi:equivalence}.
This contradicts the cost of answering the request posed at the initial position of~$\rho$ being bounded and concludes the proof of Item~\ref{lem:bounded-cost-parity:reduction:player1} of Lemma~\ref{lem:bounded-cost-parity:reduction}.


\subsubsection{Proof of Item~\ref{lem:bounded-cost-parity:reduction:player0} of Lemma~\ref{lem:bounded-cost-parity:reduction}}%
\label{subsec:bounded-parity-weights-solve:proof-player0}
To prove Item~\ref{lem:bounded-cost-parity:reduction:player0} of Lemma~\ref{lem:bounded-cost-parity:reduction}, we show that Player~$0$ wins the bounded parity game with weights~$\game$ from every vertex, if he wins each energy-parity game~$\game_{v^*}$ from $(v^*,+)$.
 To this end, we construct a strategy~$\sigma$ for Player~$0$ in~$\game$ that is winning for him from each vertex~$v \in V$.
As winning regions are disjoint, this implies the desired result.

For each energy parity game~$\game_v = (\arena_v, \col_v, \weight_v)$  we have~$n' = \card{\arena_v} \in \bigo(\card{\arena}^2)$, we have~$d' = \card{\col_v(V')} = \card{\col(V)} + 2$, and we have~$W' = \max ( \weight(E') ) = \max ( \weight(E) \cup \set{1} )$, where~$E$ and~$E'$ are the sets of edges in~$\arena$ and the~$\arena_v$, respectively.
Note that the values~$n'$,~$d'$, and~$W'$ of~$\game_v$ are independent of the vertex~$v$, which explains our notation.
Due to the assumption of the statement and Proposition~\ref{prop:energyparity:equivalence}, for each~$v \in V$, there exists a finite-state strategy~$\sigma_{v}$ with at most~$n'd'W'$ states that is winning for Player~$0$ from~$(v, +)$ in~$\game_{v}$.

We construct the winning strategy~$\sigma$ for Player~$0$ in~$\game$ by ``stitching together'' the individual~$\sigma_{v}$.
To this end, given a play prefix, we identify the request which should be answered most urgently.
Say this request was opened by visiting vertex~$v$.
The strategy~$\sigma$ then mimics the moves made by~$\sigma_{v}$ when starting in~$(v, +)$.
Once the request for~$\col(v)$ is answered,~$\sigma$ makes arbitrary moves until a new request is opened.

Formally, given a play prefix~$\pi = v_0\cdots v_j$, we say that a request for color~$c$ is open in~$\pi$ if there exists a position~$j'$ with~$0 \leq j' \leq j$ such that~$\col(v_{j'}) = c$ and, for all positions~$j''$ with~$j' \leq j'' \leq j$, we have~$\col(v_{j''}) \notin \answer{\col(v_{j'})}$.
Clearly there is never an open request for an even color.

Due to monotonicity, answering an open request of color~$c$ also answers all smaller open requests. Hence, the most relevant request, i.e., the largest color with an open request, is of special interest during a play. As alluded to above, it is this color that guides the strategy we are about to construct. To define it formally, we need to introduce some notation to refer to the position where the most relevant request has been opened.

If there is no open request in~$\pi$, the position of the most relevant request is undefined and we write~$\mrr(\pi) = \bot$.
Otherwise, let~$c$ be the maximal color for which there is an open request in~$\pi$.
We define~$\mrr(\pi)$ as the smallest position~$j'$, such that the request for color~$c$ is open in all prefixes of~$\pi$ of length greater than~$j'$.

As an example, consider the play prefix shown in Figure~\ref{fig:relevant-requests:example} using the notation \parnode{v}{\col(v)}.
We mark a position~$j$ with solid background if~$\mrr(v_0\cdots v_j) = j$ and with dashed background if~$\mrr(v_0 \cdots v_j) = \bot$.
Otherwise, i.e., if~$\bot \neq \mrr(v_0\cdots v_j) < j$, we leave~$j$ unmarked.
For those positions, $\mrr(v_0\cdots v_j)$ is equal to the largest (i.e., last visited) earlier position marked in solid background.
We furthermore denote the value of~$\mrr(v_0\cdots v_j)$ by an arrow going from position~$j$ to position~$\mrr(v_0\cdots v_j)$.

\begin{figure}
\centering
\begin{tikzpicture}
	\foreach \x in {3,6,11} { 
		\node[inner sep = 6pt, fill=lightgray,draw=black,minimum width=.75cm,minimum height=.75cm] (background-\x) at (\x, 0) {};
	}
	\foreach \x in {0,1,2,8,9,10} { 
		\node[inner sep = 6pt,draw=black,pattern=north east lines,pattern color=lightgray,minimum width=.75cm,minimum height=.75cm] (background-\x) at (\x, 0) {};
	}

	\foreach \label [count=\x from 0] in {0,0,2,1,0,1,3,1,4,0,2,1,1} {
		\node (col-\x) at (\x, 0) {\parnode{v_{\x}}{\label}};
	}
	\node[] (dots) at (13,0) {$\cdots$};
	
	\path
		(col-3) edge [loop above] (col-3)
		(col-4.north) edge [bend right] ($(col-3.north) + (.3cm,0)$)
		(col-5.north) edge [bend right] ($(col-3.north) + (.15cm,0)$)
		(col-6) edge [loop above] (col-6)
		(col-7.north) edge [bend right] ($(col-6.north) + (.2cm,0)$)
		(col-11) edge [loop above] (col-11)
		(col-12.north) edge [bend right] ($(col-11.north) + (.2cm,0)$);
	
	\path[draw,decorate,decoration={brace,amplitude=5pt,mirror},transform canvas={yshift=-3pt}]
		(col-0.south west) -- node[anchor=north,transform canvas={yshift=-5pt}] {$\mu_0$} (col-0.south east);
	\path[draw,decorate,decoration={brace,amplitude=5pt,mirror},transform canvas={yshift=-3pt}]
		(col-1.south west) -- node[anchor=north,transform canvas={yshift=-5pt}] {$\mu_1$} (col-1.south east);
	\path[draw,decorate,decoration={brace,amplitude=5pt,mirror},transform canvas={yshift=-3pt}]
		(col-2.south west) -- node[anchor=north,transform canvas={yshift=-5pt}] {$\mu_2$} ({col-2}.south east);
	\path[draw,decorate,decoration={brace,amplitude=5pt,mirror},transform canvas={yshift=-3pt}]
		(col-3.south west) -- node[anchor=north,transform canvas={yshift=-5pt}] {$\mu_3$} ({col-5}.south east);
	\path[draw,decorate,decoration={brace,amplitude=5pt,mirror},transform canvas={yshift=-3pt}]
		(col-6.south west) -- node[anchor=north,transform canvas={yshift=-5pt}] {$\mu_4$} ({col-8}.south east);
	\path[draw,decorate,decoration={brace,amplitude=5pt,mirror},transform canvas={yshift=-3pt}]
		(col-9.south west) -- node[anchor=north,transform canvas={yshift=-5pt}] {$\mu_5$} ({col-9}.south east);
	\path[draw,decorate,decoration={brace,amplitude=5pt,mirror},transform canvas={yshift=-3pt}]
		(col-10.south west) -- node[anchor=north,transform canvas={yshift=-5pt}] {$\mu_6$} ({col-10}.south east);
	
	\begin{scope}
		\clip ({col-11}.south west) rectangle ($({col-12}.south) + (.5,-20pt)$);
		\path[draw,decorate,decoration={brace,amplitude=5pt,mirror},transform canvas={yshift=-3pt}]
		(col-11.south west) -- node[anchor=north,transform canvas={yshift=-5pt}] {$\mu_7$} ($({col-11}.south east) + (1.5,0)$);
	\end{scope}
	\begin{scope}
		\clip ($({col-12}.south) + (.5,0)$) rectangle ($({col-12}.south) + (1.5,-20pt)$);
		\path[draw,dashed,decorate,decoration={brace,amplitude=5pt,mirror},transform canvas={yshift=-3pt}]
		(col-11.south west) -- node[anchor=north,transform canvas={yshift=-5pt}] {$\mu_7$} ($({col-11}.south east) + (1.5,0)$);
	\end{scope}

	\node at ($({col-12}.south) - (0,8pt)$) {};
\end{tikzpicture}
\caption{A play~$\rho$, its induced color sequence, its most relevant requests, and the \esis of~$\rho$.}%
\label{fig:relevant-requests:example}
\end{figure}
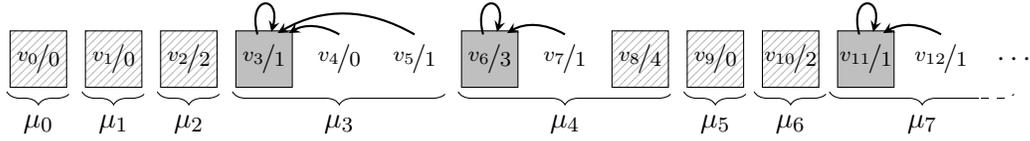

In order to leverage the moves made by the strategies for the~$\game_{v}$ in~$\game$, we need to simulate play prefixes in the latter game in the former ones.
To this end, we again define~$\sigma$ together with a simulation function~$h$.
This function~$h$ maps a play prefix consistent with~$\sigma$ to a sequence of vertices from~$V'$ (not necessarily a play prefix) such that we are able to leverage the choices made by the~$\sigma_{v}$ in order to define~$\sigma$.
Our aim is to define~$h$ such that it satisfies the following invariant:

\begin{quotation}
	Let~$\pi = v_0 \cdots v_j$ be a play prefix in~$\arena$ consistent with~$\sigma$.
	Then~$h(\pi)$ ends in some~$(v_j, p_j)$.
	Moreover, if~$\mrr(\pi) = j' \neq \bot$, then~$h(\pi)$ has a (unique) suffix~$\pi' = (v_{j'}, +)(v_{j'+1}, p_{j'+1}) \cdots (v_j, p_j)$ that is consistent with~$\sigma_{v_{j'}}$ and satisfies
	$\unpol(\pi') = v_{j'} \cdots v_j$.
\end{quotation}

\noindent
We define~$h$ and~$\sigma$ inductively and begin with~$h(v) = (v, +)$ for each~$v \in V$, which clearly satisfies the invariant.
Now let~$\pi = v_0\cdots v_j$ be a nonempty play prefix in~$\arena$ and consistent with~$\sigma$ such that $ h(\pi) $ is defined.
We again first determine a successor of~$v_j$, defining~$\sigma(\pi)$ along the way if~$v_j$ is a vertex of Player~$0$.

If~$v_j \in V_1$, let~$v_{j+1}$ be an arbitrary successor of~$v_j$ in~$\arena$, as Player~$1$ may move to any successor of $v_j$.
If, however,~$v_j \in V_0$, we distinguish two cases based on whether or not~$\mrr(\pi)$ is defined.
If~$\mrr(\pi) = \bot$, again let~$v_{j+1}$ be an arbitrary successor of~$v_j$.
This reflects the fact that Player~$0$ may move to an arbitrary successor if there is no open request.
If, however,~$\mrr(\pi) = j' \neq \bot$, then the invariant of~$h$ yields a suffix~$(v_{j'}, +) \cdots (v_j, p_j)$ of~$h(\pi)$ that is consistent with~$\sigma_{v_{j'}}$.
Let~$v_{j+1}$ be the unique vertex such that~$(v_{j'}, +) \cdots (v_j, p_j) \cdot ((v_j, v_{j+1}), p_j, 1)$ is consistent with~$\sigma_{v_{j'}}$.
Such a vertex~$v_{j+1}$ exists, because the request posed by visiting~$v_{j'}$ is open in~$\pi$ due to~$\mrr(\pi) = j'$ and since the color sequences induced by~$v_{j'} \cdots v_j$ and $(v_{j'}, +) \cdots (v_j, p_j)$ coincide, save for the irrelevant intermediate vertices of colors zero and one.
Hence,~$(v_j, p_j)$ is not an accepting sink in~$\game_{v_{j'}}$.
Since~$v_j \in V_0$, the vertex~$v_{j+1}$ is unique. In both cases, we define~$\sigma(\pi) = v_{j+1}$.
This concludes the definition of~$\sigma$.

It remains to define~$h(\pi \cdot v_{j+1})$ such that it satisfies the above invariant.
To this end, we use one of two operations.
Firstly, we define the \textbf{discontinuous extension of~\boldmath$h(\pi)$ with~\boldmath$v_{j+1}$} as \[h(\pi) \cdot (v_{j+1}, +) \enspace .\]
Secondly, we define a simulated extension of~$h(\pi)$ such that we obtain~$h(\pi \cdot v_{j+1})$ by simulating the move from~$v_j$ to~$v_{j+1}$ in some~$\game_{v}$.
Formally, we define the \textbf{simulated extension of~\boldmath$h(\pi)$ with~\boldmath$v_{j+1}$ and charge~\boldmath$m$} as
\[
h(\pi) \cdot ((v_j, v_{j+1}), p_j, 1) \cdot {((v_j, v_{j+1}), p_j, 0)}^m \cdot (v_{j+1}, p_{j+1}) \enspace , \]
where~$p_{j+1} = p_j$ if~$m = 0$ and~$p_{j+1} = \compl{p_j}$ otherwise.
This ensures that the extension is a play infix in some~$\game_{v}$.

In order to define~$h(\pi \cdot v_{j+1})$ we again distinguish whether or not~$\mrr(\pi)$ is defined.
If~$\mrr(\pi) = \bot$, we define~$h(\pi \cdot v_{j+1})$ to be the discontinuous extension of~$h(\pi)$ with~$v_{j+1}$.
This clearly satisfies the first condition of the invariant.
Moreover, the second condition of the invariant is satisfied as well:
If~$\mrr(\pi \cdot v_{j+1}) = \bot$, this condition holds true vacuously.
Otherwise, we have~$\mrr(\pi \cdot v_{j+1}) = j+1$ and observe that the suffix~$(v_{j+1}, +)$ of~$h(\pi \cdot v_{j+1})$ satisfies the second condition of the invariant.

If, however,~$\mrr(\pi) \neq \bot$, then let~$\mrr(\pi) = j'$.
By definition of~$\mrr$ we have~$\mrr(\pi \cdot v_{j+1}) \in \set{\bot, j', j+1}$.
We distinguish two sub-cases and first define~$h(\pi \cdot v_{j+1})$ for the case that~$\mrr(\pi \cdot v_{j+1}) \in \set{\bot, j+1}$.
In this case, the move to~$v_{j+1}$ either answers the most relevant request in~$\pi$, or the request posed by visiting~$v_{j+1}$ is itself the most relevant request of $\pi \cdot v_{j+1}$:
We have~$\mrr(\pi \cdot v_{j+1}) = \bot$ in the former case and~$\mrr(\pi \cdot v_{j+1}) = j+1$ in the latter one.
In either case, we define~$h(\pi \cdot v_{j+1})$ to be the discontinuous extension of~$h(\pi)$ with~$v_{j+1}$ and observe that the first condition of the invariant holds true.
If~$\mrr(\pi \cdot v_{j+1}) = \bot$, the second condition of the invariant holds true vacuously.
If, however,~$\mrr(\pi \cdot v_{j+1}) = j+1$, then the suffix~$(v_{j+1}, +)$ witnesses that the second condition of the invariant holds true.

Now assume that the move to~$v_{j+1}$ neither opens a new most relevant request, nor answers the existing one, i.e., that we have~$\mrr(\pi \cdot v_{j+1}) = j'$.
In this case, we extend the suffix of~$h(\pi)$ that is consistent with~$\sigma_{v_{j'}}$ by simulating the move from~$v_j$ to~$v_{j+1}$.
Recall that we picked the vertex~$v_{j+1}$ such that~$(v_{j'}, +) \cdots (v_j, p_j) \cdot ((v_j, v_{j+1}), p_j, 1)$ is consistent with $\sigma_{v_{j'}}$.
As we can freely choose whether or not Player~$1$ switches the polarity in the simulation, we follow the intuition stated during the construction of the polarized arena:
Recall that both players currently play ``with respect to'' the request for~$\col(v_{j'})$ opened by visiting~$v_{j'}$.
Hence, we opt to let Player~$1$ move to positive polarity if the cost of the request for~$\col(v_{j'})$ so far is nonnegative, and let him move to negative polarity otherwise.
To this end, we use the sign function~$\sign$, which is defined as
\[
	\sign(n) = \begin{cases}
 		+ & \text{if~$n \geq 0$} \enspace , \\
 		- & \text{otherwise} \enspace .
 \end{cases}
\]

If~$\sign(\weight(v_{j'} \cdots v_{j+1})) = p_j$, we define~$h(\pi \cdot v_{j+1})$ to be the simulated extension of~$h(\pi)$ with~$v_{j+1}$ and charge~$0$, thus implementing the intuition as described above.
Otherwise, i.e., if $\sign(\weight(v_{j'} \cdots v_{j+1})) = \compl{p_j}$, let~$m > 0 $ such that  $(v_{j'}, +) \cdots (v_{j}, p_{j}) \cdot ((v_j, v_{j+1}), p_j, 1) \cdot {((v_j, v_{j+1}), p_j, 1)}^{m} \cdot (v_{j+1}, \compl{p_j})$ is consistent with $\sigma_{v_{j'}}$.
Such an~$m$ exists, as otherwise the play $(v_{j'}, +) \cdots (v_{j}, p_{j}) \cdot ((v_j, v_{j+1}), p_j, 1) \cdot {((v_j, v_{j+1}), p_j, 0)}^\omega$ of Type~$-1$ that starts in~$(v_{j'}, +)$ would be consistent with the winning strategy~$\sigma_{v_{j'}}$ from~$(v_{j'}, +)$ for Player~$0$, contradicting Item~\ref{rem:types:playerzero} of Remark~\ref{rem:types}.
We define~$h(\pi \cdot v_{j+1})$ to be the simulated extension of~$h(\pi)$ with~$v_{j+1}$ and charge~$m$.
Since we have~$\mrr(v_0\cdots v_{j+1}) = j'$ by assumption, either definition of~$h(\pi \cdot v_{j+1})$ satisfies the invariant.
This completes the definition of~$h$.

It remains to show that the strategy~$\sigma$ is indeed winning for Player~$0$ from~$v^*$.
To this end, fix some play~$\rho = v_0v_1v_2\cdots$ consistent with~$\sigma$ starting in~$v^*$. Further, let $ \rho'$ be the limit of the~$h(\pi)$ for increasing prefixes~$\pi$ of~$\rho$, i.e., the unique infinite sequence with prefix~$h(\pi)$ for every prefix~$\pi$ of $\rho$.

By the construction of~$h$ and $\rho'$ we obtain~$\unpol(\rho') = \rho$.
Hence, $\rho'$ is of the form $(v_0, p_0)\cdots (v_1, p_1) \cdots (v_2, p_2) \cdots$.
We call a position~$j$ of~$\rho$ a \textbf{discontinuity} of~$\rho$ if either~$j = 0$ or if~$h(v_0 \cdots v_j)$ is the discontinuous extension of~$h(v_0 \cdots v_{j-1})$ by~$v_j$.

Let~$j$ and~$j'$ be adjacent discontinuities of~$\rho$.
We call the infix~$v_j \cdots v_{j'-1}$ of~$\rho$ an \textbf{equi-strategic infix} (\esi) of~$\rho$.
Moreover, if there only exist finitely many discontinuities of~$\rho$, let~$j^*$ be its final discontinuity.
We call the suffix~$v_{j^*}v_{j^*+1}v_{j^*+2}\cdots$ of~$\rho$ the \textbf{terminal} \esi of~$\rho$.

\begin{numrem}%
\label{rem:esi:equivalence}
Let~$\mu = v_{j}v_{j+1}v_{j+2}\cdots$ be an \esi of $\rho$.
\begin{enumerate}
	\item\label{rem:esi:equivalence:finite} If~$\mu$ is finite, then the infix~$\mu'$ of~$\rho'$ starting at position~$\card{h(v_0\cdots v_j)}-1$ and ending at position~$\card{h(v_0 \cdots v_{j + \card{\mu} - 1})}$ starts in~$(v_j, +)$, ends in some~$(v_{j + \card{\mu} - 1}, p)$, and is consistent with~$\sigma_{v_j}$.
	\item\label{rem:esi:equivalence:infinite} If~$\mu$ is infinite, then the suffix~$\mu'$ of~$\rho'$ starting at position~$\card{h(v_0\cdots v_j)}-1$ starts in~$(v_j, +)$ and is consistent with~$\sigma_{v_j}$.
\end{enumerate}
\end{numrem}

\noindent
For each position~$j$ of~$\rho$ we define~$\esi(j) = k$ if the~$k$-th \esi of~$\rho$ contains~$v_j$.
Moreover, if~$\mu = v_{j}v_{j+1}v_{j+2} \cdots$ is an \esi of~$\rho$, then we call~$\col(v_j)$ the \textbf{characteristic color} of~$\mu$.
By the construction of~$h$, if the characteristic color of an \esi~$\mu$ is even, then~$\mu$ consists only of a single vertex.
If, however, the characteristic color~$c$ of an \esi~$\mu$ is odd, then we have~$\col(v) \leq c$ for all vertices~$v$ in~$\mu$.
Moreover, let~$c'$ be the characteristic color of the \esi succeeding~$\mu$, if~$\mu$ is not the terminal \esi of~$\rho$.
Due to the construction of~$h$, we have~$c' > c$.
If~$c'$ is even, this observation implies~$c' \in \answer{\col(v)}$ for all vertices~$v$ in~$\mu$.
As the number of colors in~$\game$ is finite, this in turn implies that the number of \esis between a request and its response (if a response exists at all) is bounded.
\begin{numrem}%
\label{rem:epra:esi-bound}
Let~$j$ be some position in~$\rho$ and let~$k = \esi(j)$.
Moreover, let~$d$ be the number of colors in~$\game$.
\begin{enumerate}
	\item If the request at position~$j$ is first answered at position~$j'$, then~$\esi(j') < k+d$
	\item If the request at position~$j$ is unanswered in~$\rho$, then~$\rho$ contains less than~$k + d$ many \esis.
\end{enumerate}
\end{numrem}

\begin{proof}
\textbf{1.}
Let~$\esi(j') = k'$.
If~$\col(v_j)$ is even, we obtain~$k = k'$ and the claim is trivial.
Thus, assume that $\col(v_j)$ is odd.
Since the request at position~$j$ is answered in the~$k'$-th \esi, define the sequence~$j_0 < \cdots < j_{k' - k}$ inductively such that~$j_0 = \mrr(v_0\cdots v_j)$ and such that, for each~$\ell$ with~$0 \leq \ell \leq k'-k$ the position~$j_\ell$ is the unique position in~$\rho$ with~$\esi(j_\ell) \neq \esi(j_{\ell - 1})$.
We claim that the requests opened by visiting the~$v_{j_\ell}$ are answered at position~$j'$ at the earliest and show this by induction over~$\ell$.

First consider~$\ell = 0$.
Since~$\col(v_j)$ is odd, we obtain~$j_0 \neq \bot$ and, hence, that~$\col(v_{j_0})$ is odd.
Assume towards a contradiction that the request at position~$j_0$ is answered at position~$j'' < j'$.
If~$j'' < j$, then this contradicts our choice of~$j_0$ as the position of the most relevant request of the play prefix~$v_0 \cdots v_j$: Since~$\col(v_j)$ is odd, we would then obtain~$\mrr(v_0 \cdots v_j) > j'' > 0 = \mrr(v_0 \cdots v_j)$.
If, however,~$j < j'' < j'$, we obtain that the request at position~$j$ is answered at position~$j''$, since~$\col(v_{j_0}) \geq \col(v_j)$ by definition of the most relevant requests.
Hence, we obtain~$j'' \geq j'$.

Now consider~$j_\ell$ for~$0 < \ell \leq k' - k$: Since we have~$\col(v_{j_{\ell}}) > \col(v_{j_{\ell-1}})$, an answer to the request posed by~$v_{j_{\ell}}$ would answer the request posed by~$v_{j_{\ell - 1}}$ as well.
Via the induction hypothesis we obtain that~$v_{j_{\ell}}$ is answered at position~$j'$ at the earliest.

Thus, the only possibility to switch the mimicked strategy  and enter a new \esi is by visiting a vertex of larger odd color than the vertex at the initial position of the current \esi.
Hence, the color of the~$v_{j_\ell}$ is strictly monotonically increasing, which implies~$k' < k +d$.

\textbf{2.}
Since the request at position~$j$ is unanswered, we obtain that~$\col(v_j)$ is odd.
Again define the sequence~$j_0 < j_1 < j_2 < \cdots$ inductively such that~$j_0 = \mrr(v_0\cdots v_j)$ and such that, for each~$\ell > 0$, the position~$j_\ell$ is the unique position in~$\rho$ with~$\esi(j_\ell) \neq \esi(j_{\ell - 1})$.
Due to similar reasoning as in the previous case, we obtain that the request posed at each~$j_\ell$ is unanswered.
Thus, the only possibility to enter a new \esi is again by visiting a vertex of higher odd color than the vertex at the initial position of the current \esi, which again implies that the number of \esis in~$\rho$ is bounded by~$k + d$.
\end{proof}

Recall that the bounded parity condition with weights requires the play~$\rho$ to not only satisfy the parity condition, but also that the cost of almost all requests is bounded and that there exists no unanswered request with infinite cost in~$\rho$.
We first show that~$\rho$ satisfies the classical parity condition.
In a second step, we then show that there exists a bound on the cost of each (answered or unanswered) request in~$\rho$.
The former condition, i.e., that~$\rho$ satisfies the parity condition, is in large parts implied by Remark~\ref{rem:epra:esi-bound}.

\begin{lem}%
\label{lem:epra:winning-strategy}
The play~$\rho$ satisfies the parity condition.
\end{lem}

\begin{proof}
If~$\rho$ contains no unanswered requests, then it vacuously satisfies the parity condition.
Hence, let~$j$ be the position of such an unanswered request in~$\rho$.
Due to Remark~\ref{rem:epra:esi-bound}, we obtain that there exist only finitely many \esis in~$\rho$.
Let~$\mu = v_{j^*}v_{j^*+1}v_{j^*+2}\cdots$ be the terminal \esi of~$\rho$.
By the construction of~$h$, there exists a suffix~$\mu'$ of~$\rho'$ with~$\unpol(\mu') = \mu$.
Due to Item~\ref{rem:esi:equivalence:infinite} of Remark~\ref{rem:esi:equivalence}, the suffix~$\mu'$ begins in~$(v_{j^*}, +)$ and is consistent with the winning strategy~$\sigma_{v_{j^*}}$ for Player~$0$ from~$(v_{j^*}, +)$ in~$\game_{v_{j^*}}$.
Moreover,~$\mu'$ is a play of Type~$0$ due to~$\mu'$ being the terminal \esi of~$\rho$ and due to being consistent with~$\sigma_{v_{j^*}}$.
Hence, we obtain that~$\mu$ satisfies the parity condition via Remark~\ref{rem:unpolparity}, which in turn implies that~$\rho$ satisfies the parity condition.
\end{proof}

It remains to show that the costs of requests in~$\rho$ are bounded.
Recall that we defined $n' = \card{\arena_{v}}$,~$d'$ as the largest color of a vertex in the~$\game_{v}$, and~$W'$ as the largest absolute weight of an edge.
We claim that the costs of the most relevant requests in~$\rho$ are bounded by~${(n'd'W')}^2$.
This implies that the cost of all requests is bounded:
Due to Remark~\ref{rem:epra:esi-bound} we obtain that the number of \esis between a request and its response, if one exists, is bounded by~$d$.
Hence, it suffices to show that each~\esi contributes at most a bounded amount to the cost of answering a request.

\begin{lem}%
\label{lem:epra:esi-cost-bound}
Let~$\mu = v_j v_{j+1} v_{j+2} \cdots$ be an \esi of~$\rho$.
For each~$j'$ with~$j \leq j' < j + \card{\mu}$ we have~$\card{\weight(v_j \cdots v_{j'})} \leq d'{(n'W')}^2$.
\end{lem}

\begin{proof}
Since~$\mu$ is an \esi, Player~$0$ mimicks the choices of the strategy~$\sigma_{v_j}$, which is a winning strategy for him in the induced energy parity game~$\game_{v_j}$.
Recall that~$\sigma_{v_j}$ is of size at most~$n'd'W'$.
For the sake of readability, we define~$s = n'd'W'$.

Now, towards a contradiction, let~$j'$ be a position with~$j \leq j' < j + \card{\mu}$, such that we have $\card{\weight(v_j \cdots v_{j'})} > d'{(n'W')}^2 = n'sW'$.
We define~$\pi = v_j \cdots v_{j'}$ and assume $\weight(\pi) > n'sW'$, i.e., that the infix~$\pi$ violates the claimed bound from above.
The other case is dual.
Let~$j''$ be the minimal position in~$\mu$ such that the weight accrued by the \esi is strictly positive starting at~$j''$ until it reaches~$j'$.
Formally, we define~$j''$ as the minimal position that satisfies~$\weight(v_j \cdots v_k) > 0$ for all~$k \in \set{j'', \dots, j}$.

Due to the definition of~$W$ we obtain~$\weight(v_j \cdots v_{j''}) \leq W$.
Since we furthermore have~$\weight(\pi) = \weight(v_{j} \cdots v_{j'}) > n'sW'$ by assumption, we obtain
\[
	\weight(v_{j''} \cdots v_{j'})
	= \weight(v_{j} \cdots v_{j'}) - \weight(v_{j} \cdots v_{j''})
	> n'sW' - W'
	\geq n'sW' - W' + 1 \enspace . \]

Now consider the unique infix~$\pi' = (v_{j}, p_{j}) \cdots (v_{j'}, p_{j'})$ of~$\rho'$ that corresponds to the infix~$\pi$ of~$\rho$.
Since~$\pi$ is an infix of a single \esi,~$\pi'$ is a play prefix in the energy parity game~$\game_{v_j}$ that is consistent with the strategy~$\sigma_{v_j}$, which is winning for Player~$0$ from~$(v_j, p_j)$.
Furthermore, since we have~$\weight(v_j \cdots v_k) > 0$ for all~$k \in \set{j'', \dots, j'}$, the polarity of~$(v_{j''}, p_{j''}) \cdots (v_{j'}, p_{j'})$ is constant and positive.
Hence, by construction of~$\rho'$, we obtain
\[  \weight_{v_{j}}((v_{j''}, p_{j''}) \cdots (v_{j'}, p_{j'})) = -\weight(v_{j''} \cdots v_{j'}) \enspace. \]

Due to the above inequality, this implies
\[ \weight_{v_{j}}((v_{j''}, p_{j''}) \cdots (v_{j'}, p_{j'})) \leq -n'sW' + W' - 1 = -(n's -1)W - 1 \enspace. \]
As argued above, the infix~$(v_{j''}, p_{j''}) \cdots (v_{j'}, p_{j'})$ is an infix of a play in the energy parity game~$\game_{v_j}$ that starts in~$(v_j, p_j)$ and is consistent with the strategy~$\sigma_{v_j}$ for Player~$0$.
Furthermore,~$\sigma_{v_j}$ is winning for Player~$0$ from~$(v_j, p_j)$ and is of size~$s$.
This yields the desired contradiction to Lemma~\ref{lem:boundeddecreaseinenergyparity} on Page~\pageref{lem:boundeddecreaseinenergyparity}.
\end{proof}

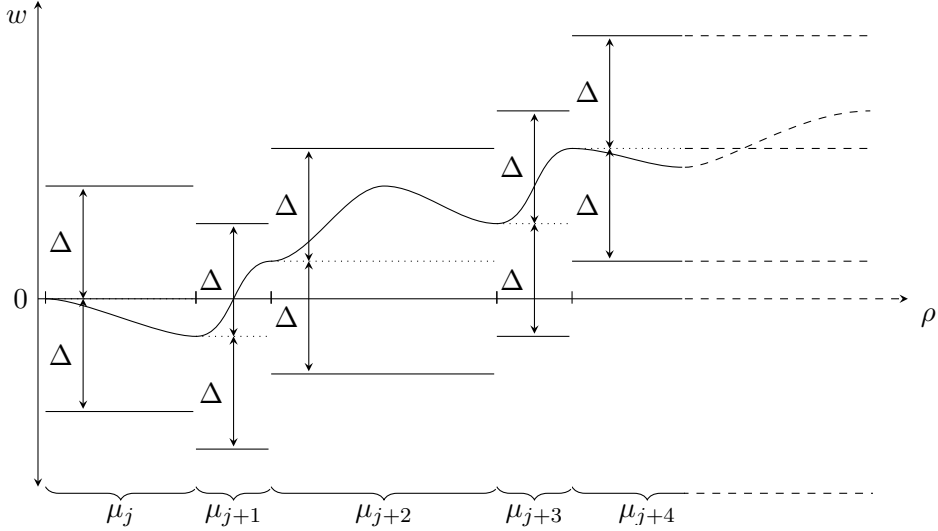
\begin{figure}
\centering
\begin{tikzpicture}
	\node[anchor=east] at (-.1,3.75) {$\weight$};
	\node[anchor=east] at (-.1,0) {$0$};
	\node[anchor=north west] at (11.5,0) {$\rho$};
	\path[draw,stealth-stealth] (-.1,-2.5) -- (-.1,4);
	\path[draw] (-.1,0) -- (8.5,0);
	\path[draw,dashed,-stealth] (8.5,0) -- (11.5,0);
	
	\draw[draw,-] (0,0)  .. controls (0.5, 0) and (1.5,-.5) ..
	              (2,-.5) .. controls (2.5,-.5) and (2.5, .5) ..
	              (3, .5) .. controls (3.5, .5) and (4,1.5) ..
	              (4.5,1.5) .. controls (5,1.5) and (5.5,1) ..
	              (6,1) .. controls (6.5,1) and (6.5,2) ..
	              (7,2) .. controls (7.5,2) and (8, 1.75) .. (8.5, 1.75);
	\draw[dashed] (8.5, 1.75) .. controls (9, 1.75) and (10,2.5) .. (11,2.5);
	              
	\newcommand{\xtick}[1]{\draw (#1,-.25) -- (#1,.25);}
	\newcommand{\costwidth}{1.5}
	\newcommand{\costbounds}[3]{
		\coordinate (west) at #1;
		\coordinate (east) at #2;
		\draw[stealth-stealth] ($(west) + (.5,0)$) -- node[anchor=east] {$\Delta$} ($(west) + (.5, \costwidth)$);
		\draw[stealth-stealth] ($(west) + (.5,0)$) -- node[anchor=east] {$\Delta$} ($(west) + (.5, -\costwidth)$);
		\draw ($(west) + (0,\costwidth)$) -- ($(east) + (0, \costwidth)$);
		\draw ($(west) + (0,-\costwidth)$) -- ($(east) + (0, -\costwidth)$);

		\coordinate (westorigin) at (west |- {{(0,0)}});
		\coordinate (eastorigin) at (east |- {{(0,0)}});
		\draw ($(westorigin) + (0,.1)$) -- ($(westorigin) - (0,.1)$);
		\draw ($(eastorigin) + (0,.1)$) -- ($(eastorigin) - (0,.1)$);
		
		\draw[decorate,decoration={brace,amplitude=5pt,mirror}]
		($(westorigin) - (0,2.5)$) -- node[anchor=north,transform canvas={yshift={-3pt}}] {$\mu_{#3}$} ($(eastorigin) - (0,2.5)$);
		
		\draw[dotted] (west) -- (east);
	}
	
	\costbounds{(0,0)}{(2,0)}{j}
	\costbounds{(2,-.5)}{(3,-.5)}{j+1}
	\costbounds{(3,.5)}{(6,.5)}{j+2}
	\costbounds{(6,1)}{(7,1)}{j+3}
	
	\draw[stealth-stealth] ($(7,2) + (.5,0)$) -- node[anchor=east] {$\Delta$} ($(7,2) + (.5, \costwidth)$);
	\draw[stealth-stealth] ($(7,2) + (.5,0)$) -- node[anchor=east] {$\Delta$} ($(7,2) + (.5, -\costwidth)$);
	\draw[dotted] (7,2) -- (8.5,2);
	\draw[dashed] (8.5,2) -- (11,2);
	\draw (7,3.5) -- (8.5,3.5);
	\draw[dashed] (8.5,3.5) -- (11,3.5);
	\draw (7,.5) -- (8.5,.5);
	\draw[dashed] (8.5,.5) -- (11,.5);
	
	\begin{scope}
		\clip (7, -2.4) rectangle (8.5, -3);
		\draw[decorate,decoration={brace,amplitude=5pt,mirror}]
		(7, -2.5) -- node[anchor=north,transform canvas={yshift={-3pt}}] {$\mu_{j+4}$} (9, -2.5);
	\end{scope}
	\begin{scope}
		\clip (8.5, -2.4) rectangle (11, -3);
		\draw[decorate,decoration={brace,amplitude=5pt,mirror},dashed]
		(2, -2.5) -- (11.5, -2.5);
	\end{scope}

\end{tikzpicture}

\caption{Bounds on the cost of a request over time given by Lemma~\ref{lem:epra:esi-cost-bound}. We write~$\Delta = d'{(n'W')}^2$.}%
\label{fig:epra:cor-bound}
\end{figure}

Due to Lemma~\ref{lem:epra:esi-cost-bound}, each \esi strictly in-between a request and its response contributes at most~$d'{(n'W')}^2$ to the cost incurred by the request.
Similarly, the \esi containing the request and its response also contribute at most~$d'{(n'W')}^2$ each to the cost of answering the given request.
Hence, via Remark~\ref{rem:epra:esi-bound}, we obtain that each (answered or unanswered) request in~$\rho$ incurs a cost of at most~${(d'n'W')}^2$.
We illustrate this argument in Figure~\ref{fig:epra:cor-bound}.
Hence,~$\sigma$ is a winning strategy for Player~$0$ from~$v^*$ in~$\game$, as each play that starts in~$v^*$ and is consistent with~$\sigma$ satisfies the parity condition due to Lemma~\ref{lem:epra:winning-strategy} and because no such play contains a request that is unanswered with infinite cost, which concludes the proof of Item~\ref{lem:bounded-cost-parity:reduction:player0} of Lemma~\ref{lem:bounded-cost-parity:reduction}.

Before we conclude this section, we formalize the above observation about the winning strategy for Player~$0$ uniformly bounding the costs of requests in the following corollary.
To this end, we use the upper bounds $n' \le 2n + 4n^2$, $d' \le d+2$, and $W' \le W+1$.

\begin{cor}%
	\label{cor:bounded-parity-with-weights:uniform-bound}
	Let~$\game$ be a bounded parity game with weights with~$n$ vertices,~$d$ colors, and largest absolute weight~$W$.
	There exists a strategy~$\sigma$ for Player~$0$ that is winning from~$\winreg_0(\game)$, such that in each play~$\rho$ consistent with~$\sigma$, each request is answered or unanswered with cost at most~${((d+2)(2n + 4n^2)(W +1))}^2$.
\end{cor}

Using arguments from Section~\ref{sec:quality}, this bound can be improved to ${((d+2)(6n)(W +1))}^2$.
However, as we only use Corollary~\ref{cor:bounded-parity-with-weights:uniform-bound} later to obtain some upper bound on the quality of such strategies, we refrain from repeating these arguments here.


\section{Memory Requirements}%
\label{sec:memory}
We now discuss the upper and lower bounds on the memory required to implement winning strategies for either player.
Recall that we use binary encoding to denote weights, i.e., weights may be exponential in the size of the game.
In this section we show pseudo-polynomial bounds on the necessary and sufficient memory for Player~$0$ to win parity games with weights.
In contrast, Player~$1$ requires infinite memory.

\begin{thm}%
\label{thm:memory}
Let~$\game$ be a parity game with weights with~$n$ vertices,~$d$ colors, and largest absolute weight~$W$ assigned to any edge in~$\game$.
\begin{enumerate}
	\item Player~$0$ has a winning strategy~$\sigma$ from~$\winreg_0(\game)$ with~$\card{\sigma} \in \bigo (nd^2W)$. This bound is tight.
	\item There exists a parity game with weights~$\game$, such that Player~$1$ has a winning strategy from each vertex~$v$ in~$\game$, but she has no finite-state winning strategy from any~$v$ in~$\game$.
\end{enumerate}
\end{thm}

\noindent
The proof of the second item of Theorem~\ref{thm:memory} is straightforward, since Player~$1$ already requires infinite memory to implement winning strategies in finitary parity games~\cite{ChatterjeeHenzingerHorn09}.
Since parity games with weights subsume finitary parity games, this result carries over to our setting.
We show the game witnessing this lower bound on the right-hand side of Figure~\ref{fig:energyparity:difference} on Page~\pageref{fig:energyparity:difference}.

In contrast, pseudo-polynomial memory is sufficient, but also necessary, for Player~$0$.
To show this claim, we first prove that the winning strategy for him in a bounded parity game with weights constructed in the proof of Item~\ref{lem:bounded-cost-parity:reduction:player0} of Lemma~\ref{lem:bounded-cost-parity:reduction} suffers at most a linear blowup in comparison to his winning strategies in the underlying energy parity games.
This is sufficient as we have argued in Section~\ref{sec:parity-weights-solve} that the construction of a winning strategy for Player~$0$ in a parity game with weights suffers no blowup in comparison to the underlying bounded parity games with weights.

\begin{lem}%
\label{lem:memory:player-0}
Let~$\game$ be a bounded parity game with weights and let~$n$,~$d$, and~$W$ be defined analogously to Theorem~\ref{thm:memory}.
Player~$0$ has a finite-state winning strategy of size at most~$d(6n)(d+2)(W+1)$ from~$\winreg_0(\game)$.
\end{lem}

\begin{proof}
Let~$V$ and~$E$ be the vertex and edge sets of~$\game$ and recall that we have defined~$P = \set{+,-}$.
In the proof of Item~\ref{lem:bounded-cost-parity:reduction:player0} of Lemma~\ref{lem:bounded-cost-parity:reduction}, we have constructed an energy parity game~$\game_v$ with vertices~$(V \times P) \cup (E \times P \times \set{0,1})$ for each vertex~$v$ of~$\game$.
We have then constructed a winning strategy~$\sigma$ for Player~$0$ for~$\game$ out of winning strategies for him in the~$\game_v$.
As it is straightforward to implement~$\sigma$ via the disjoint union of memory structures implementing the constituent strategies, this approach yields an upper bound of~$n(2n + 4n^2)(d+2)(W+1)$ on the size of~$\sigma$ due to the upper bound on the size of winning strategies for Player~$0$ in energy parity games from Proposition~\ref{prop:energyparity:equivalence}.

In the construction of the~$\game_v$, however, we only store the edges chosen by the players in the vertices of the form~$E \times P \times \set{0,1}$ for didactic purposes.
In fact, it suffices to store the target vertex of an edge instead, resulting in a vertex set of size~$6n$ of the~$\game_v$.
Moreover, recall that the definition of the~$\game_v$ only takes the color of~$v$ into account: If the vertices~$v$ and~$v'$ have the same color, then the games~$\game_{v}$ and~$\game_{v'}$ are isomorphic.
Further, Chatterjee and Doyen have shown that, if Player~$0$ wins an energy parity game~$\game'$ with~$n'$ vertices,~$d'$ colors, and largest absolute weight~$W'$, then he has a uniform strategy of size~$n'd'W'$ that is winning from all vertices from which he wins~$\game'$~\cite{ChatterjeeDoyen12}.
Hence, it suffices to combine at most~$d$ strategies, each of size~$(6n)(d+2)(W+1)$, in order to obtain a winning strategy for Player~$0$ in~$\game$.
\end{proof}

Having established an upper bound on the memory required by Player~$0$, we now proceed to show that this bound is tight.

\begin{lem}%
\label{lem:pw-memory-lb}
Let~$n,W \in \nats$.
There exists a parity game with weights~$\game_{n,W}$ with~$\bigo(n)$ vertices and largest absolute weight~$W$ such that Player~$0$ wins~$\game_{n,W}$ from every vertex, but each winning strategy for him is of size at least~$nW + 1$.
\end{lem}

\begin{figure}
\centering
\begin{tikzpicture}[thick]
	\node[p1] (outer-req) at (0,0) {\parnode{v_\req}{3}};
	\node[p1] (inner-req-1) at (2,0) {\parnode{v'_{\req, 1}}{1}};
	\node[p1] (inner-req-n) at (6,0) {\parnode{v'_{\req, n}}{1}};
	\node[p0] (delay) at (9,0) {\parnode{v_\del}{1}};
	\node[p1] (inner-ans) at (11,0) {\parnode{v'_\ans}{2}};
	\node[p1] (outer-ans) at (13,0) {\parnode{v_\ans}{4}};
	
	\path
		(outer-req)
			edge[weight=0 at .5 anchor north] (inner-req-1)
		(inner-req-1)
			edge[
				weight=$W$ at .25 anchor north,
				elided,
				weight=$W$ at .75 anchor north]
			(inner-req-n)
		(inner-req-n)
			edge[weight=$W$ at .5 anchor north] (delay)
		(delay)
			edge[loop below,weight=$-1$ at .5 anchor north] (delay)
			edge[weight=$0$ at .5 anchor north] (inner-ans)
		(inner-ans)
			edge[weight=$0$ at .5 anchor north] (outer-ans);
		
	\path[thick,brace = {$n$ vertices} amplitude 5pt,transform canvas={yshift=-10pt}]
		(inner-req-1.south west) -- (inner-req-n.south east);
		
	\path[draw, thick,->,rounded corners,weight=0 at .5 anchor south]
		(inner-ans) |- +(-1,.75) -| (inner-req-1);
	\path[draw, thick,->,rounded corners,weight=0 at .5 anchor south]
		(outer-ans) |- +(-1,1.5) -| (outer-req);
\end{tikzpicture}
\caption{A game of size~$\bigo(n)$ in which Player~$0$ only wins with strategies of size at least~$nW + 1$.}%
\label{fig:memory:lower-bound}
\end{figure}
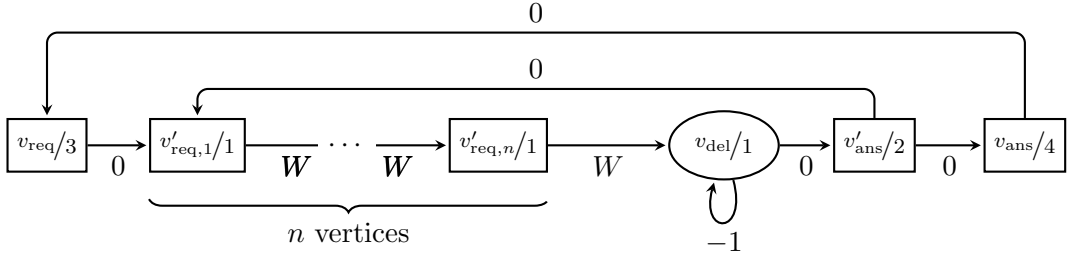

\begin{proof}
We show the game~$\game_{n,W}$ in Figure~\ref{fig:memory:lower-bound}.
This game has~$n + 4$ vertices and the largest absolute weight of an edge is~$W$ as required.
The only vertices with more than one successor are~$v_\del \in V_0$ and~$v'_\ans \in V_1$.
Thus, the only choice of Player~$0$ in~$\game_{n,W}$ consists in determining how often to take the self-loop of vertex~$v_\del$ upon each visit.
Dually, the only choice of Player~$1$ consists of deciding whether or not to move from~$v'_\ans$ to~$v'_{\req, 1}$, or to continue to~$v_\ans$.

Player~$0$ wins~$\game_{n,W}$ from each vertex by taking the self-loop of~$v_\del$~$nW$ times upon each visit to~$v_\del$ and by subsequently moving to~$v'_\ans$.
Each request in each play that is consistent with this strategy is answered or unanswered with cost at most~$nW$, independent of the choices of Player~$1$ in~$v'_\ans$.
Moreover, as the only way to visit~$v_\req$ is to move there from~$v_\ans$, the play visits~$v_\ans$ infinitely often if and only if it visits~$v_\req$ infinitely often.
Further, the play visits~$v'_{\req, 1}$ and~$v'_\ans$ infinitely often.
Hence, almost all requests are answered, i.e., this strategy is winning for Player~$0$ from all vertices.
This strategy can be implemented by a counter that counts the number of self-loops of~$v_\del$ taken so far, which is reset to zero upon leaving~$v_\del$.
As this counter is bounded by~$nW$, the strategy is of size~$nW + 1$.

It remains to show that each finite-state winning strategy for Player~$0$ has at least $nW + 1$ memory states.
Towards a contradiction, let~$\sigma$ be a winning strategy for Player~$0$ from some vertex~$v$ with less than~$nW + 1$ memory states.
We implement a strategy for Player~$1$ using a counter~$\kappa$ that is initialized with one if~$v = v_\req$ and with zero otherwise.
Upon each visit to~$v_\req$ we increment~$\kappa$.
After each visit to~$v_\req$, the strategy~$\tau$ prescribes moving from~$v'_\ans$ to~$v'_{\req, 1}$ for the first~$\kappa$ visits to~$v'_\ans$, and it prescribes moving from~$v'_\ans$ to~$v_\ans$ at the~$(\kappa+1)$-th visit to~$v'_\ans$.
Hence, after the~$(\kappa+1)$-th visit to~$v'_\ans$, the vertex~$v_\req$ is visited again,~$\kappa$ is incremented, and the behavior of~$\tau$ described above repeats with incremented~$\kappa$.

Let~$\rho$ be a play consistent with~$\sigma$ and~$\tau$.
Since~$\sigma$ is winning for Player~$0$, the play~$\rho$ does not remain in $v_\del$ from some point onwards ad infinitum, as this would violate the parity condition and thus contradict~$\sigma$ being winning from~$v$ for Player~$0$.
Hence, playing consistently with~$\tau$, Player~$1$ enforces a play that starts with a (possibly empty) finite prefix that ends before the first visit to~$v_\req$ and that continues with infinitely many rounds, each starting in~$v_\req$.
The~$j$-th round is of the form
\[
	v_\req \cdot \Pi_{k = 0, \dots, j}  (v'_{\req,1} \cdots v'_{\req, n} \cdot {(v_\del)}^{\ell_{j, k}} \cdot v'_\ans ) \cdot v_\ans \enspace .
\]

We first show~$\ell_{j,k} < nW + 1$ for all~$j, k$.
Towards a contradiction, assume~$\ell_{j,k} \geq nW + 1$ for some~$j, k \in \nats$.
Since~$\sigma$ is of size less than~$nW + 1$, a straightforward pumping argument shows that the play~$\rho'$ consisting of the finite prefix of~$\rho$ ending before the first visit to~$v_\req$ concatenated with the first~$j-1$ rounds of~$\rho$, but ending with the infinite suffix
\[
	v_\req \cdot \Pi_{k' = 0, \dots, k-1}  (v'_{\req,1} \cdots v'_{\req, n} \cdot {(v_\del)}^{\ell_{j, k'}} \cdot v'_\ans ) \cdot v'_{\req,1} \cdots v'_{\req, n} \cdot {(v_\del)}^\omega \ ,
\]
 is consistent with~$\sigma$.
 This, however, contradicts~$\sigma$ being a winning strategy for Player~$0$ from~$v$ as~$v_\del$ has the odd color one, i.e., as the resulting play violates the parity condition.
Hence, we obtain~$\ell_{j,k} < nW + 1$ for all~$j,k$, which, in turn, yields $\weight({(v_\del)}^{\ell_{j,k}}) > -nW$.

Since each edge~$(v'_{\req, n'}, v'_{\req, n' + 1})$ for $1\leq n' < n$ as well as the edge~$(v'_{\req, n}, v'_\del)$ has weight~$W$, we obtain
\[
	\weight(v'_{\req, 1} \cdots v'_{\req, n} \cdot {(v_\del)}^{\ell_{j, k}} \cdot v'_\ans) > 0
\]
for all~$j,k$.
This, in turn, implies
\[
	\weight(v_\req \cdot \Pi_{k = 0, \dots, j}  (v'_{\req,1} \cdots v'_{\req, n} \cdot {(v_\del)}^{\ell_{j, k'}} \cdot v'_\ans )  )\geq j + 1
\]
for each~$j$.
Since, as argued above, the play~$\rho$ consistent with~$\sigma$ consists of infinitely many rounds, we obtain that for each~$b \in \nats$ there exist infinitely many requests in~$\rho$ that are answered with cost at least~$b$.
Hence, the costs-of-requests along~$\rho$ diverge, which contradicts~$\sigma$ being a winning strategy for Player~$0$.
\end{proof}

This concludes the study of memory requirements for both players in parity games with weights.
For Player~$0$, the results from Theorem~\ref{thm:memory} also hold true for bounded parity games with weights:
Lemma~\ref{lem:memory:player-0} directly yields an upper bound on the memory required by Player~$0$ in order to win in bounded parity games with weights, while it is easy to see that Player~$0$ also wins the games constructed in the proof of Lemma~\ref{lem:pw-memory-lb} when interpreting them as bounded parity games with weights, but only using a strategy of size~$nW + 1$.

For Player~$1$, however, these results do not directly carry over to the setting of bounded parity games with weights.
She has, in fact, a positional winning strategy for the game witnessing necessity of infinite memory for her in parity games with weights shown in Figure~\ref{fig:energyparity:difference} on Page~\pageref{fig:energyparity:difference}, when interpreted as a bounded parity game with weights.
Recall that, in the proof of the first item of Lemma~\ref{lem:bounded-cost-parity:reduction}, for a given parity game with weights~$\game$, we construct a strategy~$\tau$ that is winning for Player~$1$ from a vertex~$v$ out of a winning strategy~$\tau_v$ for her in an induced energy parity game~$\game_v$ from a designated vertex~$v'$.

While~$\tau_v$ can be assumed to be positional as shown by Chatterjee and Doyen~\cite{ChatterjeeDoyen12} (cf. Proposition~\ref{prop:energyparity:player1}), the strategy~$\tau$ keeps track of play prefixes in~$\game_v$ and thus requires potentially infinitely many memory states.
In particular, in order to win~$\game_v$ from~$v'$, recall that it may be necessary to switch between two copies of the arena of~$\game$.
Whether or not to perform this switch is governed by the accumulated weight of the play prefix in~$\game$ thus far.

Hence, our construction from the proof of that lemma does not directly allow us to obtain positional or finite-state winning strategies for Player~$1$ in bounded parity games with weights.
It remains open whether Player~$1$ requires infinite memory to win in bounded parity games with weights.
Since she, however, has finite-state winning strategies in the special case of bounded parity games with costs as shown by Fijalkow and Zimmermann~\cite{FijalkowZimmermann14}, we conjecture that she requires at most finite memory for winning strategies in bounded parity games with weights as well.

We now turn our attention to the quantitative properties of this winning condition.
To this end, we provide tight bounds on the costs of requests that Player~$0$ can guarantee in a parity game with weights~$\game$, if she wins~$\game$ at all.


\section{Quality of Strategies}%
\label{sec:quality}
We have shown in the previous section that finite-state strategies of bounded size suffice for Player~$0$ to win in parity games with weights, while Player~$1$ clearly requires infinite memory.
However, as we are dealing with a quantitative winning condition, we are not only interested in the size of winning strategies, but also in their quality.
More precisely, we are interested in an upper bound on the cost of requests that Player~$0$ can ensure.
In this section, we show that he can guarantee a pseudo-polynomial upper bound on such costs.
Dually, Player~$1$ is required to unbound the cost of responses.

\begin{thm}%
\label{thm:costs}
Let~$\game$ be a parity game with weights with~$n$ vertices,~$d$ colors, and largest absolute weight~$W$.

There exists a~$b \in \bigo({(ndW)}^2)$ and a strategy~$\sigma$ for Player~$0$ such that, for all plays~$\rho$ beginning in~$\winreg_0(\game)$ and consistent with~$\sigma$, we have~$\limsup_{j \rightarrow \infty}\Cor(\rho, j) \leq b$.
This bound is tight.
\end{thm}

We first show that Player~$0$ can indeed ensure an upper bound as stated in Theorem~\ref{thm:costs}.
We obtain this bound via a straightforward pumping argument leveraging the upper bound on the size of winning strategies obtained in Lemma~\ref{lem:memory:player-0}.

\begin{lem}%
\label{lem:parity-games-weights:cost-upper-bound}
Let~$\game$,~$n$,~$d$, and~$W$ be as in the statement of Theorem~\ref{thm:costs} and let~$s = d(6n)(d+2)(W+1)$.
Player~$0$ has a winning strategy~$\sigma$ such that, for each play~$\rho$ that starts in~$\winreg_0(\game)$ and is consistent with~$\sigma$, we have~$\limsup_{j \rightarrow \infty}\Cor(\rho, j) \leq nsW$.
\end{lem}

\begin{proof}
\begin{figure}
\centering
\begin{tikzpicture}[thick]

	\path[draw] (-.5,0) edge (-.5,2);
	\path[draw] (-.5,0) edge (-.5,-2);
	\node[anchor=east] at (-.5,1.75) {$\weight$};

	\path[draw,dashed] (-1,0) -- (0,0);
	\path[draw] (0,0) -- (2.5,0);
	\path[draw,dashed] (2.5,0) -- (3.5,0);
	\path[draw] (3.5,0) -- (5,0);
	\path[draw,dashed] (5,0) -- (7,0);
	\path[draw] (7,0) -- (10,0);
	\path[draw,dashed] (10,0) edge (11,0);
	
	\node[anchor=north] at (10.5,0) {$\rho$};
	
	\path[draw]
		(0,0) .. controls (.5, 0) and (.5,1.5) ..
		(1,1.5) .. controls (1.5,1.5) and (1.5,-.5) ..
		(2,-.5) .. controls (2.25, -.5) and (2.25,0) .. (2.5,0);
	
	\path[draw]
		(3.5,0) .. controls (3.75, 0) and (3.75,-1.5) ..
		(4,-1.5) .. controls ( 4.25, -1.5) and (4.25,.5) ..
		(4.5,.5) ..controls (4.75,.5) and (4.75,.5) .. (5,1);
		
	\path[draw]
		(7,0) .. controls (7.5, 0) and (7.5, 1.5) ..
		(8,1.5) .. controls ( 8.5, 1.5) and (8.5,-1.5) ..
		(9,-1.5) .. controls (9.5, -1.5) and (9.5,-1) .. (10,-1);
		
	\newcommand{\xticksouth}[2]{\draw (#1,-.1) -- (#1,.1); \node[anchor=north] at (#1,0) {#2};}
	\newcommand{\xticknorth}[2]{\draw (#1,-.1) -- (#1,.1); \node[anchor=south] at (#1,0) {#2};}
	\xticksouth{0}{$j_0$};
	\xticksouth{2.5}{$j'_0$};
	\xticknorth{3.5}{$j_1$};
	\xticksouth{5}{$j'_1$};
	\xticksouth{7}{$j_2$};
	\xticksouth{10}{$j'_2$};
	
	\draw[gray] (-.5,1.4) -- (10.5,1.4);
	\node[anchor=west] at (10.5,1.4) {$nsW$};
	\draw[gray] (-.5,-1.4) -- (10.5,-1.4);
	\node[anchor=west] at (10.5,-1.4) {$-nsW$};
	
	\draw[decorate,decoration={brace,amplitude=5pt,mirror}]
		(0,-1.5) -- node[anchor=north,yshift=-2pt,align=center] {Positively\\sumptuous} (2.5,-1.5);
	\draw[decorate,decoration={brace,amplitude=5pt,mirror}]
		(3.5,-1.5) -- node[anchor=north,yshift=-2pt,align=center] {Negatively\\sumptuous} (5,-1.5);
	\draw[decorate,decoration={brace,amplitude=5pt,mirror}]
		(7,-1.5) -- node[anchor=north,yshift=-2pt,align=center] {Positively\\sumptuous} (10,-1.5);
		
	\newcommand{\ytickwest}[2]{\draw ($(#1) - (.1,0)$) -- ($(#1) + (.1,0)$); \node[anchor=east] at ($(#1) - (.1,0)$) {#2};}
	\ytickwest{.45,.5}{$\ell_0$}
	\ytickwest{.6,1}{$\ell'_0$}
	
	\ytickwest{3.75,-.5}{$\ell_1$}
	\ytickwest{3.8,-1.15}{$\ell'_1$}
	
	\ytickwest{7.35,.35}{$\ell_2$}
	\ytickwest{7.65,1.15}{$\ell'_2$}
	
\end{tikzpicture}
\caption{Illustration of the approach to the proof of Lemma~\ref{lem:parity-games-weights:cost-upper-bound}.}%
\label{fig:cost-upper-bound}
\end{figure}
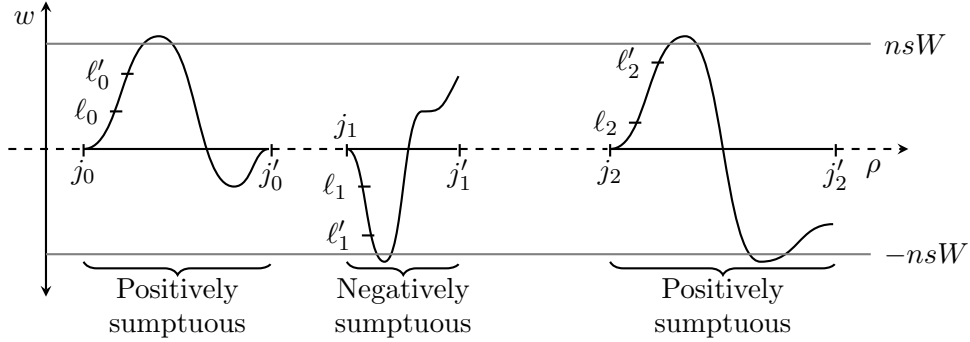

Let~$\sigma$ be a winning strategy for Player~$0$ in~$\game$ from~$\winreg_0(\game)$ of size at most~$s$.
Due to Lemma~\ref{lem:memory:player-0}, such a strategy exists.
Let~$\rho = v_0v_1v_2\cdots$ be a play that starts in~$\winreg_0(\game)$ and is consistent with~$\sigma$.
We call a position~$j \in \nats$ \textbf{sumptuous} if~$nsW < \Cor(\rho, j) < \infty$.
Each sumptuous position~$j$ has some odd color~$c$, and the request for~$c$ posed by visiting~$v_j$ is eventually answered.

Assume towards a contradiction that there exist infinitely many sumptuous positions.
We define a sequence of positions that begins with the first sumptuous position~$j_0$.
Let~$j'_0$ be the minimal position that satisfies~$\col(v_{j'_0}) \in \answer{\col(v_{j_0})}$.
This position exists since $\Cor(\rho, j_0) < \infty$ due to the definition of sumptuous positions.
We continue by defining~$j_1$ as the smallest sumptuous position greater than~$j'_0$ and by defining~$j'_1$ as the minimal position that satisfies~$\col(v_{j'_1}) \in \answer{\col(v_{j_1})}$.
Continuing in this manner, we obtain a sequence~$j_0 < j'_0 < j_1 < j'_1 < j_2 < j'_2 < \cdots$, where~$j_0$ is the first sumptuous position of~$\rho$, each~$j'_k$ for~$k \geq 0$ is the minimal position that satisfies both~$j'_k > j_k$ and $\col(v_{j'_k}) \in \answer{\col(v_{j_k})}$, and each~$j_k$ for~$k > 0$ is the smallest sumptuous position greater than~$j'_{k-1}$.
Since there exist infinitely many sumptuous positions by assumption and since each request posed at a sumptuous position is answered by definition, the sequence $j_0 < j'_0 < j_1 < j'_1 < j_2 < j'_2 < \cdots $ is indeed infinite.

Due to the definition of sumptuous positions and the~$j'_k$, we have~$\ampl(v_{j_k} \cdots v_{j'_k}) > nsW$ for each~$k \in \nats$.
Since~$\rho$ is consistent with the finite-state strategy~$\sigma$ of size~$s$, we claim that in each such~$v_{j_k} \cdots v_{j'_k}$ there exists an infix that can be repeated arbitrarily often while retaining consistency with~$\sigma$.
To identify such infixes, we partition the sumptuous positions~$j_k$:
We call a position~$j_k$ \textbf{positively sumptuous} if there exists a~$j'$ with~$j_k \leq j' \leq j'_k$ such that~$\weight(v_{j_k} \cdots v_{j'}) > nsW$ and \textbf{negatively sumptuous} otherwise.
In the latter case, there exists a~$j'$ with~$j_k \leq j' \leq j'_k$ such that~$\weight(v_{j_k} \cdots v_{j'}) < - nsW$.
 See Figure~\ref{fig:cost-upper-bound} for an illustration.
 In particular, note that the third sumptuous position is positively sumptuous, although it exceeds both the bounds~$nsW$ and~$-nsW$.

Let~$\sigma$ be implemented by~$(M, \init, \update)$.
As each edge contributes cost at most~$W$ to $\ampl(v_{j_k} \cdots v_{j'_k})$, this implies that there exist positions~$\ell_k$ and~$\ell'_k$ with~$j_k < \ell_k < \ell'_k < j'_k$ such that
\begin{itemize}
	\item $v_{\ell_k} = v_{\ell'_k}$,
	\item $\update^+(v_0 \cdots v_{\ell_k}) = \update^+(v_0 \cdots v_{\ell'_k})$,
	\item $\weight(v_{\ell_k} \cdots v_{\ell'_k-1}) > 0$, if~$j_k$ is positively sumptuous, and such that
	\item $\weight(v_{\ell_k} \cdots v_{\ell'_k-1}) < 0$, if~$j_k$ is negatively sumptuous.
\end{itemize}
The positions~$j_k$,~$\ell_k$,~$\ell'_k$, and~$j'_k$ split~$\rho$ into infixes~$\rho = \Pi_{k = 0,1,2,\dots} \pi_{k, \I} \cdot \pi_{k, \II} \cdot \pi_{k, \III} \cdot \pi_{k, \IV}$, where~$\pi_{k, \I}$,~$\pi_{k, \II}$,~$\pi_{k, \III}$, and~$\pi_{k, \IV}$ start at~$j_k$,~$\ell_k$,~$\ell'_k$, and~$j'_k$, respectively.
Due to the definition of~$\ell_k$ and~$\ell'_k$, the play $\rho' = \Pi_{k = 0,1,2,\dots} \pi_{k, \I} \cdot {(\pi_{k, \II})}^k \cdot \pi_{k, \III} \cdot \pi_{k, \IV}$ is consistent with~$\sigma$.
The costs-of-response of the requests opened by visiting the~$v_{j_k}$, however, diverge due to~$\card{\weight(\pi_{k, \II})} = \card{\weight(v_{\ell_k} \cdots v_{\ell'_k-1})} > 0$.
Hence,~$\rho'$ violates the parity condition with weights, which contradicts that~$\sigma$ is a winning strategy of Player~$0$.
\end{proof}

Having thus shown that Player~$0$ can indeed ensure a pseudo-polynomial upper bound on the incurred cost, we now proceed to show that this bound is tight.
A simple example shows that there exists a series of parity games with weights in which Player~$0$ wins from every vertex, but in which he cannot enforce a sub-pseudo-polynomial cost of any request.

\begin{lem}%
\label{lem:paritygames:weight:lower}
Let~$n,W \in \nats$.
There exists a parity game with weights~$\game_{n,W}$ with~$n$ vertices and largest absolute weight~$W$ such that for each vertex~$v \in \winreg_0(\game)$ and for each winning strategy for Player~$0$ from~$v$ there exists a play~$\rho$ starting in~$v$ and consistent with~$\sigma$ with~$\limsup_{j \rightarrow \infty}\Cor(\rho, j) \geq (n-1)W$.
\end{lem}

\begin{proof}
We show the game~$\game_{n,W}$ in Figure~\ref{fig:cost-lower-bound}.
The arena of~$\game_{n,W}$ is a cycle with~$n$ vertices of Player~$1$, where each edge has weight~$W$.
Moreover, one vertex is labeled with color two, its directly succeeding vertex is labeled with color one.
All remaining vertices have color zero.

\begin{figure}
\centering
\begin{tikzpicture}
	\node[p1] (v1) at (0,0) {\parnode{v_1}{1}};
	\node[p1] (v2) at (2,0) {\parnode{v_2}{0}};
	\node (dots) at (4,0) {$\cdots$};
	\node[p1] (vn-1) at (6,0) {\parnode{v_{n-1}}{0}};
	\node[p1] (vn) at (8,0) {\parnode{v_n}{2}};
	
	\path
		(v1) edge node[anchor=south] {$W$} (v2)
		(v2) edge[densely dashed] node[anchor=south] {$W$} (dots)
		(dots) edge[densely dashed] node[anchor=south] {$W$} (vn-1)
		(vn-1) edge node[anchor=south] {$W$} (vn);
	\path[draw,thick,-stealth,rounded corners]
		(vn.south) |- ($(dots) - (0,.8)$) -| (v1.south);
	\node[anchor=south] at ($(dots) - (0,.8)$) {$W$};
\end{tikzpicture}

\caption{The game~$\game_{n,W}$ witnessing a pseudo-polynomial lower bound on the cost that Player~$0$ can ensure.}%
\label{fig:cost-lower-bound}
\end{figure}
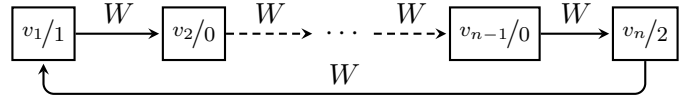

Player~$0$ only has a single strategy in this game and there exist only~$n$ plays in~$\game_{n,W}$, each starting in a different vertex of~$\game_{n,W}$.
In each play, each request for color one is only answered after~$n-1$ steps, each contributing a cost of~$W$.
Hence, this request incurs a cost of~$(n-1)W$.
Moreover, as this request is posed and answered infinitely often in each play, we obtain the desired result.
\end{proof}

While Player~$0$ is able to bound the costs from above by
\[ d(6n^2)(d+2)W(W+1) \in \bigo({(ndW)}^2) \enspace, \]
due to Lemma~\ref{lem:parity-games-weights:cost-upper-bound}, we only obtain a sequence of examples witnessing a lower bound of~$(n-1)W \in \Omega(nW)$ on these costs due to Lemma~\ref{lem:paritygames:weight:lower}.
Thus, while both the upper and the lower bound are pseudo-polynomial in the size of the game, there remains a polynomial gap between the two bounds.

The upper bound on the costs is due to the pumping argument presented in Lemma~\ref{lem:parity-games-weights:cost-upper-bound} leveraging the upper bound on the memory Player~$0$ requires to implement a winning strategy due to Lemma~\ref{lem:memory:player-0}.
We have a lower bound of~$(n-4)W \in \Omega(nW)$ on the memory required by Player~$0$ to implement a winning strategy due to Lemma~\ref{lem:memory:player-0}.
Thus, even if we are able to improve the upper bounds on the memory requirements for Player~$0$, the pumping argument leveraged in the proof of Lemma~\ref{lem:parity-games-weights:cost-upper-bound} can only improve the upper bound on the cost incurred by Player~$0$ to~$d(n-4)W^2 \in \bigo(ndW^2)$.
Hence, the proof methods used in this section do not enable us to completely close this gap between the upper and the lower bound on the cost Player~$0$ is able to enforce.

All bounds obtained in this section also hold for bounded parity games with weights:
The upper bound from Lemma~\ref{lem:parity-games-weights:cost-upper-bound} holds for bounded parity games with weights since the bounded parity condition with weights strengthens its unbounded variant.
Analogously, the games constructed for the proof of the lower bound on the incurred cost in Lemma~\ref{lem:paritygames:weight:lower} yield the same lower bound for bounded parity games with weights.

\section{From Energy Parity Games to (Bounded) Parity Games with Weights}%
\label{sec:ep2(b)pw}
We have discussed in Sections~\ref{sec:parity-weights-solve} and~\ref{sec:bounded-parity-weights-solve} how to solve parity games with weights via solving bounded parity games with weights and how to solve the latter games by solving energy parity games, both steps with a polynomial overhead.
An obvious question is whether one can also solve energy parity games by solving (bounded) parity games with weights.
In this section, we answer this question affirmatively:
We show how to transform an energy parity game into a bounded parity game with weights so that solving the latter also solves the former.
We then show how to transform a bounded parity game with weights into a parity game with weights with the same relation:
solving the latter also solves the former.
Both constructions here are gadget based and increase the size of the arenas only quadratically.
Hence, all three types of games are interreducible with at most polynomial overhead.

\subsection{From Energy Parity Games to Bounded Parity Games with Weights}%
\label{subsec:ep2(b)pw:ep2bpw}
In an energy parity game, Player~$0$ wins if the energy increases without a bound, as long as there is a lower bound.
In a bounded parity game, in contrast, he has to ensure both an upper and a lower bound.
Thus, we show in a first step how to modify an energy parity game so that Player~$0$ still has to ensure a lower bound on the energy, but can also \emph{throw away} unnecessary energy during each transition, thereby also ensuring an upper bound. The most interesting part of this construction is to determine when energy becomes unnecessary to ensure a lower bound. Here, we rely on Lemma~\ref{lem:boundeddecreaseinenergyparity}.

Formally, let $\game = (\arena, \col, \weight)$ be an energy parity game with $\arena = (V, V_0, V_1, E)$ where we assume w.l.o.g.\ that the minimal color in $\col(V)$ is strictly greater than $1$.
Now, we define $\game' = (\arena', \col', \weight')$ with $\arena = (V, V_0, V_1, E)$ where
\begin{itemize}
	\item $V' = V \cup E$, $V_0' = V_0 \cup E$, and $V_1' = V_1$,
	\item $E' = \set{ (v,e),(e,e),(e,v') \mid e = (v,v') \in E }$,
	\item $\col'(v) = \col(v)$ and $\col'(e) = 1$, and
	\item $\weight'(v,e) = \weight(e)$, $\weight'(e,e) = -1$, and $\weight(e,v') = 0$ for every $e = (v,v') \in E$.
\end{itemize}
Intuitively, every edge of $\arena$ is subdivided and a new vertex for Player~$0$ is added, where he can decrease the energy level. The negative weight also ensures that he eventually leaves this vertex in order to satisfy an energy condition.

We say that a strategy~$\sigma$ for Player~$0$ in $\arena'$ is corridor-winning for him from some $v \in V$, if there is a~$b \in \nats$ such that every play~$\rho $,
which starts in $v$ and is consistent with $\sigma$, satisfies the parity condition and $\ampl(\rho) \le b$. Hence, instead of just requiring a lower bound on the energy level as in the energy parity condition, we also require a uniform upper bound on the energy level (where we assume w.l.o.g.\ these bounds to coincide).

\begin{lem}%
\label{lem:ep:vs:cep}
Let $\game$ and $\game'$ be as above and let $v \in V$.
Player~$0$ has a winning strategy for~$\game$ from $v$ if and only if he has a corridor-winning strategy for $\game'$ from $v$.
\end{lem}

\begin{proof}
We first show the direction from left to right.
To this end, assume that Player~$0$ wins~$\game$ from~$v$.
Due to Proposition~\ref{prop:energyparity:equivalence}, he has a finite-state winning strategy~$\sigma$ for $\game$ from $v$, say of size~$s$.
Furthermore, there is an initial credit~$c_0$ such that every play prefix~$\pi$ that starts in $v$ and is consistent with $\sigma$ satisfies $\weight(\pi) \ge -c_0$. Finally, define $b = Wns$, where $n$ and $W$ denote the number of vertices and the largest absolute weight occurring in~$\game$, respectively.

We define a strategy~$\sigma'$ for Player~$0$ in $\game'$ such that it mimics the behavior of $\sigma$ and additionally ensures that the energy level of a play prefix never exceeds $b$ by more than~$W$. Formally, consider a play prefix~$\pi'$ in $\game'$ starting in $v$. If~$\pi'$ ends in some $v' \in V$, then we define $\sigma'(\pi) = (v', \sigma(f(\pi)))$ where $f \colon {(V\cup E)}^* \cup {(V\cup E)}^\omega \rightarrow V^* \cup V^\omega$  is the homomorphism induced by $f(v) = v$ and $f(e) = \epsilon$. On the other hand, assume $\pi'$ ends in some $e = (v',v'') \in E$. If $\weight(\pi') \le b$, then we define $\sigma'(\pi') = v''$, otherwise, we define $\sigma'(\pi') = e$, i.e., the self-loop is used until the energy level of the play prefix is exactly $b$. This completes the definition of $\sigma'$.

Now, consider a play~$\rho'$ in $\game'$ that starts in $v$ and is consistent with $\sigma'$. By definition of~$\sigma'$, $\rho'$ visits infinitely many vertices in $V$. Hence, by construction of $\arena$, $\rho = f(\rho')$ is a play in $\game$ that starts in $v$. Further, $\rho$ is consistent with $\sigma$, as $\sigma'$ mimics $\sigma$. Hence, $f(\rho')$ satisfies the parity condition. As the vertices removed from $\rho'$ all have color one, and as all colors in~$\col(V)$ are greater than one, we conclude that $\rho'$ satisfies the parity condition as well.

To conclude, we show that every prefix~$\pi'$ of $\rho'$ satisfies $-c_0 \le \weight'(\pi') \le b + W$. This implies that $\sigma'$ is indeed a corridor-winning strategy from $v$. The upper bound~$b+W$ is satisfied by construction of $\sigma'$: As soon as the weight exceeds $b$, it is decreased to $b$ by the strategy. As this correction happens after each transition, the bound~$b$ can be exceeded by at most $W$, the largest absolute weight of an edge.

To conclude, we consider two cases: first, assume $\rho'$ has no prefix whose energy level exceeds~$b$, then we have $\weight'(\pi') = \weight(f(\pi)) \ge -c_0$ for every prefix~$\pi'$ of $\rho'$.
Second, assume that~$\rho'$ has at least one prefix whose weight exceeds~$b$.
Let~$\pi$ be the shortest such prefix.
For every prefix~$\pi'$ shorter than~$\pi$ we obtain~$-c_0 \leq \weight'(\pi') \leq b+W$ via the above argument.
We show that every prefix longer than~$\pi$ has nonnegative weight, which concludes the proof.

Towards a contradiction, assume that there is a longer suffix with negative weight.
Then there is an infix of $\rho'$ of weight strictly smaller than $-b$, such that Player~$0$ never uses a self-loop in $\arena'$ to throw away energy.
Hence, $f(\rho)$ also has an infix with weight strictly smaller than~$-b$.
This, however, contradicts Lemma~\ref{lem:boundeddecreaseinenergyparity} and thus concludes the proof of this direction.

We now show the other direction of the lemma.
To this end, let $\sigma'$ be a corridor-winning strategy for Player~$0$ in $\game'$ from $v$. Further, let $f$ be defined as above.

We define a strategy~$\sigma$ for Player~$0$ from $v$ in $\game$ that is obtained by simulating play prefixes in~$\game'$. To this end, we again use a simulation function~$h$ that maps a play prefix~$v_0 \cdots v_j$ in~$\game$ that starts in $v$ and is consistent with $\sigma$ to a play prefix~$h(v_0 \cdots v_j)$ in $\game'$ that starts in $v$, is consistent with $\sigma'$, and ends in $v_j$.

Hence, we define $h(v) =v$. Now, assume we have a play prefix~$v_0 \cdots v_j$ in $\game$ that starts in $v$ and is consistent with $\sigma$. From our construction, we obtain a play prefix~$h(v_0 \cdots v_j)$ in $\game'$ that starts in $v$, is consistent with $\sigma'$, and ends in $v_j$. If $v_j \in V_0 \subseteq V_0'$, then let $\sigma'(h(v_0 \cdots v_j)) = (v_j,v_{j+1})$. We define $\sigma(v_0 \cdots v_j) = v_{j+1}$, which is a legal move due to the construction of $\arena'$.
If $v_j \in V_1$, then let $v_{j+1}$ be an arbitrary successor of $v_j$ in $\arena$.

In both cases, we have to define $h(v_0 \cdots v_j v_{j+1})$. As $\sigma'$ is a corridor-winning strategy for Player~$0$ from $v$ in $\game'$, there is a unique play of the form~$h(v_0 \cdots v_j) {(v_j, v_{j+1})}^m v_{j+1}$ that is consistent with $\sigma'$. We define $h(v_0 \cdots v_{j}v_{j+1})$ to be equal to this play, which satisfies the properties required above.

Let $b$ be the uniform bound on the amplitude of plays in $\game'$ consistent with $\sigma'$ starting in $v$. Now, fix a play $\rho$ in $\game$ starting in $v$ and consistent with $\sigma$. Furthermore, let $\rho'$ be the limit of the $h(\pi)$ for increasing prefixes of $\rho$. By construction, $\rho'$ starts in $v$ as well and is consistent with $\sigma'$. Hence, $\rho'$ visits infinitely many vertices from $V$ and never gets stuck in a self-loop throwing away energy. This implies $f(\rho') = \rho$. Furthermore, as $\rho'$ satisfies the parity condition, $\rho$ does as well: the colors removed by applying $f$ are inconsequential in this situation.

Let $\pi_j$ be the prefix of length~$j$ of $\rho$. A straightforward induction proves that the energy level of $\pi_j$ is greater or equal to that of $h(\pi_j)$. As the latter is bounded from below by $b$, we conclude that $\sigma$ is winning for Player~$0$ in $\game$ from $v$ with initial credit~$b$.
\end{proof}

Now, we turn $\game'$ into a bounded parity game with weights.
In such a game, the cost-of-response of every request has to be bounded, but the overall energy level of the play may still diverge to $-\infty$.
To rule this out, we open one unanswerable request at the beginning of each play, which has to be unanswered with finite cost in order to satisfy the bounded parity condition with weights.
If this is the case, then the energy level of the play is always in a bounded corridor, i.e., we obtain a corridor-winning strategy.

Formally, for every vertex~$v \in V$, we add a vertex~$\overline{v}$ to $\arena'$ of an odd color~$c^*$ that is larger than every color in $\col(V)$, i.e., the request can never be answered.
Furthermore,~$\overline{v}$ has a single outgoing edge to $v$ of weight~$0$, i.e., it is irrelevant whose turn it is.
Call the resulting arena~$\arena''$, the resulting coloring~$\col''$, and the resulting weighting~$\weight''$, and let $\game''=(\arena'', \bcp(\col'', \weight''))$.

\begin{lem}%
\label{lem:cep:vs:bpw}
Let $\game'$ and $\game''$ be as above and let $v \in V$.
 Player~$0$ has a corridor-winning strategy for $\game'$ from $v$
if and only if  $\overline{v} \in \winreg_0(\game'')$.
\end{lem}

\begin{proof}
We again first show the direction from right to left.
To this end, let $\sigma'$ be a corridor-winning strategy for Player~$0$ in $\game'$ from $v$.
Further, let~$b$ be the corresponding uniform bound on the amplitude of plays that start in $v$ and are consistent with $\sigma'$.
We define a strategy~$\sigma''$ for Player~$0$ from~$\overline{v}$ via $\sigma''(\overline{v} \pi) = \sigma'(\pi)$.

Let $\overline{v}\rho$ be a play that is consistent with $\sigma''$. By construction, $\rho$ starts in $v$ and is consistent with $\sigma'$. Hence, it satisfies the parity condition and its amplitude is bounded by~$b$. Thus, almost all requests in $\rho$ are answered with cost at most $b$ and there is no unanswered request of infinite cost. This implies that $\overline{v}\rho$ satisfies the bounded parity condition with weights.
Hence, $\overline{v} \in \winreg_0(\game'')$.

We now show the obverse direction of the lemma.
To this end, let $\sigma''$ and~$b$ be a winning strategy for Player~$0$ in $\game''$ from $\overline{v}$ and a bound such that every request in a play starting in $\overline{v}$ and consistent with $\sigma''$ is answered or unanswered with cost at most $b$. Due to Corollary~\ref{cor:bounded-parity-with-weights:uniform-bound}, such a strategy~$\sigma''$ exists. We define a strategy~$\sigma'$ for Player~$0$ from~$v$ in~$\game'$ via $\sigma'( \pi) = \sigma''(\overline{v}\pi)$.

Let $\rho$ be a play starting in $v$ that is consistent with $\sigma'$. By construction, $\overline{v}\rho$ is consistent with $\sigma''$. Hence, $\overline{v}\rho$ satisfies the parity condition and every request is answered or unanswered with cost at most $b$. In particular, this holds true for the unanswered request posed by visiting $\overline{v}$. Hence, the amplitude of $\overline{v}\rho$ (and thus also that of $\rho$) is bounded by $b$.

Thus, $\rho$ satisfies the parity condition and the energy level of all its prefixes is between~$-b$ and~$b$. As~$\rho$ is picked arbitrarily, we have that~$\sigma'$ is corridor-winning from $v$.
\end{proof}


\subsection{From Bounded Parity Games with Weights to Parity Games with Weights}%
\label{subsec:ep2(b)pw:bpw2pw}
Next, we show how to turn a bounded parity game with weights into a parity game with weights so that solving the latter also solves the former. The construction here uses the same restarting mechanism that underlies the proof of Lemma~\ref{lem:parity-to-bounded-parity:reduction}: as soon as a request has incurred a cost of $b$, restart the play and enforce a request of cost $b+1$, and so on.
Unlike the proof of Lemma~\ref{lem:parity-to-bounded-parity:reduction}, however, where Player~$1$ could restart the play at any vertex, here we always have to return to a fixed initial vertex we are interested in.
While resetting, we have to answer all requests in order to prevent Player~$1$ to use the reset to prevent requests from being answered.
Assume $v^* \in V$ is the initial vertex we are interested in.
Then, we subdivide every edge in to allow Player~$1$ to restart the play by answering all open requests and then moving back to $v^*$.

Formally, fix a bounded parity game with weights~$\game = (\arena, \bcp(\col, \weight))$ with $\arena = (V, V_0, V_1, E)$ and a vertex~$v^* \in V$. We define the parity game with weights~$\game_{v^*} = (\arena_{v^*}, \cp(\col_{v^*}, \weight_{v^*}))$ with $\arena_{v^*} = (V', V_0', V_1', E')$ where
\begin{itemize}
	\item $V' = V \cup E \cup \set{\top}$, $V_0' = V_0$, and $V_1' = V_1 \cup E  \cup \set{\top}$,
	\item $E' = \set{ (v,e), (e,\top), (e,v') \mid e = (v,v') \in E } \cup \set{(\top, v^*)}$,
	\item $\col_{v^*}(v) = \col(v)$, $\col_{v^*}(e) = 0$ for every~$e \in E$, and $\col_{v^*}(\top) = 2\max( \col(V))$, and
	\item $\weight_{v^*}(v,e) = \weight(e)$ for $(v,e) \in V \times E$ and $\weight_{v^*}(e') = 0$ for every other edge~$e' \in E'$.
\end{itemize}

\begin{lem}%
\label{lem:bpw:vs:pw}
Let $\game$ and $\game_{v^*}$ as above. Then, $v^* \in \winreg_0(\game)$ if and only if $v^* \in \winreg_0(\game_{v^*})$.
\end{lem}

\begin{proof}
 First, we show the direction from left to right. To this end, let $\sigma$ be a winning strategy for Player~$0$ for $\game$ from~$v^*$ such that there is a $b$, such that every request in a play that starts in $v^*$ and is consistent with $\sigma$ is answered or unanswered with cost at most $b$.
Due to Corollary~\ref{cor:bounded-parity-with-weights:uniform-bound}, such a strategy~$\sigma$ exists.

We define a winning strategy~$\sigma'$ for Player~$0$ from~$v^*$ in $\game'_{v^*}$ as follows:
Given a play or a play prefix~$\pi'$ in $\arena_{v^*}$ that does not end in vertex~$\top$, let $\suffix(\pi')$ be the longest suffix of $\pi'$ that does not contain~$\top$.
Hence, if $\pi'$ starts  in $v^*$, then $\suffix(\pi')$ starts in $v^*$ as well, as $v^*$ is the unique successor of~$\top$.
Further, let $f\colon {(V \cup E)}^* \cup {(V \cup E)}^\omega \rightarrow   V^* \cup V^\omega$ be the homomorphism induced by~$f(v) = v$ for $v \in V$ and $f(e) = \epsilon$ for $e \in E$.
Now, if~$\pi'$ is a play (prefix) in $\arena_{v^*}$ that does not visit~$\top$, then $f(\pi')$ is a play (prefix) in $\arena$ of the same weight that induces the same sequence of colors (save for the occurrences of the inconsequential minimal color zero at the vertices from~$E$ that are deleted by~$f$).

Let $\pi'$ be a play prefix in $\arena_{v^*}$ that ends in a vertex $v \in V_0' = V_0 $.
We define $\sigma'(\pi') = (v, \sigma(f(\suffix(\pi'))))$ and show that $\sigma'$ is winning for Player~$0$ in $\game_{v^*}$ from $v^*$.
To this end, let $\rho'$ be a play in $\arena_{v^*}$ starting in $v^*$ that is consistent with $\sigma'$.
We consider two cases, depending on whether or not the play~$\rho'$ visits the vertex~$\top$ infinitely often.

If~$\rho'$ visits~$\top$ only finitely often, then~$\suffix(\rho')$ is an infinite play in~$\game'_{v^*}$ starting in~$v^*$.
By definition of~$f$ and construction of~$\sigma'$, the infinite play~$f(\suffix(\rho'))$ in $\arena$ starts in~$v^*$ and is consistent with~$\sigma$.
Hence, $f(\suffix(\rho'))$ satisfies the bounded parity condition with weights.
Since this condition strengthens the parity condition with weights and since the latter condition is~$0$-extendable, we conclude that $f(\rho')$ satisfies the parity condition with weights as well.
This, in turn, implies that the complete play~$\rho'$ satisfies the parity condition with weights due to the construction of $\arena_{v^*}$

Now, assume that the play~$\rho'$ visits $\top$ infinitely often.
Then,~$\rho'$ is of the form $\pi_0' \top \pi_1' \top \pi_2' \top \cdots$, where none of the $\pi_j'$ visits~$\top$.
Hence, by definition of~$\arena'$,~$f$, and~$\sigma'$, each play prefix~$f(\pi_j')$ in~$\arena$ starts in~$v^*$ and is consistent with~$\sigma$.
Furthermore, every request in each~$\pi_j'$ is answered by the next visit to the vertex~$\top$ at the latest, i.e.,~$\rho'$ satisfies the parity condition.
Thus, it suffices to show that the cost-of-response of all requests in~$\rho'$ is bounded.
This follows immediately from the fact that $\sigma$ only admits answered or unanswered requests of cost at most~$b$ when starting in~$v^*$ and that each $f(\pi_j')$ starts in~$v^*$ and is consistent with~$\sigma$.
This property is inherited by the $\pi_j'$ due to the construction of $\arena_{v^*}$.
Thus,~$\rho'$ satisfies the parity condition with weights, i.e., $\sigma'$ is indeed winning for Player~$0$ from~$v^*$.

To show the direction from right to left, we proceed by contraposition. Due to the determinacy of both games, it suffices to show that $v^* \in \winreg_1(\game)$ implies  $v^* \in \winreg_1(\game_{v^*})$. Hence, let $\tau$ be a winning strategy for Player~$1$ in $\game$ from $v$. Further, let $\suffix$ and $f$ be defined as above.

Now, we define a strategy~$\tau'$ for Player~$1$ from $v^*$ in $\game_{v^*}$ that is controlled by a counter~$\kappa$, which is initialized with zero, and which is incremented during a play every time the costs of some request exceed~$\kappa$. We construct our strategy such that each time $\kappa$ is updated, Player~$1$ restarts the play by moving to $\top$ and then to $v^*$.

Assume we have a play prefix~$\pi'$ in $\arena_{v^*}$ that ends in a vertex of Player~$1$ and have to define $\tau'(\pi')$. We consider several cases depending on the last vertex of $\pi'$. If $\pi'$ ends in $\top$, then we define $\tau'(\pi')= v^*$, which is the only successor of $\top$.

If $\pi'$ ends in $v \in V_1 \subseteq V_1'$, then we define $\tau'(\pi') = (v, \tau(f(\suffix(\pi'))))$, i.e., we discard everything up to and including the last occurrence of $\top$. Finally, if $\pi'$ ends in $e = (v,v') \in E \subseteq V_1'$, then we consider two cases. Let $\kappa$ be the current counter value. If $\suffix(\pi')$ contains a request such that the remaining part of $\pi'$ that starts at this request has amplitude greater than $\kappa$, then we define $\tau'(\pi') = \top$ and increment $\kappa$. Otherwise, we define $\tau'(\pi') = v'$ and leave $\kappa$ unchanged.

It remains to show that $\tau'$ is winning in $\game_{v^*}$ from $v^*$. To this end, let $\rho'$ be a play in $\game_{v^*}$ that starts in ${v^*}$ and is consistent with $\tau'$. If $\rho'$ visits $\top$ infinitely often, then $\rho'$ contains, for every $b \in \nats$, a (different) request that is answered or unanswered with cost at least $b$. Hence, $\rho'$ violates the parity condition with costs.

Finally, if $\rho'$ visits $\top$ only finitely often, then there is a $b \in \nats$ (the final value of $\kappa$, which is incremented only finitely often in this case) such that every request in $\rho'$ is answered or unanswered with cost at most $b$. Furthermore, let $\rho$ be the suffix of $\rho'$ that starts after the last occurrence of $\top$. As in the previous case, $f(\rho)$ is a play in $\arena$ that starts in $v^*$ and is consistent with $\tau$. As $\rho$ and $f(\rho)$ have essentially the same evolution of the weights (save for the removed edges of weight zero) and the same color sequence (save for the removed vertices of color zero), every request in $f(\rho)$ is answered or unanswered with cost at most $b$. However, as $\rho$ is consistent with $\tau$, it violates the bounded parity condition with weights. This is, in this situation, only possible by violating the parity condition. Hence $\rho$, and thus also $\rho'$, violates the parity condition as well. Therefore, $\rho'$ in particular violates the parity condition with weights.

In both cases, $\rho'$ is winning for Player~$1$, i.e., $\tau'$ has the desired properties.
\end{proof}


\subsection{Relation to Mean-Payoff Parity Games}%
\label{subsec:bounded-parity-weights-solve:mppg}
Recall that we have shown in Section~\ref{sec:parity-weights-solve} and Section~\ref{sec:bounded-parity-weights-solve} how to solve parity games with weights by solving polynomially many energy parity games of polynomial size.
Subsequently, in Section~\ref{subsec:ep2(b)pw:ep2bpw} and Section~\ref{subsec:ep2(b)pw:bpw2pw} we have shown the converse direction, i.e., how to solve energy parity games by solving polynomially many parity games with weights, again of polynomial size.
In summary, we have shown the problem of solving parity games with weights and that of solving energy parity games with weights to be polynomial time equivalent.

Due to this strong connection between the two games, we moreover obtain a connection to another widely used class of games, so-called mean-payoff parity games as introduced by Chatterjee, Henzinger, and Jurdzi{\'n}ski~\cite{ChatterjeeHenzingerJurdzinski05}.
A mean-payoff parity game is played on a colored arena with weights.
It is the task of Player~$0$ to not only satisfy the parity condition induced by the coloring, but also to ensure that the average weight of the traversed edges is nonnegative.
Formally, given an arena with vertex set~$V$ and set of edges~$E$, a coloring~$\col$ of~$V$, and a weight function~$\weight$ over~$E$, the mean-payoff parity condition is defined as
\[
	\mpp(\col, \weight) =
	\set{v_0v_1v_2\cdots \in V^\omega \mid \liminf_{j \rightarrow \infty} \frac{1}{j} \weight(v_0 \cdots v_j) \geq 0 }
	\cap \parity(\col) \enspace .
\]
A game~$(\arena, \mpp(\col, \weight))$ is called a mean-payoff parity game.

Chatterjee and Doyen~\cite{ChatterjeeDoyen12} showed that the problem of solving energy parity games and that of solving mean-payoff parity games are logarithmic space equivalent.
\begin{propC}[\cite{ChatterjeeDoyen12}]%
\label{prop:mean-payoff-parity}
Let~$\game = (\arena, \mpp(\col, \weight))$ be a mean-payoff parity game with~$n$ vertices and let~$\game' = (\arena, \energyparity(\col,\weight'))$, where $\weight'(e) = \weight(e) + \frac{1}{1+n}$ for all edges~$e$ of~$\arena$.
Then~$\winreg_0(\game) = \winreg_0(\game')$.
\end{propC}

Hence, we can use the techniques underlying Theorem~\ref{thm:bounded-parity-with-weights:complexity} to reduce the problem of solving parity games with weights to the problem of solving mean-payoff parity games.
In fact, it is Proposition~\ref{prop:mean-payoff-parity} that underpins the proof of the second part of Proposition~\ref{prop:energyparity:complexity}, i.e., the proof that energy parity games can be solved in pseudo-quasi-polynomial time.

\begin{cor}
	The following decision problems are polynomial time equivalent:

	\begin{quotation}
	Given a parity game with weights~$\game$ and a vertex~$v$ of~$\game$, does Player~$0$ have a winning strategy from~$v$ in~$\game$?
\end{quotation}

\begin{quotation}
	Given a mean-payoff parity game~$\game$ and a vertex~$v$ of~$\game$, does Player~$0$ have a winning strategy from~$v$ in~$\game$?
\end{quotation}
\end{cor}

\noindent
Solving parity games with weights by iteratively solving mean-payoff parity games as sketched above, however, does not allow us to obtain the memory bounds from Theorem~\ref{thm:memory}.
This is due to Player~$0$, in general, requiring infinite memory in order to win a mean-payoff parity game~\cite{ChatterjeeHenzingerJurdzinski05}.


\section{The Threshold Problem}%
\label{sec:optimality}
In Section~\ref{sec:parity-weights-def}, we have defined parity games with weights as a generalization of parity games with costs.
Up to this point, we have only considered the problem of solving parity games with weights, i.e., deciding for a given parity game with weights~$\game$ and a vertex~$v$ of~$\game$, whether Player~$0$ has a strategy~$\sigma$ such that all plays~$\rho$ starting in~$v$ and consistent with~$\sigma$ satisfy~$\limsup_{j \rightarrow \infty}\Cor(\rho, j)  < \infty$.
We have shown this problem to be a member of~$\np\cap\conp$ and that it can be solved in pseudo-quasi-polynomial time.

While our algorithm for deciding the above decision problem yields a strategy witnessing the ability of Player~$0$ to eventually ensure a finite cost-of-response, it does not provide any guarantee on the ``quality'' of the strategy beyond the very general bound~$b \in \bigo({(ndW)}^2)$ due to Theorem~\ref{thm:costs}.
In particular, it may be the case that there exists a strategy~$\sigma'$ such that all plays~$\rho$ starting in~$v$ and consistent with~$\sigma$ satisfy~$\limsup_{j \rightarrow \infty}\Cor(\rho, j) \leq b' < b$.
Hence, in this section, we investigate the so-called threshold problem for parity game with weights, i.e., the problem to decide, given a parity game with weights~$\game$, a vertex~$v$ of~$\game$, and a bound~$b \in \nats$, whether there exists a strategy~$\sigma$ for Player~$0$ such that for all plays~$\rho$ starting in~$v$ and consistent with~$\sigma$ satisfy~$\limsup_{j \rightarrow \infty}\Cor(\rho, j) \leq b$.

Weinert and Zimmermann~\cite{WeinertZimmermann17} have shown that the threshold problem for finitary parity games as well as for parity games with costs is \pspace-complete and that exponential memory is both necessary and sufficient for both players to implement witnessing strategies.
In this work, we show that the threshold problem for parity games with weights is \exptime-complete and that the memory bounds remain unchanged in comparison to the case of parity games with costs.

In order to simplify notation for the remainder of this section, if~$\game$ is clear from the context, for a given strategy~$\sigma$ for Player~$0$ and a vertex~$v$ of~$\game$ we define
\[\Cost_v(\sigma) = \sup_\rho\limsup_{j \rightarrow \infty}\Cor(\rho, j)\enspace,\] where~$\rho$ ranges over all plays~$\rho$ of~$\game$ that start in~$v$ and that are consistent with~$\sigma$. Dually, for a strategy~$\tau$ for Player~$1$ and a vertex~$v$, we define $\Cost_v(\tau) = \inf_\rho\limsup_{j \rightarrow \infty}\Cor(\rho, j)$, where~$\rho$ ranges over all plays~$\rho$ of~$\game$ that  start in~$v$ and that are consistent with~$\tau$.

\begin{numrem}%
\label{rem:value}
Let $\sigma$ be a strategy for Player~$0$, let $\tau$ be a strategy for Player~$1$, and let $v$ be a vertex.
Then, $\Cost_v(\sigma) \ge \Cost_v(\tau)$.
\end{numrem}

The main theorem of this section settles the complexity of the threshold problem for parity games with weights.

\begin{thm}
The following decision problem is \exptime-complete:
	\begin{quotation}
		Given a parity game with weights~$\game$, some vertex~$v$ of~$\game$, and a bound~$b \in \nats$, does Player~$0$ have a strategy~$\sigma$ with~$\Cost_{v}(\sigma) \leq b$ in~$\game$?
	\end{quotation}
\end{thm}

\noindent
The remainder of this section is organized as follows:
First, in Section~\ref{subsec:optimality:exptime-membership} we show the threshold problem for parity games with weights to be in \exptime, before showing \exptime-hardness of the problem in Section~\ref{subsec:optimality:exptime-hardness}.
We conclude this section by showing tight bounds on the memory requirements of witnessing strategies for both players in Section~\ref{subsec:optimality:memory}.

\subsection{ExpTime-Membership}%
\label{subsec:optimality:exptime-membership}
In order to solve the threshold problem for parity games with weights, we follow the approach used by Weinert and Zimmermann~\cite{WeinertZimmermann17} to solve the same problem for parity games with costs.
Given a parity game with costs~$\game$ and a threshold~$b$, they construct a classical parity game~$\thrgame$ that satisfies the following property:
\begin{quotation}
Player~$0$ has a strategy of cost at most~$b$ from some vertex~$v$ in~$\game$ if and only if he has a winning strategy from some designated vertex~$v'$ in~$\thrgame$.
\end{quotation}

\noindent
We show how to lift the construction of Weinert and Zimmermann to the setting of parity games with weights and show that this construction yields \exptime-membership of the threshold problem for parity games with weights.
For the remainder of this section, we fix some parity game with weights~$\game$ with~$n$ vertices,~$d$ odd colors and a threshold~$b \in \nats$.

The idea behind the construction of~$\thrgame$ is to track, for each odd color~$c$ of~$\game$, an overapproximation of the costs incurred by all open requests for~$c$ in the current play prefix.
The authors then showed that, if Player~$1$ is able to violate the bound~$b$ at least~$n$ times, then she can do so infinitely often.
Thus, by additionally equipping~$\thrgame$ with an~$n$-bounded counter that counts the number of violations of the given threshold, we obtain the desired property given above.

Recall that, in parity games with costs, all weights are nonnegative.
Hence, it suffices to track an upper bound on the cost of open requests, as these costs are implicitly bounded from below by zero.
In the setting of parity games with weights, in contrast, we have to track both an upper as well as a lower bound on the cost of open requests.
To this end, we first define the set of intervals
\[
	I = \set{ (l, h) \mid -b \leq l \leq h \leq b } \enspace .
\]
Clearly, we have
\[
	\card{I} = (2b+1) + \cdots + 1 = (2b+1)(2b+2) / 2 = (4b^2 + 4b + 2b + 2) / 2 = 2b^2 + 3b + 1 \enspace .
\]
Using the set of intervals, we now define request functions that enable the overapproximation of the costs of open requests described above.
To this end, we denote the set of odd colors occurring in~$\game$ by~$D$.
A ($b$-bounded) request function $r \colon D \rightarrow \set{\bot} \cup I$ is a function mapping each odd color of~$\game$ either to
\begin{itemize}
	\item $\bot$, denoting that currently no request for color~$c$ is open, or to
	\item some~$(l, h) \in I$, denoting that
	\begin{itemize}
		\item there exists an open request for color~$c$ that has accumulated weight~$l$, that
		\item there exists an open request for color~$c$ that has accumulated weight~$h$, and that
		\item all requests for color~$c$ have accumulated weight at least~$l$ and at most~$h$.
	\end{itemize}
\end{itemize}

\noindent
Given some request function~$r$, we define the lower and upper residual request functions $r_\downarrow$ and~$r_\uparrow$ as~$r_\downarrow(c) = l$ and~$r_\uparrow(c) = h$, if~$r(c) = (l, h)$, and as~$r_\downarrow(c) = r_\uparrow(c) = \bot$ if~$r(c) = \bot$.
We write~$R$ to denote the set of all request functions.
We have~$\card{R} = {(\card{I} + 1)}^{\card{D}} = {(2b^2 + 3b + 2)}^d$, i.e., there exist exponentially many request functions when measured in the size of the game~$\game$, but only polynomially many when measured in the bound~$b$.

As stated above, Weinert and Zimmermann~\cite{WeinertZimmermann17} showed that it suffices for Player~$1$ to violate the threshold~$b$~$n$ times in order to witness that she can do so infinitely often.
Hence, we now define a memory structure comprising request functions and an ``overflow counter'' that, together with the game~$\game$, induces the desired parity game~$\thrgame$.

Recall that we fixed some parity game with weights~$\game$ with~$n$ vertices and~$d$ odd colors as well as a bound~$b \in \nats$.
Using the set~$R$ of request functions defined above, we define the set of memory states~$M = \set{0,\dots,n} \times R$.
As we aim to track the cost of open requests using the functions from~$R$, we define the initial memory element~$\init(v) = (0, r_v)$, where~$r_v$ is defined as
\[
	r_v(c) = \begin{cases}
 		(0,0) & \text{if $\col(v)$ is odd and $c = \col(v)$, and} \\
 		\bot & \text{otherwise.}
 	\end{cases}
\]

We define the update function~$\update\colon M \times E \rightarrow M$ implementing the above intuition.
Let~$m = (o, r) \in M$ and let~$e = (v, v') \in E$ with~$\weight(e) = w$.
This update function updates the memory state via~$\update(m, e) = (o', r')$ by performing the following steps in order:

\begin{description}
	\item[Weight] First, we resolve the effect of traversing the edge~$e$ with weight~$w$ by defining~$r'_\I$ as
\[
	r'_\I(c) = \begin{cases}
 		(r_\downarrow(c) + \weight(e), r_\uparrow(c) + \weight(e)) & \text{if~$r(c) \neq \bot$, and} \\
 		\bot & \text{otherwise.}
	 \end{cases}
\]

	\item[Overflow] In a second step, we check whether some request has violated the bound~$b$ during the move to~$v'$ and update the overflow counter if this is the case.
Thus, if there exists a color~$c$ such that either ${(r'_\I)}_\downarrow(c) < -b$ or~${(r'_\I)}_\uparrow(c) > b$, then we define~$r'_\II(c) = \bot$ for all~$c \in D$ and set~$o'$ to the minimum of~$o + 1$ and~$n$.
Otherwise, we define~$r'_\II = r'_\I$ and~$o' = o$.

	\item[Request] Finally, we resolve the effect of moving to the vertex~$v'$ with color~$\col(v')$ as follows:
If~$\col(v')$ is even, then we define
\[
	r'_\III(c) = \begin{cases}
 		\bot & \text{if $c \leq \col(v')$, and} \\
 		r'_\II(c) & \text{otherwise.}
 	\end{cases}
\]
If, however,~$\col(v')$ is odd, then we define
\[
	r'_\III(c) = \begin{cases}
 		(\min\set{{(r'_\II)}_\downarrow(\col(v')), 0}, \max\set{{(r'_\II)}_\uparrow(\col(v')), 0})
 			& \text{if $c = \col(v')$, and} \\
 		r'_\II(c) & \text{otherwise.}
 	\end{cases}
\]
\end{description}

\noindent
In either case, we define~$r' = r'_\III$, which concludes the definition of~$m' = (o', r')$.
The resulting~$o'$ is at most~$n$ and the resulting function~$r'$ is an element of~$R$.
We combine these elements in the memory structure~$\mem = (M, \init, \update)$.

Recall that throughout this section we fixed a parity game with weights~$\game$.
Let~$\game = (\arena, \cp(\col, \weight))$.
We define the ($b$-)threshold game of~$\game$ 
\[
	\thrgame = (\arena', \parity(\col')) \enspace ,
\]
with $\thrarena = (V', V'_0, V'_1, E')$, where $V' = V \times M$,~$V'_i = V_i \times M$ for~$i \in \set{0,1}$, where $((v, m), (v', m')) \in E'$ if and only if~$(v, v') \in E$ and~$\update(m, (v, v')) = m'$, and where
\[
	\col'(v, o, r) = \begin{cases}
	 	\col(v) & \text{if~$o < n$, and}\\
	 	1 & \text{otherwise,}
	\end{cases}
\]
which concludes the definition of~$\thrgame$.

Via a straightforward adaptation of results by Weinert and Zimmermann~\cite{WeinertZimmermann17} we obtain that this construction indeed satisfies the above property, i.e., that it suffices to solve~$\thrgame$ in order to solve the threshold problem for~$\game$:

\begin{lem}%
\label{lem:optimality:threshold-game-equivalence}
Let~$v^*$ be a vertex of~$\game$.
Player~$0$ has a strategy~$\sigma$ with~$\Cost_{v^*}(\sigma) \leq b$ if and only if he wins~$\thrgame$ from~$(v^*, \init(v^*))$.
\end{lem}

We split the proof of Lemma~\ref{lem:optimality:threshold-game-equivalence} into several lemmas.
The direction from right to left is relatively straightforward:
Since~$\thrgame$ is a parity game, if Player~$0$ wins~$\thrgame$ from~$(v^*, \init(v^*))$, then she does so with a positional strategy~$\sigma'$.
This strategy assigns to each vertex~$(v, o, r)$ of Player~$0$ in~$\thrgame$ a unique successor~$(v', o', r')$, where the values of~$o'$ and~$r'$ are deterministic updates of~$o$ and~$r$ via the update function~$\update$.
Hence,~$\sigma'$ can be interpreted as picking only a successor vertex~$v'$ of~$v$ with respect to the current memory state~$(o, r)$.
Thus, the choices of~$\sigma'$ can be mimicked in~$\game$.

\begin{lem}%
\label{lem:optimality:threshold-games:correctness:right-to-left}
If Player~$0$ wins~$\thrgame$ from~$(v^*, \init(v^*))$, then she has a strategy~$\sigma$ in~$\game$ with $\Cost_{v^*}(\sigma) \leq b$.
\end{lem}

\begin{proof}
Let~$\sigma'\colon V'_0 \rightarrow V'$ be a positional winning strategy for Player~$0$ from~$(v^*, \init(v^*))$ in~$\thrgame$.
Since~$\thrgame$ is a parity game, such a strategy exists by assumption.
We define the finite-state strategy~$\sigma$ for Player~$0$ in~$\game$ as the unique strategy induced by the memory structure~$\mem$ and the next-move function~$\nxt$ defined as~$\nxt(v,m) = v'$, if $\sigma'(v,m) = (v',m')$ for some $m'$.
It remains to show~$\Cost_{v^*}(\sigma) \leq b$.

Let $\rho = v_0 v_1 v_2 \cdots$ be a play in~$\game$ that begins in~$v^*$ and that is consistent with $\sigma$.
Let
\[
	\rho' = (v_0, o_0, r_0)(v_1, o_1, r_1)(v_2, o_2, r_2)\cdots
\]
be the unique play in $\thrgame$ that satisfies both~$(o_0, r_0) = \init(v_0)$ as well as $(o_j, r_j) = \update((o_{j-1}, r_{j-1}), (v_{j-1}, v_j))$ for all~$j > 0$.

A straightforward induction shows that~$\rho'$ is consistent with $\sigma'$.
As~$\rho'$ moreover starts in~$(v^*, \init(v^*))$ by definition, it is winning for Player~$0$, i.e., it satisfies the parity condition.
This in particular implies that the overflow counter along $\rho'$ never saturates, i.e., that we have~$o_j < n$ for all~$j \in \nats$, since vertices with saturated overflow counter form a losing sink for Player~$0$.
Hence, the plays~$\rho$ and~$\rho'$ coincide on their color sequences.
Since~$\rho'$ is winning for Player~$0$ in~$\game'$, it satisfies the parity condition, which in turn implies that~$\rho$ satisfies the parity condition.
It remains to show that almost all requests in~$\rho$ are answered with cost at most~$b$.

As argued above, the overflow counter along~$\rho'$ eventually stabilizes at some value less than~$n$.
Moreover, since~$\rho$ satisfies the parity condition, all but finitely many requests are answered.
Hence, there exists a position~$j$ such that $o_{j'} = o_j < n$ and such that~$\Cor(\rho, j') < \infty$ for every $j' >j$.
Assume towards a contradiction that there exists some~$j' \geq j$ with~$b < \Cor(\rho, j') < \infty$.
Then, due to the construction of~$\thrgame$, the overflow counter is incremented along the play~$\rho'$ at some point after~$j'$.
This, however, contradicts the choice of~$j$.
Hence, we obtain that almost all requests in~$\rho$ are answered with cost at most~$b$.
\end{proof}

We now turn our attention to the other direction of the statement of Lemma~\ref{lem:optimality:threshold-game-equivalence}, i.e., we aim to show that, if Player~$0$ has a strategy~$\sigma$ in~$\game$ with~$\Cost_{v^*}(\sigma) \leq b$, then she wins~$\thrgame$ from~$(v^*, \init(v^*))$.
We show this claim via contraposition:
Assume that Player~$0$ does not win~$\thrgame$ from~$(v^*, \init(v^*))$.
Since~$\thrgame$ is a parity game, it is determined. Hence, Player~$1$ wins~$\thrgame$ from~$(v^*, \init(v^*))$, say with the positional strategy~$\tau'$.
Such a positional strategy for him exists again due to~$\thrgame$ being a parity game.
We construct a strategy~$\tau$ for Player~$1$ in~$\game$ that enforces~$\limsup_{j \rightarrow \infty}\Cor(\rho, j) > b$ for each play~$\rho$ starting in~$v^*$ and consistent with~$\tau$.
Since this implies~$\Cost_{v^*}(\sigma) > b$ for each strategy~$\sigma$ of Player~$0$ (see Remark~\ref{rem:value}), this suffices to show the desired statement.

Recall that the overflow counter along each play starting in~$(v^*, \init(v^*))$ is monotonically increasing and bounded from above by the number~$n$ of vertices in~$\game$.
Hence, the value of the overflow counter either stabilizes at some value less than~$n$, or it eventually saturates at value~$n$.
In the former case,~$\tau'$ has to ensure that the resulting play violates the parity condition.
Hence, it suffices to mimic the moves made by~$\tau'$ in~$\game$ ad infinitum in this case, which results in a play with infinitely many unanswered requests in~$\game$.
In the latter case, however, mimicking~$\tau'$ does not yield a strategy in~$\game$ with the desired property, as~$\tau'$ does not necessarily prescribe ``meaningful'' moves in~$\thrgame$ once the overflow counter saturates.
This is due to the fact that these vertices form a sink component that is trivially winning for Player~$1$.
In order to leverage~$\tau'$ even after saturation of the overflow counter, we intervene whenever the overflow counter is incremented, by resetting it to the smallest possible value from which~$\tau'$ is still winning, thereby ensuring that the sink component is never reached.
Hence, the strategy always mimics ``meaningful'' moves in~$\thrgame$.

Formally, we define the set~$\reach$ that contains all vertices~$(v, o, r)$ that are visited by some play that starts in~$(v^*, \init(v^*))$ and that is consistent with~$\tau'$.
Recall that we defined~$r_v$ as the function denoting the requests opened by visiting the vertex~$v$.
Given a vertex~$v$, we then define~$o_v = \min(\set{n} \cup \set{o \mid (v,o,r_v) \in \reach})$.
In particular, we have $o_{v^*} = 0$, since~$(v^*, \init(v^*)) = (v^*, 0, r_v) \in \reach$.

\begin{lem}
The strategy~$\tau'$ is winning for Player~$1$ from~$(v,o_v,r_v)$ in $\thrgame$ for all~$v \in V$.
\end{lem}

\begin{proof}
If~$o_v = n$, then all plays starting in~$(v, o_v, r_v)$ violate the parity condition by construction of the arena of~$\thrgame$.
Thus, for the remainder of this proof, assume~$o_v < n$.
Let~$(v, o_v, r_v)\rho$ be a play starting in~$(v, o_v, r_v)$ that is consistent with~$\tau'$.
Moreover, let~$\pi$ be a play prefix starting in~$(v^*, \init(v^*))$, consistent with~$\tau'$, and ending in $(v, o_v, r_v)$.
Such a play prefix exists due to the definition of~$o_v$ and due to the assumption of~$o_v < n$, which implies~$(v, o_v, r_v) \in \reach$.

Since~$\tau'$ is positional, the play~$\pi\rho$ starts in~$(v^*, \init(v^*))$ and is consistent with~$\tau'$.
Thus,~$\pi\rho$ violates the parity condition, which implies that~$\rho$ violates the parity condition due to prefix-independence.
\end{proof}

We now define a new memory structure~$\mem'$ implementing the strategy~$\tau$.
Recall that we defined the memory structure $\mem = (M, \init, \update)$ during the construction of the threshold game~$\thrgame$.
Using the components of that memory structure, we define~$M' = M = \set{0,\ldots, n} \times R$,~$\init' = \init$, and
\[
\update'((o,r),(v,v')) = \begin{cases}
	(o,r') &\text{if $\update((o,r),(v,v')) = (o,r')$, and}\\
	(o_{v'},r') &\text{if $\update((o,r),(v,v')) = (o+1,r')$.}\\
\end{cases}
\]
and combine these elements into the memory structure~$\mem' = (M', \init', \update')$.

In the second case of the definition of~$\update'$, we have~$r' = r_{v'}$ by definition of~$\update$.
Finally, we define the next-move function~$\nxt'$ via $\nxt'(v,m) = v'$, if $\tau'(v,m) = (v',m')$ for some $m' \in M$ and let $\tau$ be the finite-state strategy implemented by $\mem'$ and~$\nxt'$.
We claim that for each play~$\rho$ starting in~$v^*$ that is consistent with~$\tau$ we have~$\limsup_{j \rightarrow \infty}\Cor(\rho, j) > b$.

To show this claim, let $\rho = v_0 v_1 v_2 \cdots $ be some play in~$\game$ that starts in~$v^*$ and that is consistent with $\tau$.
Moreover, let
\[
	\rho' = (v_0, o_0, r_0) (v_1, o_1, r_1) (v_2, o_2, r_2) \cdots
\]
be the unique play in $\arena \times \mem'$ that satisfies $(o_0, r_0) = \init'(v_0)$ as well as $(o_j, r_j) = \update'((o_{j-1}, r_{j-1}), (v_{j-1}, v_j))$ for all~$j > 0$.
We say that $j$ is a reset position if $j=0$ or if
\[
	\update((o_{j-1},r_{j-1}),(v_{j-1},v_{j})) = (o_{j-1}+1,r_{j}) \enspace ,
\]
i.e., if the second case in the definition of $\update'$ is applied.

The play~$\rho'$ is not necessarily a play in $\thrgame$, since~$\thrgame$ is defined with respect to~$\mem$ instead of~$\mem'$, but every infix of~$\rho'$ that starts at a reset position and does not contain another one, is a play infix in~$\thrgame$ that is consistent with~$\tau'$.
At every reset position, instead of incrementing the overflow counter, we set it to~$o_v$.

Via a straightforward induction leveraging the same arguments as Weinert and Zimmermann~\cite{WeinertZimmermann17} we obtain~$o_j < n$ for all~$j \in \nats$, i.e., the overflow counter along~$\rho'$ never saturates.
Intuitively, this implies that the strategy~$\tau$ always uses ``meaningful'' moves of~$\tau'$ for its choice of move and thus allows us to subsequently argue that~$\tau$ is indeed winning for Player~$1$.
Moreover, while the set~$M'$ of memory states as defined above is of size~$(n+1)\card{R}$, the above observation allows us to implement the strategy~$\tau'$ using a set of memory states of size~$n\card{R}$ by omitting the memory states indicating a saturated overflow counter.

It remains to show that we indeed have~$\limsup_{j \rightarrow \infty}\Cor(\rho, j) > b$.
As argued above, this then directly implies~$\Cost_{v^*}(\sigma) > b$ for each strategy~$\sigma$ of Player~$0$, concluding the proof of the direction from left to right of Lemma~\ref{lem:optimality:threshold-game-equivalence}.

\begin{proof}[Proof of Lemma~\ref{lem:optimality:threshold-game-equivalence}]
	The direction from right to left is encapsulated in Lemma~\ref{lem:optimality:threshold-games:correctness:right-to-left}.

	For the direction from left to right, recall that we defined a strategy~$\tau$ for Player~$1$, that we picked~$\rho$ beginning in~$v^*$ and consistent with~$\tau$ arbitrarily and that we defined~$\rho' = (v_0,o_0,r_0)(v_1,o_1,r_1)(v_2,o_2,r_2)\cdots$.
	First assume that the overflow counter of~$\rho'$ eventually stabilizes, i.e., there exists some~$j \in \nats$ such that~$o_{j'} = o_j$ for all~$j' > j$.
	Then, there exists a suffix of~$\rho'$ that is consistent with~$\tau'$, which therefore violates the parity condition.
	Hence, it suffices to note that the color sequences induced by~$\rho'$ and by~$\rho$ coincide in this case since the overflow counter along~$\rho'$ never saturates as argued above, and due to the construction of~$\thrgame$, in which vertices of the form $(v, o, r)$ inherit the coloring of the vertex~$v$ for~$o < n$.
	Thus,~$\rho$ violates the parity condition and, in turn, also the parity condition with weights with respect to any bound.

	Now assume that the overflow counter of~$\rho'$ does not stabilize.
	Then, there are infinitely many reset positions in~$\rho'$.
	We obtain that there exist a request that incurs cost greater than~$b$ between any two adjacent such positions as a direct consequence of the construction of the arena of~$\thrgame$.
	Hence, we obtain~$\limsup_{j \rightarrow \infty}\Cor(\rho, j) > b$, which concludes this direction of the proof.
\end{proof}

Lemma~\ref{lem:optimality:threshold-game-equivalence} allows us to decide the threshold problem for a parity game with weights with~$n$ vertices and~$d$ odd colors by solving a classical parity game of size~$\bigo(nb^d)$, i.e., the size of the parity game depends on both the size of the parity game with weights as well as the threshold.
Lemma~\ref{lem:parity-games-weights:cost-upper-bound}, however, allows us to bound the parameter~$b$:
If~$b$ is at least $nd(6n)(d+1)(W+1)W$, then solving the threshold problem reduces to solving the parity game with weights.
We formalize these observations as Algorithm~\ref{alg:optimality:parity-with-weights}.

\begin{algorithm}
\begin{algorithmic}[1]
\REQUIRE{Parity game with weights~$\game$ with~$n$ vertices,~$d$ odd colors, and largest absolute weight~$W$, vertex~$v$ of~$\game$, bound~$b \in \nats$.}
 \IF{$b \geq nd(6n)(d+1)(W+1)W$}%
 \label{line:branch}
	\RETURN{$v \in \winreg_0(\game)$}%
		\COMMENT{Requires solving~$\game$, e.g.\ via Algorithm~\ref{algorithm_fixpoint}}%
		\label{line:classical-solver}
 \ELSE%
 	\STATE{$\thrgame = \text{$b$-threshold game of~$\game$}$}%
 		\COMMENT{$\thrgame$ is explicitly constructed}%
 		\label{line:construction}
 	\RETURN{$(v, 0, r_v) \in \winreg_0(\thrgame)$}%
 		\COMMENT{Requires solving~$\thrgame$}%
 		\label{line:threshold-solving}
 \ENDIF%
\end{algorithmic}
\caption{An algorithm deciding whether or not Player~$0$ has a strategy~$\sigma$ with $\Cost_v(\sigma) \leq b$ in a parity game with weights~$\game$.}%
\label{alg:optimality:parity-with-weights}
\end{algorithm}
Recall that our aim in this section is to argue that the threshold problem for parity games with weights is in \exptime.
To this end, we have constructed Algorithm~\ref{alg:optimality:parity-with-weights}, the correctness of which follows directly from Lemma~\ref{lem:optimality:threshold-game-equivalence}.
It remains to argue that this algorithm runs in exponential time.

To this end, recall that parity games can be solved in polynomial time in the number of vertices and in exponential time in the number of odd colors~\cite{Jurdzinski/98/UP,CaludeJKLS/17/qp}.
Moreover, recall that the ($b$-)threshold game~$\thrgame$ has~$\bigo(nb^d)$ many vertices, where~$n$ and~$d$ are the number of vertices and the number of odd colors in~$\game$, respectively. 
As we are able to bound the value of~$b$ from above by a polynomial in~$n$,~$d$, and the largest absolute weight~$W$ of~$\game$ due to Lemma~\ref{lem:parity-games-weights:cost-upper-bound}, we obtain that we are able to solve~$\thrgame$ in exponential time, which we formalize in the following theorem.

\begin{thm}%
\label{thm:optimality:exptime-membership}
The following problem is in \exptime:
\begin{quotation}
	Given a parity game with weights~$\game$, a vertex~$v$ of~$\game$, and a bound~$b \in \nats$, does Player~$0$ have a strategy~$\sigma$ with~$\Cost_v(\sigma) \leq b$?
\end{quotation}
\end{thm}

\begin{proof}
We show that Algorithm~\ref{alg:optimality:parity-with-weights} witnesses the claimed membership in \exptime.
The correctness of this algorithm follows directly from Lemma~\ref{lem:parity-games-weights:cost-upper-bound} and Lemma~\ref{lem:optimality:threshold-game-equivalence}.
It remains to show that Algorithm~\ref{alg:optimality:parity-with-weights} terminates in exponential time.
Let~$n$ be the number of vertices of~$\game$, let~$d$ be the number of odd colors of~$\game$, and let~$W$ be the largest absolute weight in~$\game$.

If~$b \geq nd(6n)(d+1)(W+1)W$, then the dominating factor for the runtime of Algorithm~\ref{alg:optimality:parity-with-weights} is the call to a solver for parity games with weights in Line~\ref{line:classical-solver}.
This solver only has to solve the given parity game with weights as discussed in Section~\ref{sec:parity-weights-solve}.
Due to Theorem~\ref{thm:bounded-parity-with-weights:complexity}, the problem of solving these games is in~$\np \cap \conp$.
Since~$\np \subseteq \exptime$, we obtain the desired runtime in this case.

If, however,~$b < nd(6n)(d+2)(W+1)W$, Algorithm~\ref{alg:optimality:parity-with-weights} constructs and solves the~$b$-threshold game~$\thrgame$ of~$\game$ in Line~\ref{line:construction} and Line~\ref{line:threshold-solving}, respectively.
Let~$n'$ be the number of vertices of~$\thrgame$.
By construction of~$\thrgame$ we obtain
\[
	n' = n \card{\set{0,\dots, n} \times R} = n (n+1){(2b^2 + 3b + 2)}^d \in \bigo(n^2b^{2d}) \enspace ,
\]
where~$R$ denotes the set of request functions as defined above.

Due to our case analysis based on the cardinality of~$b$ in Line~\ref{line:branch}, we furthermore obtain~$b \in \bigo(n^2 d^2 W)$, which in turn implies
\[
	n' \in \bigo(n^2{(n^2 d^2 W)}^{2d}) = \bigo(n^{2+4d}d^{4d}W^{2d}) \enspace .
\]

As we assume weights to be given in binary encoding, we additionally obtain~$W \in \bigo(2^{\card{\game}})$, which finally implies
\[
	n' \in \bigo(n^{2+4d}d^{4d}{(2^{\card{\game}})}^{2d}) =
	n^{2+4d}d^{4d}{2^{2d\card{\game}}} \enspace,
\]
i.e.,~$\thrgame$ contains only exponentially many vertices and~$d' = \max\set{1, d} \in \bigo(d)$ many colors in terms of~$\card{\game}$.
Recall that parity games can be solved in polynomial time in the number of vertices and in exponential time in the logarithm of the number of colors using, e.g., the recent algorithm by Calude et al.~\cite{CaludeJKLS/17/qp}.
Hence,~$\thrgame$ can indeed be solved in exponential time in~$\card{\game}$, which implies membership of the above problem in~\exptime.
\end{proof}

Algorithm~\ref{alg:optimality:parity-with-weights} furthermore yields an algorithm determining the optimal~$b$ such that Player~$0$ has a strategy of cost at most~$b$ from a given~$v$:
Given a parity game with weights~$\game$ with~$n$ vertices,~$d$ odd colors, and largest absolute weight~$W$, and a vertex~$v$ of~$\game$, we first solve~$\game$ and determine whether or not~$v \in \winreg_0(\game)$.
If this is not the case, then no such bound~$b$ exists.
Otherwise, the minimal~$b$ with the above property can be determined with a binary search over the range~$0,\dots,nd(6n)(d+1)(W+1)W$.
This binary search requires deciding at most~$\log(nd(6n)(d+1)(W+1)W)$, i.e., polynomially many instances of the threshold problem, each of which takes at most exponential time.
Hence, the optimal~$b$ such that Player~$0$ has a strategy of cost at most~$b$ can be determined in exponential time in the size of~$\game$.

This concludes the proof of \exptime-membership of the threshold problem for parity games with weights.
It remains to show that this bound is tight.
To this end, we show in the next section that the threshold problem is \exptime-hard via a reduction from the problem of solving countdown games.

\subsection{ExpTime-Hardness}%
\label{subsec:optimality:exptime-hardness}
In the previous section we have shown \exptime-membership of the threshold problem for parity games with weights.

In this section we provide a matching lower bound by showing that the threshold problem is \exptime-complete.
To this end, we reduce the $\exptime$-hard problem of solving countdown games to the threshold problem for parity games with weights.
Countdown games were introduced by Jurdzi{\'n}ski, Laroussinie, and Sproston~\cite{JurdzinskiSprostonLaroussinie08}.

In a countdown game, some initial credit is fixed at the beginning of a play.
Both players then move in alternation in an arena whose edges are labeled with negative weights.
In each turn, first Player~$0$ announces some weight, before Player~$1$ has to move along some edge of that weight, reducing the initial credit by the weight of the traversed edge.
If the credit at some point hits zero, Player~$0$ wins.
If, however, the credit at some point turns negative, Player~$1$ wins.
Since each edge has strictly negative weight, each play is won by either player after finitely many moves.

When formulated in our framework of arenas and winning conditions, a countdown game~$\game = (\arena, \countdown(\weight, c^*))$ consists of an arena~$\arena = (V, V_0, V_1, E)$, a weighting~$\weight$ of~$\arena$, and some initial credit~$c^* \in \nats$, which satisfies the following conditions:
\begin{enumerate}
	\item There exists a designated sink vertex~$v_\bot \in V_1$,
	\item we have
		\begin{itemize}
			\item $E \subseteq (V_0 \times V_1) \cup (V_1 \times V_0) \cup (\set{v_\bot} \times \set{v_\bot})$,
			\item $V_0 \times \set{v_\bot} \subseteq E$, and
			\item $(v_\bot, v_\bot) \in E$,
		\end{itemize}
	\item for each~$e_1 = (v, v'_1), e_2 = (v, v'_2) \in E$, with $v \in V_0$ we have that $e_1 \neq e_2$ implies $\weight(e_1) \neq \weight(e_2)$,
	\item for each~$e \in E \cap (V_0 \times V)$ we have
		$\weight(e) \begin{cases}
			= 0 & \text{if $e \in (V_0 \times \set{v_\bot})$, and} \\
 			< 0 & \text{otherwise, and}
		\end{cases} $
	\item for each~$e \in E \cap (V_1 \times V)$ we have~$\weight(e) = 0$.
\end{enumerate}

\noindent
As discussed above, a countdown game is played in turns.
Each turn starts at a vertex~$v$ of Player~$0$, from where Player~$0$ first picks some outgoing edge~$e$ leading to vertex~$v'$ of Player~$1$.
That edge has unique weight among the outgoing edges of~$v$ due to the third requirement.
Moreover, if Player~$0$ does not opt to end the play by moving to~$v_\bot$, the weight of the edge is negative due to the fourth condition.
Subsequently, Player~$1$ picks a successor of~$v'$ and moves to that successor along an edge of weight zero due to the fifth requirement, where the next turn of the play starts.

The countdown condition is defined as
\[
	\countdown(\weight, c^*) =
		\set{ \rho = v_0v_1v_2\cdots \in V^\omega \mid \exists j.\, v_j = v_\bot \text{ and } c^* + \weight(\rho) = 0} \enspace .
\]
Our definition of countdown games differs from the one given by Jurdzi{\'n}ski, La\-roussi\-nie, and Sproston~\cite{JurdzinskiSprostonLaroussinie08}, as we adapted it to fit our framework of games introduced in Section~\ref{sec:preliminaries}.
It is, however, easy to see that our definition and the one given by the authors are equivalent.

As countdown games are essentially of finite duration, we obtain that they are determined due to Zermelo~\cite{Zermelo13}.

\begin{rem}%
\label{rem:optimality:exptime-hardness:countdown-games-determinacy}
Countdown games are determined.
\end{rem}

Jurdzi\'nski, Laroussinie, and Sproston showed that solving countdown games is \exptime-hard via a reduction from the word problem for alternating Turing machines with polynomially bounded space.
In order to concisely encode the exponential number of configurations attainable by the Turing machine during its run on the input word and the transitions between these configurations, this reduction requires the weights along the edges of the countdown game as well as the initial credit~$c^*$ to be given in binary encoding.

\begin{propC}[\cite{JurdzinskiSprostonLaroussinie08}]%
\label{prop:optimality:exptime-hardness:countdown-games-hardness}
	The following problem is \exptime-hard:
	\begin{quotation}
		Let~$\game$ be a countdown game and let~$v$ be a vertex of~$\game$.
		Does Player~$0$ have a winning strategy from~$v$ in~$\game$?
	\end{quotation}
\end{propC}

\noindent
We reduce the problem of solving countdown games to the threshold problem for parity games with weights.
To this end, for the remainder of this section, fix some countdown game~$\game = (\arena, \countdown(\weight, c^*))$ where~$\arena = (V, V_0, V_1, E)$, as well as some initial vertex~$v^* \in V$.

Intuitively, we construct the parity game with weights~$\game'$ such that at the beginning of the play a single request is opened, which is only answered upon visiting~$v_\bot$.
After visiting~$v_\bot$, the play returns to the initial vertex~$v^*$ reopening the unique request of~$\game'$ along the way.

In a countdown game, only Player~$0$ may decide to move to~$v_\bot$.
Since he should, intuitively, only do so after traversing a play prefix of weight~$-c^*$, we equip the edges leading from his vertices to~$v_\bot$ with weight~$2c^*$.
Thus, he can enforce ``tallying the score'' by moving to~$v_\bot$.

In order to afford Player~$1$ the same possibility, we add edges that allow her to move from his vertices to~$v_\bot$.
Furthermore, in order to incentivize her to only take these edges once she has exceeded the lower bound of~$-c^*$, these edges have weight zero.
All remaining weights remain unchanged, thus the costs incurred by the unique request in~$\game'$ model the remaining credit in the corresponding play in~$\game$.

Formally, let~$v_\top$ be some vertex not occurring in~$V$.
We define the parity game with weights~$\game' = (\arena', \cp(\col, \weight'))$, where~$\arena' = (V', V'_0, V'_1, E')$, with
\begin{itemize}
	\item $V' = V \cup \set{v_\top}$, $V'_0 = V_0 \cup \set{v_\top}$, $V'_1 = V_1$, and
	\item $E' = (E \setminus \set{(v_\bot, v_\bot)}) \cup (V_1 \times \set{v_\bot}) \cup \set{(v_\bot, v_\top), (v_\top, v^*)}$.
\end{itemize}
Since there exists a unique outgoing edge of~$v_\bot$ leading back to the initial vertex~$v^*$ of the countdown game via~$v_\top$, the play is restarted after visiting the sink vertex.

Furthermore, we define the weight function
\[
	\weight'(e) = \begin{cases}
		\weight(e) & \text{ if $(v, v') \in E \setminus (V_0 \times \set{v_\bot})$} \\
		2c^* & \text{ if $e \in V_0 \times \set{v_\bot}$} \\
		0 & \text{ otherwise}
	\end{cases}
\]
as well as the coloring
\[
	\col'(v) = \begin{cases}
		1 & \text{ if $ v = v_\top$}, \\
		2 & \text{ if $ v = v_\bot$}, \\
		0 & \text{ otherwise}
	\end{cases}
\]
and claim that Player~$0$ has a strategy~$\sigma$ with~$\Cost_{v_\top}(\sigma) \leq c^*$ if and only if he wins~$\game$ from~$v^*$.
We illustrate this construction in Figure~\ref{fig:optimality:exptime-completeness:countdown2parity}.

\begin{figure}[htb]\centering
	\begin{tikzpicture}
	\node[p0] (top) at (0,0) {\parnode{v_\top}{1}};
	\node[p0] (init) at (3,0) {\parnode{v^*}{0}};
	\node[p0] (pre-bot-0) at (7,.75) {\parnode{}{0}};
	\node[p1] (pre-bot-1) at (7,-.75) {\parnode{}{0}};
	\node[p0] (bot) at (9,0) {\parnode{v_\bot}{2}};
	\node at ($(init) ! .35 ! (bot)$) {$\cdots$};
	
	\node[draw, rounded corners, fit={(init) (pre-bot-0) (pre-bot-1) (bot)},label=above:{$\arena$}, inner ysep=.1cm, inner xsep=.5cm] (box) {};
	
	\path (pre-bot-0) edge[<-] +(-1, 0) edge[<-] +(-1, -.5);
	\path (pre-bot-1) edge[<-] +(-1, 0) edge[] +(-1, .5);
	
	\path
		(top) edge[weight = $0$ at .5 anchor south] (init)
		(init)
			edge +(1.5,.75)
			edge +(1.5,0)
			edge +(1.5, -.75);
	\path (pre-bot-0) edge[weight = $2c^*$ at .5 anchor south] (bot);
	\path (pre-bot-1) edge[weight = $0$ at .5 anchor north] (bot);
	\path[draw, rounded corners, ->, weight=0 at .5 anchor north] 
		(bot) |- ($(box.south) - (0,.25)$) -| (top);
\end{tikzpicture}
	\caption{Construction of the parity game with weights~$\game'$ from a given countdown game~$\game$. We omit unchanged weights for the sake of readability.}%
	\label{fig:optimality:exptime-completeness:countdown2parity}
\end{figure}
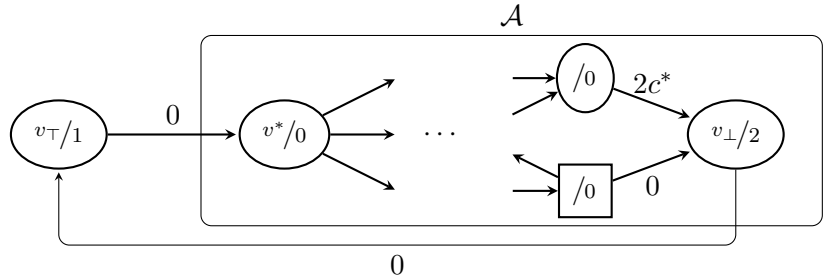

We claim that this construction implements our intuition, i.e., that solving~$\game'$ with respect to the bound~$c^*$ is indeed equivalent to solving~$\game$.

\begin{lem}%
\label{lem:optimality:exptime-hardness:correctness}
Player~$0$ wins~$\game$ from~$v^*$ if and only if he has a strategy~$\sigma$ with~$\Cost_{v_\top}(\sigma) \leq c^*$ in~$\game'$.
\end{lem}

\begin{proof}
We first show the direction from left to right.
To this end, let~$\sigma$ be a winning strategy for Player~$0$ in~$\game$ from~$v^*$.
Moreover, let~$\pi = v_0 \cdots v_j$ be a play prefix in~$\game'$ starting in~$v_\top$ and let~$j'$ be the largest position in~$\pi$ with~$v_{j'} = v_\top$.
We define the strategy~$\sigma'$ for Player~$0$ in~$\game'$ via~$\sigma'(\pi) = \sigma(v_{j' + 1} \cdots v_j)$ and claim~$\Cost_{v_\top}(\sigma') \leq b$.

To prove this claim, let~$\rho' = v_0v_1v_2\cdots$ be a play in~$\game'$ starting in~$v_\top$ that is consistent with~$\sigma'$.
First assume towards a contradiction that~$\rho'$ only visits~$v_\top$ finitely often.
Then, due to the structure of the arena,~$\rho'$ only visits~$v_\bot$ finitely often.
By construction of~$\sigma'$, this implies that~$\rho'$ contains a suffix that begins in~$v^*$, is consistent with~$\sigma$, but never visits~$v_\bot$.
This contradicts~$\sigma$ being a winning strategy for Player~$0$ in~$\game$ from~$v^*$.
Hence,~$\rho'$ visits~$v_\bot$ infinitely often.

Thus,~$\rho'$ is of the form
\[ \rho' = v_\top\pi_0v_\bot \cdot v_\top\pi_1v_\bot \cdot v_\top\pi_2v_\bot  \cdots \enspace , \]
where each~$\pi_j$ starts in~$v^*$ and is consistent with~$\sigma$.
We first argue that, if~$\pi_j$ ends in a vertex of Player~$1$, then we have $0 \geq \weight'(\pi_j) \geq -c^*$:
All weights in~$\game'$ except for those along the edges from~$V_0 \times \set{v_\bot}$ are nonpositive.
Hence, $0 \geq \weight'(\pi_j)$ follows directly from the construction of~$\game'$.
Moreover, $\weight'(\pi_j) < -c^*$ would contradict~$\pi_j$ being consistent with the winning strategy~$\sigma$ for Player~$0$ in~$\game$, since Player~$0$ would be unable to continue the play prefix~$\pi_j$ such that the resulting play is winning for her.
Hence, we have $0 \geq \weight'(\pi_j) \geq -c^*$.
Moreover, since all edges leading from~$v_\top$ and all edges leading to~$v_\bot$ have weight zero, and since~$\weight'(\pi)$ is decreasing for increasing prefixes~$\pi$ of~$\pi_j$ due to construction of~$\game'$, we obtain $\ampl(v_\top \pi_j v_\bot) \le c^*$.

If, however, $\pi_j$ ends in a vertex of Player~$0$, then we have~$\weight'(\pi_j) = -c^*$---and therefore $\ampl(v_\top \pi_j v_\bot) = c^*$---as~$\pi_j$ is consistent with the winning strategy~$\sigma$ for Player~$0$ in~$\game$.
In either case, we obtain that the unique request in~$v_\top \pi_j v_\bot$ posed by visiting~$v_\top$ is answered with cost at most~$c^*$.
Hence,~$\rho'$ has cost at most~$c^*$, which concludes this direction of the proof.

We show the other direction of the statement via contraposition:
Assume Player~$0$ does not win~$\game$ from~$v^*$.
Since~$\game$ is determined due to Remark~\ref{rem:optimality:exptime-hardness:countdown-games-determinacy}, Player~$1$ wins~$\game$ from~$v^*$, say with strategy~$\tau$.
We define a strategy~$\tau'$ for Player~$1$ in~$\game'$ that is winning for her from~$v_\top$ via mimicking moves made by~$\tau$ until the initial credit is used up.
At that point, we define~$\tau'$ to prescribe moving to~$v_\bot$ in order to witness exceeding the initial credit and to restart the play.

Formally, let~$\pi' = v_0\cdots v_j$ be a play prefix in~$\game'$ that starts in~$v_\top$ and ends in some vertex of Player~$1$.
Moreover, let~$j' \leq j$ be the largest position such that $v_{j'} = v_\top$.
If~$\weight(v_{j'} \cdots v_j) \geq -c^*$, we define~$\tau'(\pi') = \tau(v_{j' + 1} \cdots v_j)$.
Otherwise, we define $\tau'(\pi') = v_\bot$.
It remains to show $\Cost_{v_\top}(\tau') > c^*$, which concludes the proof due to Remark~\ref{rem:value}.

To this end, let~$\rho'$ be a play in~$\game'$ starting in~$v_\top$ consistent with~$\tau'$.
Due to the structure of~$\arena'$, every infix of~$\rho'$ that visits neither~$v_\top$ nor~$v_\bot$ traverses edges of weight zero and negative weight in alternation.
Moreover, since~$\tau'$ prescribes moving to~$v_\bot$ once the play infix since the last visit to~$v_\top$ has incurred weight exceeding~$-c^*$, the play~$\rho'$ is of the form
\[ \rho' = v_\top\pi_0v_\bot \cdot v_\top\pi_1v_\bot \cdot v_\top\pi_2v_\bot  \cdots \enspace , \]
where each~$\pi_j$ starts in~$v^*$ and is consistent with~$\tau$.

We aim to show~$\ampl(v_\top\pi_{j}v_\bot) > c^*$ for all~$j$, which implies $\limsup_{k \rightarrow \infty}(\rho', k) > c^*$ due to the construction of~$\game'$ and thus suffices to show the desired statement.
To this end, fix some~$j \in \nats$ and first consider the case that~$\pi_j$ ends in a vertex of Player~$1$.
Then we obtain~$\weight(\pi_j) < -c^*$ by definition of~$\tau'$.
This directly implies the desired statement.

Now consider the case that~$\pi_j$ ends in a vertex of Player~$0$.
We first argue that $\weight'(\pi_j) \neq -c^*$ holds true.
Towards a contradiction, assume $\weight'(\pi_j) = -c^*$ and recall that~$\pi_j$ starts in~$v^*$ and is consistent with the winning strategy~$\tau$ for Player~$1$ from~$v^*$.
Thus the play~$\pi_j{(v_\bot)}^\omega$ in~$\game$ is consistent with~$\tau$.
Hence, $\weight'(\pi_j) = -c^*$ contradicts~$\tau$ being a winning strategy for Player~$1$ from~$v^*$.

It remains to show $\ampl(v_\top\pi_{j}v_\bot) > c^*$ for the case that~$\pi_j$ ends in a vertex of Player~$0$.
If $\weight'(\pi_j) < -c^*$, we directly obtain~$\ampl(v_\top\pi_{j}v_\bot) > c^*$.
If, however, $\weight'(\pi_j) > -c^*$, we have $\weight'(\pi_j  v_\bot) > c^*$, since we defined $\weight'( v, v_\bot ) = 2c^*$ for each vertex~$v$ of Player~$0$.
Thus, we obtain $\ampl(v_\top\pi_{j}v_\bot) > c^*$ for each infix $v_\top\pi_{j}v_\bot$ of~$\rho'$.
Hence, each request posed by visiting~$v_\top$ is answered with cost greater than~$c^*$.
Since we argued above that~$\rho'$ contains infinitely many visits to~$v_\top$, we obtain~$\limsup_{k \rightarrow \infty}(\rho', k) > c^*$.
\end{proof}

Due to Lemma~\ref{lem:optimality:exptime-hardness:correctness} we obtain a polynomial reduction from the problem of solving countdown games to the threshold problem for parity games with weights.
As the former problem is \exptime-hard due to Proposition~\ref{prop:optimality:exptime-hardness:countdown-games-hardness}, this implies \exptime-hardness of the latter problem.

\begin{lem}%
\label{lem:optimality:exptime-hardness}
	The following decision problem is \exptime-hard:
	\begin{quotation}
		Given a parity game with weights~$\game$, some vertex~$v^*$ of~$\game$, and a bound~$b \in \nats$, does Player~$0$ have a strategy~$\sigma$ with~$\Cost_{v^*}(\sigma) \leq b$ in~$\game$?
	\end{quotation}
\end{lem}

\begin{proof}
We reduce the problem of solving countdown games to the given problem.
To this end, let~$\game = (\arena, \countdown(\weight, c^*))$ be a countdown game and let~$v^*$ be a vertex of~$\arena$.
We construct the parity game with weights~$\game'$ as described above.
Due to Lemma~\ref{lem:optimality:exptime-hardness:correctness}, Player~$0$ wins~$\game$ from~$v^*$ if and only if he has a strategy~$\sigma$ with $\Cost_{v_\top} (\sigma)\leq c^*$.
As the problem of solving countdown games is known to be \exptime-hard due to Proposition~\ref{prop:optimality:exptime-hardness:countdown-games-hardness}, this implies the desired result.
\end{proof}

In the following section, we consider the memory requirements of both players when playing optimally.


\subsection{Memory Requirements}%
\label{subsec:optimality:memory}
Recall that, if Player~$0$ just aims to win a parity game with weights with~$n$ vertices,~$d$ odd colors, and largest absolute weight~$W$,
then a memory structure of size polynomial in~$n$,~$d$, and~$W$ suffices to implement a winning strategy due to Theorem~\ref{thm:memory}.
Dually, Player~$1$ requires, in general, infinite memory in order to implement a strategy winning for her, again due to Theorem~\ref{thm:memory}.

In this section, we show that these bounds change significantly in the context of the threshold problem: 
In order to implement a strategy that enforces cost of at most~$b$ in a parity game with weights with~$d$ odd colors,
memory of size polynomial in~$b$ and exponential in~$d$ is both necessary and sufficient for Player~$0$.
Dually, if Player~$1$ just aims to enforce a cost of the resulting play larger than some threshold~$b$,
strategies of size exponential in~$d$ suffice for her to do so.
Both of these bounds are tight.

We first argue that exponential memory indeed suffices for both players to satisfy or violate a given threshold in a parity game with weights, respectively.
To this end, recall that, given a parity game with weights~$\game$ and a threshold~$b$, we determined the solution of the threshold problem by solving the~$b$-threshold game~$\game'$ of~$\game$.
This threshold game is a classical parity game whose arena consists of the arena of~$\game$ augmented with request functions and an overflow counter that is bounded from above by~$n$.

Furthermore recall that, in the proof of Lemma~\ref{lem:optimality:threshold-game-equivalence}, we showed how to leverage a winning strategy for either player in~$\thrgame$ in order to construct a winning strategy for them in~$\game$.
To this end, we implemented the strategy for Player~$0$ in~$\game$ using the set~$\set{0,\dots,n} \times R$ as memory states, where the set~$\set{0,\dots,n}$ implements the overflow counter and where~$R$ denotes the set of request functions.

The first component of that memory structure is, however, irrelevant for Player~$0$:
If he has a winning strategy whose behavior is dependent not only on the current vertex and the request function, but also on the value of the overflow counter, then he also has one that only depends on the current vertex and the request function.

\begin{lem}%
\label{lem:optimality:memory:upper-bound:player0}
Let $\game$ be a parity game with weights containing~$d$ odd colors and let~$b \in \nats$.
Moreover, let~$v^*$ be a vertex of~$\game$.
If Player~$0$ has a strategy $\sigma$ in~$\game$ with $\Cost_{v^*}(\sigma) = b$, then he also has a strategy~$\sigma'$ with~$\Cost_{v^*}(\sigma') \leq b$ and $\card{\sigma'} = {(2b^2 + b + 2)}^d$.
\end{lem}

\begin{proof}
Recall that we argued previously that if Player~$0$ has a strategy~$\sigma$ in~$\game$ with~$\Cost_{v^*}(\sigma) \leq b$, then he also has a winning strategy from~$(v^*, \init(v^*))$ in the threshold game~$\thrgame$ as defined in Section~\ref{subsec:optimality:exptime-membership}.
Moreover, since~$\thrgame$ is a parity game, we obtain that if Player~$0$ wins~$\thrgame$ from~$(v^*, \init(v^*))$, then he also has a positional winning strategy doing so.
Thus, let~$\sigma_b$ be a positional winning strategy for Player~$0$ from~$(v^*, \init(v^*))$ in~$\thrgame$.

Let~$\reach$ be the set of vertices reached by plays starting in~$(v^*, \init(v^*))$ and consistent with~$\sigma_b$.
Since~$\sigma_b$ is positional, and since the parity condition is prefix-independent,~$\sigma_b$ is winning from all vertices in~$\reach$.
Furthermore, for each vertex~$v$ and each request function~$r$, we define
\[
	o_{v,r} = \max(\set{0} \cup \set{o \mid (v, o, r) \in \reach}) \enspace ,
\]
i.e.,~$o_{v,r}$ is the maximal value such that Player~$0$ wins~$\thrgame$ from~$(v, o_{v, r}, r)$ using~$\sigma_b$, or zero, if no such value exists.

We now define the strategy~$\sigma'$ for Player~$0$ in~$\game$ such that it has the above properties.
To this end, recall that we defined the memory structure~$\mem = (M, \init, \update)$ for the construction of~$\thrgame$, where~$M = \set{0,\dots, n} \times R$, and where~$R$ is the set of request functions.
We define~$M' = R$, the update function~$\update'(r, (v, v')) = r'$, if~$\update((o_{v, r}, r), (v, v')) = (o', r')$, as well as the initialization function~$\init'(v) = (o_{v, r_v}, r_v)$, if $\init(v) = (0, r_v)$.
Finally, we define the next-move function~$\nxt'(v, r) = v'$, where~$v'$ is the unique vertex that satisfies $\sigma_b(v, o_{v, r}, r) = (v', o', r')$, and claim that the strategy~$\sigma'$ implemented by~$\mem' = (M', \init', \update')$ and~$\nxt'$ has $\Cost_{v^*}(\sigma') \leq b$, which suffices to show the desired statement.

To prove this claim, let~$\rho = v_0v_1v_2\cdots$ be a play starting in~$v^*$ and consistent with~$\sigma'$ and let~$(v_0, r_0)(v_1, r_1)(v_2, r_2)\cdots$ be the unique play defined via~$r_0 = \init'(v_0)$ and~$r_j = \update'(r_{j-1}, (v_j, v_{j+1}))$ for all~$j > 0$.

A straightforward induction yields $(v_j, o_{v_j, r_j}, r_j) \in \reach$ for all~$j \in \nats$.
We first argue that we have~$o_{v_j, r_j} \leq o_{v_{j+1}, r_{j+1}}$ for all~$j \in \nats$.
To this end, let $(o, r)  = \update((o_{v_j, r_j}, r_j), (v_j, v_{j+1}))$.
By construction of the arena of the threshold game we have~$o \geq o_{v_j, r_j}$ and~$r = r_{j+1}$.
Moreover, since~$(v_j, o_{v_j, r_j}, r_j) \in \reach$ and due to our definition of~$\sigma'$, we obtain $(v_{j+1}, o, r) = (v_{j+1}, o, r_{j+1}) \in \reach$.
Hence,~$o \leq o_{v_{j+1}, r_{j+1}}$, which implies $o_{v_{j+1}, r_{j+1}} \geq o_{v_j, r_j}$.

Thus, the~$o_{v_j, r_j}$ are monotonically increasing.
Furthermore, we easily obtain~$o_{v_j, r_j} < n$ due to all~$(v_j, o_{v_j, r_j}, r_j)$ being in~$\reach$, the definition of~$\reach$, and due to~$\sigma_b$ being winning for Player~$0$ from~$(v^*, \init(v^*))$.
Hence, the sequence of the $o_{v_j, r_j}$ eventually stabilizes, i.e., there exists a~$j \in \nats$ such that~$o_{v_{j'}, r_{j'}} = o_{v_j, r_j}$ for all~$j' \geq j$.

We argue that the play $(v_j, o_{v_j, r_j}, r_j)(v_{j+1}, o_{v_{j+1}, r_{j+1}}, r_{j+1})(v_{j+2}, o_{v_{j+2}, r_{j+2}}, r_{j+2}) \cdots$ is consistent with~$\sigma_b$:
Let~$j' \geq j$ be such that $(v_{j'}, o_{v_{j'}, r_{j'}}, r_{j'}) \in V'_0$ and let $\sigma_b(v_{j'}, o_{v_{j'}, r_{j'}}, r_{j'}) = (v_{j'+1}, o, r_{j'+1})$.
We then clearly obtain~$o \leq o_{v_{j'+1}, r_{j'+1}}$ by definition of the latter.
Furthermore, we have~$o_{v_{j'}, r_{j'}} \leq o$ due to the construction of the arena of the threshold game, which yields $o = o_{v_{j'+1}, r_{j'+1}}$ due to our assumption $o_{v_{j'}, r_{j'}} = o_{v_{j'+1}, r_{j'+1}}$.

Thus, the play $(v_j, o_{v_j, r_j}, r_j)(v_{j+1}, o_{v_{j+1}, r_{j+1}}, r_{j+1})(v_{j+2}, o_{v_{j+2}, r_{j+2}}, r_{j+2}) \cdots$ starts in a vertex from~$\reach$, is consistent with~$\sigma_b$, and shares a color sequence with a suffix of~$\rho$ due to $o_j \leq o_{v_j, r_j} < n$.
The strategy $\sigma_b$ being winning for Player~$0$ from~$(v^*, \init(v^*))$, the construction of~$\thrgame$ and prefix-independence of the parity condition with weights then yield~$\weight(\rho) \leq b$.
\end{proof}

For Player~$1$, in contrast, it is open whether one can omit the overflow counter when implementing a strategy with cost at least~$b$.
Hence, we have to include it in the resulting memory structure.
Recall, however, that we have argued in Section~\ref{subsec:optimality:exptime-membership} that we are able to omit those memory states modeling a saturated overflow counter, thus slightly reducing the size of the resulting strategy in comparison to a naive implementation.
The following upper bound thus results directly from the results of Section~\ref{subsec:optimality:exptime-membership}.

\begin{cor}%
\label{cor:optimality:memory:upper-bound:player1}
Let $\game$ be a parity game with weights with~$n$ vertices and~$d$ odd colors and let~$b \in \nats$.
Moreover, let~$v^*$ be a vertex of~$\game$.
If Player~$1$ has a strategy~$\tau$ in~$\game$ with~$\Cost_{v^*}(\tau) = b$, then she also has a strategy~$\tau'$ with~$\Cost_{v^*}(\tau') \geq b$ and~$\card{\tau'} = n {(2b^2 + 3b + 2)}^d$.
\end{cor}

Having argued that exponential memory suffices for both players to implement optimal strategies, we now turn our attention to providing matching lower bounds.
These exponential lower bounds are inherited from the special case of finitary parity games, for which Weinert and Zimmermann~\cite{WeinertZimmermann17} showed that both players require exponential memory in order to implement strategies that ensure or violate a given threshold.
We reprint these results here for the sake of completeness.

\begin{propC}[\cite{WeinertZimmermann17}]\hfill
\begin{enumerate}
\item For every $d \geq 1$ there exists a finitary parity game~$\game_{d}$ with a vertex~$v^*$ such that
	\begin{itemize}

		\item $\game_{d}$ has $d$ odd colors and $\card{\game_{d}} \in \bigo(d^2)$,

		\item Player~$0$ has a strategy~$\sigma$ in~$\game_{d}$ with~$\Cost_{v^*}(\sigma) = d^2 + 2d$,

		\item there exists no strategy~$\sigma'$ for Player~$0$ with~$\Cost_{v^*}(\sigma') < d^2 + 2d$, and

		\item for every strategy~$\sigma$ for Player~$0$ in~$\game_{d}$, $\Cost_{v^*}(\sigma) = d^2 + 2d$ implies $\card{\sigma} \geq 2^{d-1}$.

	\end{itemize}

\item For every $d \geq 1$ there exists a finitary parity game~$\game_d$ with a vertex~$v^*$ such that
	\begin{itemize}

		\item $\game_d$ has $\bigo(d)$ many vertices and $2d$ odd colors,

		\item Player~$1$ has a strategy~$\tau$ in~$\game_d$ with~$\Cost_{v^*}(\tau) = 5(d-1) + 7$,

		\item there exists no strategy~$\tau'$ for Player~$1$ with~$\Cost_{v^*}(\tau') > 5(d-1) + 7$, and

		\item every strategy~$\tau$ for Player~$0$ in~$\game_d$ with~$\Cost_{v^*}(\tau) = 5(d-1) + 7$ has size at least~$2^d$.

	\end{itemize}
\end{enumerate}
\end{propC}


\section{Conclusions and Future Work}%
\label{sec:conclusion}
We have established that parity games with weights and bounded parity games fall into the same complexity class as energy parity games.
This is interesting, because, while solving such games has the signature complexity class~$\np \cap \conp$, they are not yet considered a class in their own right.
It is also interesting because their properties appear to be inherently different:
While they both combine the qualitative parity condition with quantified costs, parity games with weights \emph{combine} these aspects on the property level, whereas energy parity games simply look at the combined---and totally unrelated---properties.
We show the characteristic properties of parity games and of games with combinations of a parity condition with quantitative conditions relevant for this work in Table~\ref{tab:characteristics}.

\begin{table}[h]
\footnotesize
\centering
\begin{tabular}{lccccc}\toprule
	& Complexity & Mem.\ Pl.~$0$/Pl.~$1$ &  Bounds \\
	\midrule
	Parity Games~\cite{CaludeJKLS/17/qp} & quasi-poly. & pos./pos. & -- \\ 
	Energy Parity Games~\cite{ChatterjeeDoyen12,DaviaudJurdzinskiLazic18} &  pseudo-quasi-poly. & $\bigo(ndW)$/pos. & $\bigo(nW)$ \\ \midrule
	Finitary Parity Games~\cite{ChatterjeeHenzingerHorn09} & poly. & pos./inf. & $\bigo(nW)$ \\
	Parity Games with Costs~\cite{FijalkowZimmermann14,MogaveroMS15} &  quasi-poly. & pos./inf. & $\bigo(nW)$ \\
	Parity Games with Weights &  pseudo-quasi-poly. & $\bigo(nd^2W)$/inf. & $\bigo({(ndW)}^2)$ \\ \bottomrule
\end{tabular}
\caption{Characteristic properties of variants of parity games.}%
\label{tab:characteristics}
\end{table}

As future work, we are looking into the natural extensions of parity games with weights to Streett games with weights~\cite{ChatterjeeHenzingerHorn09,FijalkowZimmermann14}, and at the complexity of determining optimal bounds and strategies that obtain them~\cite{WeinertZimmermann17}.
We are also looking at variations of the problem.
The two natural variations are
\begin{itemize}
 \item to use a one-sided definition (instead of the absolute value) for the amplitude of a play, i.e., using $\ampl(\pi) = \sup_{j < \card{\pi}} \weight(v_0 \cdots v_j) \in \nats_\infty$ (instead of $\ampl(\pi) = \sup_{j < \card{\pi}} \abs{\weight(v_0 \cdots v_j)} \in \nats_\infty$), and
 \item to use an arbitrary consecutive subsequence of a play, using the definition $\ampl(\pi) = \sup_{j \leq k < \card{\pi}} \abs{\weight(v_j \cdots v_k)} \in \nats_\infty$.
\end{itemize}
There are good arguments in favor and against using these individual variations---and their combination to $\ampl(\pi) = \sup_{j \leq k < \card{\pi}} \weight(v_j \cdots v_k) \in \nats_\infty$---but we feel that the introduction of parity games with weights benefit from choosing one of the four combinations as \emph{the} parity games with weights.

We expect the complexity to rise when changing from maximizing over the absolute value to maximizing over the value, as this appears to be close to pushdown boundedness games~\cite{ChatterjeeF/13/boundednessGames}, and we conjecture this problem to be \pspace-complete.


\bibliographystyle{alpha}
\bibliography{literature}

\end{document}